\documentclass[11pt]{article}
\usepackage[utf8]{inputenc}
\usepackage{fullpage}
\usepackage[table,xcdraw]{xcolor}
\usepackage{amsmath, amsthm, amssymb}
\usepackage{algpseudocode,algorithm,algorithmicx}
\usepackage{algorithmicx}
\usepackage{mathtools}
\usepackage[numbers,sort&compress]{natbib}
\usepackage{xfrac}
\usepackage{hyperref}
\usepackage{multirow}
\usepackage{caption}
\usepackage{bm}
\usepackage{newfloat}
\usepackage{enumitem}
\usepackage{xfrac}
\usepackage{bbm}
\usepackage{soul}
\usepackage{makecell}
\usepackage{wrapfig}
\usepackage{hhline}
\usepackage{graphicx}
\usepackage{subcaption}
\usepackage[normalem]{ulem}
\usepackage{tikz}
\usepackage{pgfplots}
\pgfplotsset{compat=1.16}

\newtheorem{theorem}{Theorem}[section]
\newtheorem{corollary}[theorem]{Corollary}
\newtheorem{lemma}[theorem]{Lemma}

\newtheorem{claim}[theorem]{Claim}
\newtheorem{definition}[theorem]{Definition}

\newtheorem{observation}[theorem]{Observation}

\newtheorem{result}{Result}

\numberwithin{equation}{section}

\renewcommand{\tilde}{\widetilde}
\renewcommand{\bar}{\overline}

\newcommand{\indicator}{\mathbbm 1}
\newcommand{\ind}{k}

\newcommand{\opt}{\textbf{OPT}}

\newcommand{\greedy}{\mathcal{A}_{BPB}}

\DeclareMathOperator{\poly}{poly}

\DeclareMathOperator{\spn}{\textbf{Span}}

\DeclareMathOperator{\OPT}{OPT}

\DeclareMathOperator*{\argmin}{arg\,min}

\newcommand{\rank}{\textbf{Rank}}

\let\vec\mathbf

\def\min{\qopname\relax n{min}}
\def\max{\qopname\relax n{max}}
\def\argmin{\qopname\relax n{argmin}}
\def\argmax{\qopname\relax n{argmax}}

\def\Pr{\qopname\relax n{\mathbf{Pr}}}

\def\I{\mathcal{I}}

\def\M{\mathcal{M}}

\def\R{\mathbb{R}}

\newcommand{\cpv}{{\normalfont\textsc{Cpv}}}
\newcommand{\optbpb}{{\normalfont\textsc{OptBpb}}}

\newcommand{\vecc}{\vec{c}}
\newcommand{\vecv}{\vec{v}}
\newcommand{\vecp}{\vec{p}}

\def\eps{\varepsilon}

\renewcommand{\vec}{\mathbf}
\newcommand{\eat}[1]{}
\newcommand{\SE}{S_{\vec {\bar p}}}

\newcommand{\barp}{\vec {\bar p}}

\newenvironment{lp*}{\begin{equation*}  \begin{array}{lll}}{\end{array}\end{equation*}}

\newcommand{\instance}{\langle M, \vec v, \vec c, \mathcal I, B\rangle }

\algtext*{EndWhile}% Remove "end while" text
\algtext*{EndIf}% Remove "end if" text
\algtext*{EndFor}% Remove "end for" text
\newcommand{\rejectedCoords}[1]{
\addplot[only marks, color=red, mark=x, mark size=3pt] coordinates {#1}
}

\newcommand{\acceptedCoords}[1]{
\addplot[only marks, color=blue] coordinates {#1}
}

\newcommand{\uninspectedCoords}[1]{
\addplot[only marks, mark=text, text mark=?, color=blue] coordinates {#1}
}

\title{Equilibria and Learning in Modular Marketplaces}
\author{
    Kshipra Bhawalkar \\ Google \\ \texttt{kshipra@google.com} \and
    Jeff Dean \\ Google \\ \texttt{jeff@google.com} \and
    Christopher Liaw \\ Google \\ \texttt{cvliaw@google.com} \and
    Aranyak Mehta \\ Google \\ \texttt{aranyak@google.com} \and
    Neel Patel \\ University of Southern California \\ \texttt{neelbpat@usc.edu}
}
\date{\today}

\newcommand{\plotWeightedMatroidIterationOne}[1]{
\begin{tikzpicture}[scale=#1]
  \begin{axis}[
    xlabel={$v$ (value)},
    ylabel={$p$ (price)},
    title={Budget 15, step 1},
    xmin=0, xmax=6,
    ymin=0, ymax=11,
    axis x line*=bottom,
    axis y line*=left,
  ]
    \acceptedCoords{
      (3, 1) 
    };

    \uninspectedCoords{
      (4, 2)
      (1, 1)
      (5, 6)
      (1, 5)
    };

    \node[above] at (axis cs:3,1) {$1$};
    \node[below] at (axis cs:3,1) {$(3, 1)$};

    \node[above] at (axis cs:4,2) {$2$};
    \node[below] at (axis cs:4,2) {$(4, 2)$};

    \node[above] at (axis cs:1,1) {$3$};
    \node[below] at (axis cs:1,1) {$(1, 1)$};

    \node[above] at (axis cs:5,6) {$4$};
    \node[below] at (axis cs:5,6) {$(5, 6)$};
    
    \node[above] at (axis cs:1,5) {$5$};
    \node[below] at (axis cs:1,5) {$(1, 5)$};

    \addplot[dashed, domain=0:6] {x / 3};
  \end{axis}
\end{tikzpicture}
}

\newcommand{\plotWeightedMatroidIterationTwo}[1]{
\begin{tikzpicture}[scale=#1]
  \begin{axis}[
    xlabel={$v$ (value)},
    ylabel={$p$ (price)},
    title={Budget 15, step 2},
    xmin=0, xmax=6,
    ymin=0, ymax=11,
    axis x line*=bottom,
    axis y line*=left,
  ]
    \acceptedCoords{
      (3, 1.5)
      (4, 2)
    };

    \uninspectedCoords{
      (1, 1)
      (5, 6)
      (1, 5)
    };

    \node[above] at (axis cs:3,1.5) {$1$};
    \node[below] at (axis cs:3,1.5) {$(3, 1.5)$};

    \node[above] at (axis cs:4,2) {$2$};
    \node[below] at (axis cs:4,2) {$(4, 2)$};

    \node[above] at (axis cs:1,1) {$3$};
    \node[below] at (axis cs:1,1) {$(1, 1)$};

    \node[above] at (axis cs:5,6) {$4$};
    \node[below] at (axis cs:5,6) {$(5, 6)$};
    
    \node[above] at (axis cs:1,5) {$5$};
    \node[below] at (axis cs:1,5) {$(1, 5)$};

    \addplot[dashed, domain=0:6] {x / 2};
  \end{axis}
\end{tikzpicture}
}

\newcommand{\plotWeightedMatroidIterationThree}[1]{
\begin{tikzpicture}[scale=#1]
  \begin{axis}[
    xlabel={$v$ (value)},
    ylabel={$p$ (price)},
    title={Budget 15, step 3},
    xmin=0, xmax=6,
    ymin=0, ymax=11,
    axis x line*=bottom,
    axis y line*=left,
  ]
    \acceptedCoords{
      (3, 3)
      (4, 4)
      (1, 1)
    };

    \rejectedCoords{
      (5, 6)
    };
    \uninspectedCoords{
      (1, 5)
    };

    \node[above] at (axis cs:3,3) {$1$};
    \node[below] at (axis cs:3,3) {$(3, 3)$};

    \node[above] at (axis cs:4,4) {$2$};
    \node[below] at (axis cs:4,4) {$(4, 4)$};

    \node[above] at (axis cs:1,1) {$3$};
    \node[below] at (axis cs:1,1) {$(1, 1)$};

    \node[above] at (axis cs:5,6) {$4$};
    \node[below] at (axis cs:5,6) {$(5, 6)$};
    
    \node[above] at (axis cs:1,5) {$5$};
    \node[below] at (axis cs:1,5) {$(1, 5)$};

    \addplot[dashed, domain=0:6] {x};
  \end{axis}
\end{tikzpicture}
}

\newcommand{\plotWeightedMatroidIterationFour}[1]{
\begin{tikzpicture}[scale=#1]
  \begin{axis}[
    xlabel={$v$ (value)},
    ylabel={$p$ (price)},
    title={Budget 15, step 4},
    xmin=0, xmax=6,
    ymin=0, ymax=11,
    axis x line*=bottom,
    axis y line*=left,
  ]
    \acceptedCoords{
      (4, 4.8)
      (1, 1.2)
      (5, 6)
    };

    \rejectedCoords{
      (3, 3.6)
     };
    \uninspectedCoords{
      (1, 5)
    };

    \node[above] at (axis cs:3,3.6) {$1$};
    \node[below] at (axis cs:3,3.6) {$(3, 3.6)$};

    \node[above] at (axis cs:4,4.8) {$2$};
    \node[below] at (axis cs:4,4.8) {$(4, 4.8)$};

    \node[above] at (axis cs:1,1.2) {$3$};
    \node[below] at (axis cs:1,1.2) {$(1, 1.2)$};

    \node[above] at (axis cs:5,6) {$4$};
    \node[below] at (axis cs:5,6) {$(5, 6)$};
    
    \node[above] at (axis cs:1,5) {$5$};
    \node[below] at (axis cs:1,5) {$(1, 5)$};

    \addplot[dashed, domain=0:6] {1.2*x};
  \end{axis}
\end{tikzpicture}
}

\newcommand{\plotWeightedMatroidIterationFive}[1]{
\begin{tikzpicture}[scale=#1]
  \begin{axis}[
    xlabel={$v$ (value)},
    ylabel={$p$ (price)},
    title={Budget 15, step 5},
    xmin=0, xmax=6,
    ymin=0, ymax=11,
    axis x line*=bottom,
    axis y line*=left,
  ]
    \acceptedCoords{
      (4, 6)
      (1, 1.5)
      (5, 7.5)
    };

    \rejectedCoords{
      (3, 3.6)
      (1, 5)
    };

    \node[above] at (axis cs:3,3.6) {$1$};
    \node[below] at (axis cs:3,3.6) {$(3, 3.6)$};

    \node[above] at (axis cs:4,6) {$2$};
    \node[below] at (axis cs:4,6) {$(4, 6)$};

    \node[above] at (axis cs:1,1.5) {$3$};
    \node[below] at (axis cs:1,1.5) {$(1, 1.5)$};

    \node[above] at (axis cs:5,7.5) {$4$};
    \node[below] at (axis cs:5,7.5) {$(5, 7.5)$};
    
    \node[above] at (axis cs:1,5) {$5$};
    \node[below] at (axis cs:1,5) {$(1, 5)$};

    \addplot[dashed, domain=0:6] {1.5*x};
  \end{axis}
\end{tikzpicture}
}

\allowdisplaybreaks

\begin{document}

\maketitle

\begin{abstract}
We envision a marketplace where diverse entities offer specialized ``modules'' through APIs, allowing users to compose the outputs of these modules for complex tasks within a given budget. This paper studies the market design problem in such an ecosystem, where module owners strategically set prices for their APIs (to maximize their profit) and a central platform orchestrates the aggregation of module outputs at query-time. One can also think about this as a first-price procurement auction with budgets. The first observation is that if the platform’s algorithm is to find the optimal set of modules then this could result in a poor outcome, in the sense that there are price equilibria which provide arbitrarily low value for the user. We show that under a suitable version of the ``bang-per-buck'' algorithm for the knapsack problem, an $\eps$-approximate equilibrium always exists, for any arbitrary $\eps > 0$. Further, our first main result shows that with this algorithm any such equilibrium provides a constant approximation to the optimal value that the buyer could get under various constraints including (i) a budget constraint and (ii) a budget and a matroid constraint. Finally, we demonstrate that these efficient equilibria can be learned  through  decentralized  price  adjustments  by  module  owners  using  no-regret  learning  algorithms.
\end{abstract}

\newcommand{\tikzscale}{0.6}

\section{Introduction}
Our work is inspired by a marketplace model where sellers may participate to provide some product or service for a fee and a budget-constrained buyer whose goal is to purchase from some subset of sellers to maximize some objective function. Our primary motivation is in a computing ecosystem where the sellers may be offering specialized ``modules'' which are accessed through Application Programming Interfaces (API). The sellers would offer a price which must be paid every time a buyer issues an API call to their module. These modules may be specialized for particular tasks and may be built with proprietary data or with proprietary techniques.

We define and study the market design problem inherent in building such an ecosystem. We consider the following setup:
\begin{enumerate}[label=(\alph*),itemsep=0pt,topsep=0pt]
\item Module owners provide their module API at a price per API call.

\item Users submit requests and a budget to a central platform.

\item The platform calls a feasible subset of available APIs, respecting the budget constraint.

\item The platform aggregates the results which is provided back to the user.
\end{enumerate}

One particular example of such an ecosystem stems from the recent development of large language models (LLMs). While there may be a few flagship models that boast multiple generic capabilities, we are also witnessing a rise in specialized LLMs tailored for specific domains such as code generation (e.g., \cite{li2022competition, roziere2023code, guo2024deepseek}), medical knowledge (e.g., \cite{singhal2023large, singhal2023towards, luo2022biogpt}), legal tasks (e.g., \cite{wehnert2022applying, shui2023comprehensive}), and mathematical problem-solving (e.g., \cite{trinh2024solving}). In such an ecosystem, one can imagine that a user first enters a prompt into the generic LLM along with a budget constraint. Next, the generic LLM produces an answer by combining results from a subset of the available modules that respects the budget constraint.
Another example that fits into our marketplace model are labor markets such as Upwork \cite{upwork}. In such a market, the platform may allow workers to post hourly wages and then buyers need to assemble teams based on the workers' skills and a total budget.

The marketplace problem in this paper can also be thought of as a procurement auction. More specifically, the mechanism we consider in this paper is some type of a first-price auction. The vast majority of the literature on procurement auctions focuses on designing mechanisms that are truthful \citep{singer2010budget,chen2011approximability, balkanski2022deterministic, han2024triple, ensthaler2014dynamic, jarman2017ex}. This necessitates having a centralized platform that determines the price of every module depending on the prices of every other module. In contrast, a marketplace can be more decentralized; the price of every module can be chosen more or less independently.
One work that studies first-price procurement auctions is \cite{milgrom2020clock}.
\citet{milgrom2020clock} show that descending auctions, which are truthful, can be converted to first-price auctions.
However, one issue with such a reduction, as pointed out by \citet{milgrom2020clock}, is that these could produce additional equilibria for which there are no guarantees on their efficiency.
Thus we can also view our work as understanding what types of allocation ensure that a first-price auction contain only (approximately) efficient equilibria.
We discuss more about procurement auctions in Section~\ref{subsec:related}.

\subsection{Model and Results}
\label{subsec:model_and_results}
In this model, API prices are set by the module owners.  These prices depend on various factors.
First, module owners incur costs for module production and maintenance, including data licensing, partnerships, engineering, and compute. The module owner thus might have a private cost $c$ per API call that it wishes to recoup. The module owner can post a price $p$ per API call with the goal of optimizing its profit of $p - c$, whenever the API is called. 

Depending on a particular request, the different modules may be complementary or substitutable. Going back to the modular LLM example in the introduction, it might be that the user issued a prompt to understand treatment options for some condition. A medical module may be used to help diagnose the patient while a legal module may be useful to clarify insurance policies. However, multiple medical or legal modules may provide limited value and thus are substitutable. We assume a value function that captures the value from any subset of modules. We also assume a budget $B$ on the total that can be spent on any request. The goal is to pick a subset of modules to optimize the value subject to the budget constraint.

The value function and the budget create competition between the module owners and the module owners must strategically set their prices to maximize their profit while ensuring that they are selected. Our results concern the equilibria of this price setting game, where module owners facing repeated queries from the users do not wish to deviate from the API price they have set. 

We would like to understand the existence, potential inefficiency, and learnability of such equilibria. 
Our benchmark for inefficiency is $\opt$, which is the value obtained by the optimal allocation when each module owner sets their price equal to their cost $c$.
Our first, fairly trivial, result shows that the most natural selection rule fails miserably.
\begin{result}
If, given a set of prices, the platform always selects the optimal set of modules subject to the budget constraint then there is an equilibrium where all modules set prices equal to the budget and the user can obtain only one module. 
\end{result}
To see this, imagine that there are $n$ modules who all have cost $0$ and our budget is $1$.
Module $1$ has value $1+\eps$ while all other modules have value $1$.
The optimal value is $n+\eps$ but if all modules set a price of $1$ then we only obtain module $1$ for a value of $1+\eps$.
Moreover, no module can deviate by itself to get selected so this is, in fact, an equilibrium.

At this point, the most pressing question is whether or not there even exists an allocation rule which guarantees that all equilibria are efficient.
A second question is whether or not the allocation can be computed efficiently.
Our main result is that there is an efficient algorithm, based on a natural greedy algorithm, that guarantees that all equilibria are approximately optimal.
This result holds for a large set of value functions generalizing additive values, as described next.

\paragraph{Generalized Additive Valuations}
We assume that the value for a set of modules is additive over the set of selected modules. 
Further, we also assume that there is a matroid constraint that specifies which subset of modules is feasible to select; this captures substitutability between modules.
Matroid feasibility constraints capture many diverse settings.
For instance, a linear matroid models scenarios where a set of modules $S_1$ ``covers'' the content of another set $S_2$ 
(e.g., $S_2$ is linearly dependent on $S_1$), rendering $S_2$ redundant. A partition matroid, a special case, divides modules by expertise, allowing only one module per partition (e.g., medical or legal modules).

We consider the natural ``bang-per-buck'' algorithm, $\greedy$, which works as follows.
We first sort the modules from highest to lowest bang-per-buck (i.e.~$\frac{\vec v(i)}{\vec p(i)}$).
Then we select modules in this order while maintaining both budget and matroid feasibility.
The stopping rule of this algorithm is the \emph{first} time that the budget constraint is violated.
We give a precise definition of this algorithm can be found in Algorithm~\ref{alg:greedy}.

Our first main result, stated informally here, is that for such a selection rule, equilibria exist.
\begin{result}
\label{result:eq_exists}
When the platform implements $\greedy$, equilibria exists where no seller can profitably deviate.
\end{result}
We note that our formal results prove the existence of ``$\eps$-equilibria'' where no module can gain more than $\eps$ utility by deviating.
This is mainly to deal with tie-breaking and we can make $\eps$ arbitrarily small.
The formal results are stated in Theorem~\ref{theorem:additive_eq_exists} (for only a budget constraint), Theorem~\ref{thm:unweighted_eq_exist} (for uniform values with matroid and budget constraints), and Theorem~\ref{thm:weighted_matroid_eq_exists} (for general additive values with matroid and budget constraints).

Our second main result is that all equilibria are efficient.
\begin{result}
\label{result:approx}
Suppose that the platform implements $\greedy$ for module selection and that the cost of any module is at most $\lambda B$ for some $\lambda < 1$ (that is, it is small relative to the budget).
Let $\OPT$ be the maximum value achievable when the costs are \emph{known}.
Then any equilibrium yields a value of at least:
\begin{itemize}
\item $\frac{(1-\lambda)^2}{2 - \lambda} \cdot \OPT$ when there is only a budget constraint (Theorem~\ref{thm:additive_approx});
\item $\frac{(1-\lambda)^2}{2 - \lambda} \cdot \OPT$ when there is a matroid and a budget constraint but the values are uniform across all modules (Theorem~\ref{theorem:unweighted_approx}); and
\item $\frac{(1-3\lambda)^2}{2-3\lambda} \cdot \OPT$ when there is a matroid and a budget constraint. For this result, we also require $\lambda < 1/3$ (Theorem~\ref{thm:weighted_matroid_approx}).
\end{itemize}
Note that all bounds approach $\frac{1}{2}$ as $\lambda \to 0$.
\end{result}
Note that the last approximation result stated above is a simplified but slightly worse version of the result in Theorem~\ref{thm:weighted_matroid_approx}.
Result~\ref{result:approx} only shows that all equilibria are good.
However, it does not show how the market can achieve such any equilibrium.
Indeed, equilibrium prices are a complex function of all the parameters of the market, including the set of modules $M$, their private costs $\vec c$, the value function $\vec v$ and the underlying constraint over the modules. It is impossible for market participants to compute an equilibrium with no prior knowledge of the parameters. However, one may expect module owners to employ algorithmic techniques to adjust prices through reasonable learning algorithms, especially if they anticipate long-term participation.
Such a learning algorithm would respond to the feedback obtained in past rounds at various price levels without knowledge of the other market participants or even the value function of the platform.

Our next result proves convergence guarantees if all modules use \emph{multiplicative weight update} type no-regret learning algorithms. 

\begin{result}
For any $\eps > 0$, if the platform implements $\greedy$ for module selection and each module, unaware of market parameters, employs a multiplicative weight update-style price learning algorithm, then the price dynamics converges to an $\eps$-equilibrium price with high probability under any arbitrary matroid constraint over the modules. Consequently as $\eps \rightarrow 0$, the total expected value for a query for the platform is at least $\frac{(1-3\lambda)^2}{2-3\lambda} \cdot \OPT$, with high probability.
\end{result}
We note that our convergence results assumes that every round is identical in that the buyer's valuation function does not change over each round. If there are multiple ``types'' of buyers, one can imagine that the type of the buyer is used as context for the learning agents. In such a setting, the problem reduces to the single buyer valuation setup in this paper.

\subsection{Related Work}
\label{subsec:related}
\paragraph{Procurement Auctions.}
Our work fits into the literature of procurement auctions. There is a large body of work on designing procurement auctions subject to a budget constraint that was initiated by \citet{singer2010budget}.
These procurement auctions can either be done via sealed-bid auctions  \cite{singer2010budget,chen2011approximability} or descending clock auctions \cite{balkanski2022deterministic, han2024triple, ensthaler2014dynamic, jarman2017ex}. 

When the buyer's objective function is submodular, \citet{balkanski2022deterministic} designed deterministic clock auctions with an approximation ratio of 4.5, improving upon the approximation ratio of both deterministic and randomized mechanisms from a series of works including \cite{chen2011approximability,jalaly2021simple}. For the special case of submodular functions, \citet{leonardi2017budget} designed a $4$-approximation mechanism for unweighted matroid value functions, and \citet{gravin2020optimal} designed an optimal $(1+\sqrt 2)$-approximate mechanism for additive valuations. In addition, a series of works \cite{anari2018budget,rubinstein2022beyond} designed mechanisms with a small seller assumption $( \lambda \rightarrow 0)$ and obtained an optimal $(1-1/e)$-approximation for additive valuations and a $1/2$-approximation for monotone submodular valuations. Our result on matroid value functions matches the current best-known bound under the small seller assumption.\footnote{In addition, we ensure that all potential price equilibria are $1/2$-approximate.}

In this paper, our goal is to formally understand the marketplace model where module owners simply post a public price for their API call. This is a model that is currently being employed in the industry (e.g.~\cite{togetherpricing, vertexaipricing}) due to its simplicity. Equivalently, our model is a first-price procurement auction with a budget constraint. In such a marketplaces, since the payment is decided by the sellers, the mechanism needs to ensure that all potential equilibria prices lead to efficient outcomes which is not necessarily satisfied by existing mechanisms.

Such a marketplace setting was formally studied by \citet{milgrom2020clock} in their beautiful work which proposed a clean characterization of descending clock auctions.
\citet{milgrom2020clock} prove that the outcome of a descending clock auction can be converted to a first-price auction.
Moreover, the outcome of the descending clock auction corresponds to equilibrium bids in the first-price auction.
However, the allocation rule in the first-price auction is somewhat complex as it is defined by ``simulating'' a clock auction. A larger issue is that the conversion they describe may result in additional equilibria unless the allocation satisfies a ``non-bossy'' assumption \citep[Appendix E]{milgrom2020clock}. This assumption requires that a bidder cannot change the allocation of any other bidder without changing their own allocation. This is a non-trivial assumption to satisfy. For example, the allocation in \cite{balkanski2022deterministic}, which achieves the best approximation ratio for procurement auctions, does not satisfy the non-bossy assumption. If the allocation rule does not satisfy the non-bossy condition then the additional equilibria may have poor efficiency. In fact, it is not difficult to observe that their allocation rule leads to an equilibria that obtains $O(1/n)$ fraction of optimal value for the same example for which the optimal selection rule leads to inefficient outcome. Thus a contribution of our work can also be seen as designing allocation rules where all equilibria in a first-price auction are approximately efficient.

In another previous work, \citet{immorlica2005first} study a first-price auction in a procurement auction setting, without a budget constraint.
They study path auctions where the agents correspond to edges and the constraint is a set of edges between a source and a sink with a specified flow capacity. This makes it a covering problem whereas we consider a packing problem.
In addition, the goal of the paper is to minimize the buyer's cost while our goal is to maximize the buyer's value.

Having a marketplace model also has practical implications in that it obviates the need for complex auction infrastructure such as second-price auctions or clock auctions.
While truthful reporting is no longer optimal, our results show that if the module owners simply implement some reasonable learning algorithm, they are guaranteed to converge to an equilibrium.
This essentially allows module owners to use learning algorithms to adjust their prices based on demand feedback.
\paragraph{Pricing and Bidding Dynamics.}
One of the goals of this paper is to understand pricing dynamics in markets.
For Arrow-Debreu markets \cite{mckenzie1954equilibrium,arrow1954existence}, there is a large literature on understanding the dynamics that lead to equilibria (see for example, \cite{rabani2021invisible,cheung2013tatonnement} and references therein).
However, the models studied in these papers assume that the goods are a finite resource and are divisible. In our setting, the goods are digital and are indivisible.
Furthermore, these papers assume that agents do not directly control prices, which are set in a centralized manner based on aggregate supply and demand. In contrast, our model allows agents to directly change their prices based on the response from the market.

Our work also fits into a growing literature of understanding bidding dynamics in auctions \cite{feng2021convergence,deng2022nash,kolumbus2022auctions}.
These papers assume that bidders use some sort of no-regret or mean-based learning algorithm.
\citet{feng2021convergence, deng2022nash, kolumbus2022auctions} aim to understand the convergence properties of these dynamics. \citet{kolumbus2022auctions} also study the incentive properties of different auctions when an agent uses no-regret learning algorithms.
Our paper also assumes that agents use no-regret learning algorithms to adjust their prices.
However, a key difference is that, to the best of our knowledge, our work is the first to study no-regret dynamics in a procurement setting.

Our convergence result falls within the "learning in games" literature \cite{hannan1957approximation,giannou2021survival,hart2000simple}, which considers the setting where players have no information about game parameters and interact with the platform over multiple iterations, using an online learning algorithm to select actions from a finite set. Crucially, here the game does not change over the interactions. Our convergence result is interesting because \cite{giannou2021survival} showed that if players implement a "follow the regularized leader" type learning algorithm (which includes multiplicative weight update), then the dynamics converge to a strict-pure Nash equilibrium. However, in our setting, strict pure Nash equilibria do not exist.  Instead, the "stability" properties of our allocation rule and equilibrium prices ensure that the learning sellers converge to a pure equilibrium price.

\section{Preliminaries} \label{sec:prelim}
\paragraph{Notation.}
We will generally use bolded letters such as $\vecv$ to denote vectors.
For an index $i$, we use $\vecv(i)$ to denote the $i$th coordinate of $\vecv$ and $\vecv_{-i}$ to denote the vector $\vecv$ without the $i$th coordinate.
For a set $S$ of indices, we use $\vecv(S) = \sum_{i \in S} \vecv(i)$.
For two vectors $\vecv^1, \vecv^2$, we write $\vecv^1 \leq \vecv^2$ to indicate that $\vecv^1(i) \leq \vecv^2(i)$ for all indices $i$ (and similarly for all other comparisons). For any positive integer $n$, we let $[n] = \{1,\dots , n\}$.
We use $\pi \colon [n] \to [n]$ to denote a permutation where $\pi(i)$ denotes the element at index $i$.
For a set $S \subseteq [n]$, we write $\pi(S) = \{ \pi(i) \,:\, i \in S\}$. We also denote $\pi[k] = \pi([k])$.
In the remainder of the paper, we will generally refer to a module only via its index $i$ instead of $m_i$. We provide basic preliminaries of matroid theory in Appendix~\ref{appendix:matroid_prelims}.

\paragraph{Price Competition Game}
Throughout this paper, we use buyer interchangeably with the platform and a seller interchangeably with a module owner.
An instance of the price competition game is defined by $\langle M, \vec v, \vec c, \mathcal I, B \rangle $ where $M= \{m_1, \dots, m_k\}$ are the set of modules, $\vec v(i)\in \mathbb R_+\cup \{0\}$ is the value of module $m_i$ for the buyer, $\vec c(i) \in \mathbb R_+\cup \{0\}$ is the cost of module $m_i$, $\mathcal I\subseteq 2^M$ is the set of feasible modules for the buyer, and $B$ is the buyer's budget.  

For the given instance of a price competition game $\langle M, \vec v, \vec c, \mathcal I, B \rangle $, we assume that each module owner has full information of the instance including the values $\vec v$, costs $\vec c$ and constraints $\mathcal I$. We do note that full information about the valuation function of the buyer, constraint $\mathcal I$, and cost of other module owners is unrealistic. However, understanding the price competition game in the full-information setting already turns out to be non-trivial and acts as a building block towards understanding the modular marketplace where the module owners do not have any information about the price competition instance except their private costs. In addition, we assume that the buyer does have the information about the values $\vec v$, constraint $\mathcal I$, and budget $B$ but does not have any information about the costs of the modules $\vec c$.

Given an instance of a price competition game $\langle M, \vec v, \vec c, \mathcal I, B\rangle $, the strategy space of the module owner $m_i\in M$ is a price of their module $\vec p(i) \in [\vec c(i), B]$.\footnote{
Note that $\vecp(i) < \vecc(i)$ is always dominated by setting a price of $\vecc(i)$ since bidding less than $\vecc(i)$ always results in non-positive utility for the seller.}
We let $\mathcal A \colon \mathbb R_+^M \rightarrow 2^M$ be the selection rule for the buyer that takes the prices of the modules $\vec p$ as an input and outputs a set of modules such that for any price vector $\vec p$, we have (a) $\sum_{i\in \mathcal A(\vec p)} \vec p(i) \leq B$ and (b) $\mathcal A(\vec p) \in \mathcal I$. 

For the given price vector $\vec p$, we let $u_i^\mathcal A (\vec p) = ( \vec p(i) - \vec c(i)) \cdot \mathbbm 1[i\in \mathcal A(\vec p)]$ be the utility of module $i$ when the buyer implements selection rule $\mathcal A$.
We say that $\vec p$ is an equilibrium price for the price competition game if no module can deviate from its posted price and increase its utility, i.e.~for all $i\in [k]$ and $\vec p(i)' \neq \vec p(i)$, we have,
$u_i(\vec p(i)',\vec p_{-i}) \leq u_i (\vec p(i), \vec p_{-i})$.  Similarly, for any $\eps >0$, we say that the prices $\vec p$ is an $\eps$-equilibrium price if  for all $i\in [k]$ and $\vec p(i)' \neq \vec p(i)$, we have,
$u_i(\vec p(i)',\vec p_{-i}) \leq u_i (\vec p(i) , \vec p_{-i}) + \eps$.

To evaluate the quality of the selection rule $\mathcal A$, we compare the value obtained by the buyers at any equilibrium to an omniscient benchmark $\opt$ which is the best possible value achieved by the buyer in the case when it knows the costs of the modules $\vec c$, i.e. $\max_{S\in \mathcal I, \vec c(S)\leq B} \vec v(S).$ We say that allocation rule $\mathcal A$ has an \emph{approximation ratio} of $\alpha \in [0,1]$ if for any equilibrium price $\vec p$ for the price competition instance $\langle M, \vec v, \vec c, \mathcal I, B \rangle $ with selection rule $\mathcal A$, we have $\sum_{i\in \mathcal A(\vec p)} \vec v(i) \geq \alpha \cdot \opt$.

\paragraph{Learning Dynamics in the Price Competition Game}
In the above discussion, we defined a price competition game between the modules where we assume that all the modules possess full information about the game instance  $\langle M, \vec v, \vec c, \mathcal I, B\rangle $.  However, the cost of each module is private information, i.e. only module $m_i$ has information of their private $\vec c(i)$, and the rest of the modules are unaware of the private cost of module $m_i$. In addition, the modules have no information about the value vector $\vec v$ and underlying constraints $\mathcal I$. As a result, without full information about market parameters, the modules cannot implement equilibrium prices directly.

In this section, we formally define the learning dynamics of the price competition market, where each module owner employs a learning algorithm to set their prices. More specifically, we fix an instance of the price competition game $\langle M, \vec v, \vec c, \mathcal I, B\rangle $, and we consider a setting where a single platform repeatedly interacts with the modules over $T$ rounds.  At each round, a price competition game instance $\langle M, \vec v, \vec c, \mathcal I, B\rangle $ is played. 

We assume that module $m_i\in M$ implements an \emph{online learning algorithm} $\mathcal L(i)$ to set its price.
At round $t$, the algorithm selects a price $\vec p^t(i)$ (possibly at random) from the set $\mathcal B = \{\delta, 2\cdot \delta , \dots ,B\}$,
where $\delta \in (0, 1)$ represents the minimum price increment.
The price $\vec p^t(i)$ at round $t$ is determined based on the history of the game up to round $t-1$. At each round $t$, the platform provides feedback to module $m_i$ regarding the maximum price at which the module could have been selected in that round.
We denote the history at round $t$ for module $m_i$ as the set of prices, the set of selected modules, and the received feedback from the platform up to round $t-1$.

The goal of each module's learning algorithm $\mathcal L(i)$ is to adjust its price iteratively based on the feedback received, with the aim of converging towards an optimal pricing strategy. In particular, the platform's feedback provides each module with information that helps refine its price selection to maximize its reward over time. This repeated interaction across $T$ rounds allows modules to "learn" about the market conditions and the behavior of competitors, even with incomplete information about other modules' costs and the platform's value structure.

Given the price competition game instance $\langle M, \vec v, \vec c, \mathcal I, B\rangle $ and learning algorithms of the modules $\{\mathcal L(i) \,:\, i\in M\}$ induces a price dynamics denoted as a sequence of price vectors $\{\vec p^t \,:\, t=1,2,\dots\}$. We next define convergence of the price dynamics of the price competition game which is a generalization of the standard definition of convergence of the sequences. 
\begin{definition}\label{def:multiplicative_weight_learning}
The learning dynamics of the price competition game instance $\langle M, \vec v, \vec c, \mathcal I , B \rangle$ with learning algorithms of the modules $\{\mathcal L(i) \,:\, i\in M\}$ converges to the price vector $\vec p$ if for any $\gamma >0$, there exists $T^* \geq  g(\gamma)$ such that 
\begin{equation*}
    \Pr \left[ |\vec p^t(i) - \vec p(i)|\leq \sqrt \delta: \forall i\in M \text{ and } t\geq T^* \right]\geq 1 - \gamma.
\end{equation*}
\end{definition}
In this work, we consider the dynamics where each module implements \emph{multiplicative weight update} type learning algorithm to set the price which we define in the following definition.

\begin{definition}[Multiplicative Weight Type Learning Algorithm]
Given set of bids $\mathcal B$ and total cumulative reward of price $p\in \mathcal B$ denoted as $\sigma_t(a) = \sum_{t'\leq t} u_i^\mathcal A(p,\vec p_{-i}^t)$, we say that the algorithm $\mathcal L$ is multiplicative weight type with learning rate $\gamma_t$ if for any two prices $p,p' \in \mathcal A$ and some constant $\tau>0$, whenever $\sigma_{t}(p') \leq \sigma_t(p) - \tau \cdot t$ then the probability that the algorithm sets price $p'$ at round $t+1$ is less than $ \exp(-\gamma_t \cdot \tau \cdot t)$.
\end{definition}

\section{Technical Overview}
\label{sec:technical_overview}
In this section, we give a high-level overview of our techniques.
First, Algorithm~\ref{alg:greedy} defines the greedy algorithm that we use in this paper.
Throughout this paper, we refer to Algorithm~\ref{alg:greedy} as $\greedy$.
Note that $\greedy$ exits as soon as the budget constraint is violated.

\begin{algorithm}
\caption{Greedy Algorithm ($\greedy$)}
\label{alg:greedy}
\begin{algorithmic}[1]
\State \textbf{Input}: price vector $\vec p$, values $\vec v$, budget $B$, feasible sets $\I \in 2^{[n]}$.
\State \textbf{Output}: set $G$ of modules.
\State Initialize $G \gets \emptyset$.
\State Re-arrange modules such that $\frac{\vecv(1)}{\vecp(1)} \geq \frac{\vecv(2)}{\vecp(2)} \geq \dots \geq \frac{\vecv(n)}{\vecp(n)}$ with ties broken arbitrarily.
\For{$i=1,2 \dots ,n$}
\If{$i \in \spn(G)$}
\State Let $C$ be the unique circuit in $G \cup \{i\}$.
\State Let $j \in \argmin\{ \vecv(j') \,:\, j' \in C\}$.
\State Set $G' \gets G \cup \{i\} \setminus \{j\}$.
\Else
\State Set $G' \gets G \cup \{i\}$.
\EndIf
\If{$\vecp(G') > B$}
\State \textbf{break}
\EndIf
\State $G \gets G'$.
\EndFor
\State \textbf{return} $G$
\end{algorithmic}
\end{algorithm}

\subsection{Existence of Equilibria}

In this paper, we will give constructive proofs that $\eps$-equilibria exist in the pricing game when the allocation rule is given by Algorithm~\ref{alg:greedy} (see the formal results and proofs in Section~\ref{subsec:eq_exists}, Section~\ref{sec:matroid_price_eqm_existance}, and Section~\ref{sec:weighted_matroid_equilibrium}).
In this section, we give a high-level intuition of how to make such constructions.
For simplicity, let us assume that the modules are sorted in increasing order of ``cost-per-value'' and that these values are distinct, i.e.~$\frac{\vecc(1)}{\vecv(1)} < \ldots < \frac{\vecc(n)}{\vecv(n)}$.
Note that this is equivalent to decreasing bang-per-buck order.

First, assume there is no matroid constraint.
At a high-level, the algorithm starts by assuming every module bids their cost.
However, this means that the module with the highest cost-per-value (the inverse of the bang-per-buck) can raise its bid since it is the first module inspected by $\greedy$.
So it increases its cost-per-value up until the next module inspected by $\greedy$.
Then the two modules, as a coalition, increase their price to the third module inspected by $\greedy$.
This process continues until a step $k$ where increasing modules $1, \ldots, k$ to the cost-per-value of module $k+1$ would cause the first $k+1$ modules to exceed the budget constraint.
When this happens, we instead raise the cost-per-value of the first $k$ modules as high as possible so that the first $k$ modules fit within the budget constraint.
We note that the set of module which increases their price remains selected by the greedy algorithm. In addition, increasing the price of these module can only decrease the utility of the rest of the module. Therefore, 
intuitively, this should be an equilibrium since no accepted module can increase their price while remaining accepted by $\greedy$ and no rejected module is willing to lower their price since it is already bidding its cost.

The argument becomes more complex when a matroid constraint is introduced.  Imagine we try to iteratively increase the prices of the modules, similar to how we did in the additive case. With a matroid constraint, a module (or set of modules) might have its price increased and still be selected by the greedy algorithm.  This price update, however, could cause the algorithm to drop a module with the lowest bang-per-buck (due to the matroid constraint) and swap it for a module with a higher bang-per-buck that wasn't previously selected at the original price.

This means that modules rejected earlier during the price update process can suddenly become ``important'' and might be willing to increase their prices in future updates.  This ``non-monotonic'' behavior makes things tricky. To handle this, we have carefully designed an algorithm (Algorithm~\ref{alg:equilibrium_dynamics_weighted}) that keeps track of these modules that are suddenly selected. In our equilibrium construction, we make sure that if such modules exhibit profitable price deviations, we update their prices accordingly. 

Figure~\ref{fig:matroidAlg} illustrates a sample run of this process for unweighted matroid using the graphic matroid shown in the top left.
In other words, the goal is to pick a subset of the edges with the constraint that no cycle is formed. We note that the unweighted case does satisfies certain nice monotonicity properties (Claim~\ref{claim:winners_stays_winners}) which is not present in the weighted case which requires more careful technical treatment for the quality of equlibria. 

In Figure~\ref{fig:matroidAlg}, each point corresponds to a module.
The value above the module is the name of the module while the coordinates below the module correspond to the value and price, respectively, of the module.
Each plot corresponds to a separate step $t$ for the construction of the equilibrium.
In addition, each plot only shows what $\greedy$ would do up until iteration $t$.
Question marks correspond to modules that have not been reached at the current step of the construction.
The \textcolor{blue}{blue dots} indicate modules \textcolor{blue}{accepted} by $\greedy$ by its $t$ iteration and \textcolor{red}{red crosses} indicate modules \textcolor{red}{rejected} by $\greedy$ by its $t$th iteration.

In the first three steps of the construction, we slowly raise the cost-per-value.
In each of these steps $t$, the first $t$ iterations of $\greedy$ ensures that all modules are accepted.
Let us now consider step $4$ of the construction and assume that $\greedy$ tie-breaks in lexicographic order.
In this case, $\greedy$ would start by accepting modules $1$, $2$, and $3$.
When it reaches module $4$, it would realize that $4$ forms a cycle with $1$ and $2$.
Module $1$ has the lowest value in the cycle so it is replaced by module $4$.
At this point, the price of module $1$ is moot but to better illustrate the operation of $\greedy$, we assume that it freezes at price.

Finally, let us look at step $5$.
The lowest cost-per-value module is module $1$.
So $\greedy$ begins by taking module $1$.
It then takes modules $2$ and $3$ and finally replaces $1$ with $4$.
Once it takes module $4$, the budget is exhausted and so $\greedy$ terminates.

Observe that the result in step $5$ is indeed an equilibrium.
Modules $2$, $3$, and $4$ cannot increase their prices since if they do, they would be fourth in the order and be rejected.
Module $1$ will always be replaced by $4$ whenever it becomes before $4$ and module $5$'s cost is too high for it be relevant.

\begin{figure}[h]
\centering
\begin{subfigure}{0.3\textwidth}
  \centering
  \scalebox{0.9}{
\begin{tikzpicture}[
    node distance=1cm
]
  \node[circle,draw,fill] (a) {};
  \node[circle,draw,fill,right of=a] (b) {};
  \node[circle,draw,fill,right of=b] (c) {};
  \node[circle,draw,fill,right of=c] (d) {};
  \node[circle,draw,fill,right of=d] (e) {};

  \draw (a) -- node[above] {1} (b);
  \draw (b) -- node[above] {2} (c);
  \draw (c) -- node[above] {3} (d);
  \draw (d) -- node[above] {5} (e);
  \draw (a) edge[bend right=60,looseness=1.5] node[above] {4} (c);
  
  \node[above of=c] {\normalsize The graphic matroid};
\end{tikzpicture}
}
\end{subfigure}%
\hfill
\begin{subfigure}{0.3\textwidth}
  \centering
  \plotWeightedMatroidIterationOne{\tikzscale}
\end{subfigure}%
\hfill
\begin{subfigure}{0.3\textwidth}
  \centering
  \plotWeightedMatroidIterationTwo{\tikzscale}
\end{subfigure}%

\medskip

\begin{subfigure}{0.3\textwidth}
  \centering
  \plotWeightedMatroidIterationThree{\tikzscale}
\end{subfigure}%
\hfill
\begin{subfigure}{0.3\textwidth}
  \centering
  \plotWeightedMatroidIterationFour{\tikzscale}
\end{subfigure}%
\hfill
\begin{subfigure}{0.3\textwidth}
  \centering
  \plotWeightedMatroidIterationFive{\tikzscale}
\end{subfigure}%
\caption{Sample run for equilibrium construction. See main text for description.}
\label{fig:matroidAlg}
\end{figure}
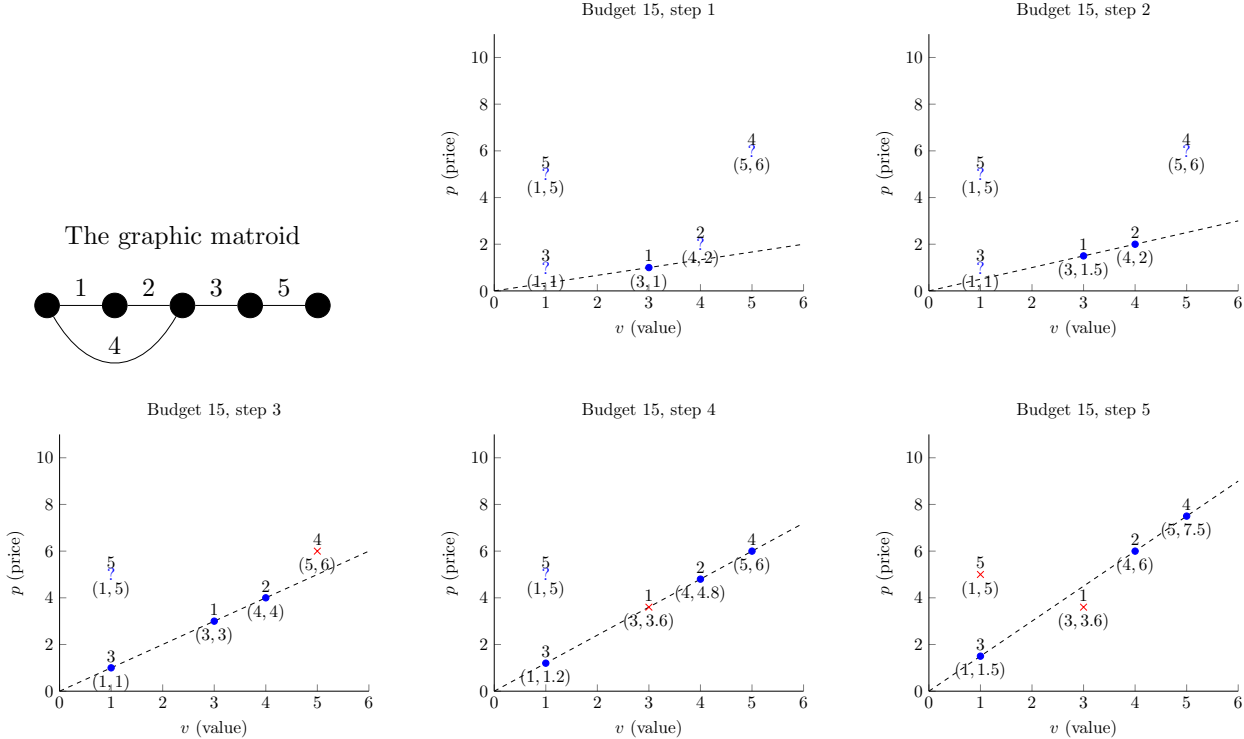

\subsection{Quality of Equilibria}
\paragraph{Characterizing worst-case equilibria.}
Our first step to prove approximation results of equilibria is to first characterize the set of ``worst-case'' equilibria.
We first prove a technical result that shows that in the worst-case equilibria, roughly speaking, all rejected modules bid their cost.
This statement is exactly true in the setting without a matroid constraint and in the setting with a matroid constraint but with uniform values.
For the setting with general values, this statement is only partly true and we can only make this characterization up to a single module; we will leave the technical details of this to Section~\ref{sec:weighted_matroid}.

Our overall approach is to demonstrate that, starting from any equilibrium price, we can reach another equilibrium through a sequence of price deviations.  This sequence will ensure that all "crucial" non-selected module sets have their prices equal to their costs, while simultaneously decreasing the total value for the platform. We then analyze the value of such equilibria using properties of the greedy algorithm, allowing us to directly compare it to the optimal value, OPT.

The main technical challenge lies in the following: once we update the price of a non-selected module to its cost, it can create several new potential profitable deviations for other modules. These deviations must be addressed to reach a new equilibrium without changing the price of the non-selected module we just adjusted.  Our core technical analysis focuses on tracking all these potential new deviations and arguing that our specific sequence of deviations ultimately leads to an equilibrium where the "crucial" non-selected modules have their prices set equal to their costs.

To give more high-level intuition of our approach, we start with an equilibrium price vector $\vecp$ and some module $i$ such that $\vecp(i) > \vecc(i)$.
We then decrease the price of module $i$ to $\vecc(i)$.
Note that module $i$ must remain rejected because otherwise it could have deviated in $\vecc(i)$.\footnote{Here, we ignore the edge case where module $i$ can be accepted with non-negative utility if and only if it bids $\vecc(i)$.}
Now consider the new price vector given by $(\vecc(i), \vecp_{-i})$.
Let $S$ denote the set of modules selected by $\greedy$ when $\greedy$ considers module $i$.
If $i$ is budget-feasible when $\greedy$ inspects it but is then rejected anyway then it turns out this is already an equilibrium.
The reason is as follows.
Observe that $i$'s deviation could result in one of two possibilities.
First, it could be that when $\greedy$ reaches module $i$, including $i$ would have already formed a circuit and module $i$ would be the least valuable module in that circuit.
The other is that $\greedy$ temporarily adds module $i$. In this case, it must be that module $i$ is later removed since it forms a circuit.
In both of these cases, the intuition is that the price of $i$ is somewhat moot (provided that it bids above its cost).
In which case, it did not really matter than $i$ had placed a bid above its cost in the first place.

The difficult case is when module $i$ could have been added to $S$ by $\greedy$ but was not added because doing so would violate the budget constraint.
This may cause modules after $i$ to now be rejected which were previously accepted.
To resolve this, we need a two-step process.
First, we modify the prices of all newly rejected modules so that their price is the greater of their own cost or the price given by module $i$'s current bang-per-buck.
Note that not all of these modules may be able to accept the new price.
This can allow some modules which come before $i$ and that were previously rejected to now be accepted.
These modules now have an incentive to deviate.
So the second step is to raise their prices up to the bang-per-buck of $i$.
We show that this modification results in another equilibrium but that this equilibrium can only be worse.

Figure~\ref{fig:worst_equil_sketch} pictorially shows the construction in the latter case where including $i$ would violate the budget constraint.
For simplicity, we assume that all sets are independent.
The \textcolor{red}{red crosses} correspond to modules that would be \textcolor{red}{rejected} by $\greedy$ and the \textcolor{blue}{blue dots} correspond to modules that would be \textcolor{blue}{accepted} by $\greedy$.
In the left plots, the number above the points correspond to the module name while the numbers below the points correspond to their value, current price, and true cost, respectively.
For the middle and right plots, we only label the true cost of each module.
In the top row, module $4$ deviates while in the bottom row, module $3$ deviates.
In both cases, the budget is $4.4$.

The left plots show an  initial equilibrium where modules $1$ are accepted at a budget of $4.4$.
This is indeed an equilibrium since module $3$'s true cost already exceeds the buyer's budget and module $4$'s bang-per-buck at its true cost is much larger than $1$, the current bang-per-buck in the equilibrium.

\paragraph{Top row: module $4$ deviates.}
This is the most straightforward case since module $4$'s deviation has no affect on any other module.
Thus, this is already an equilibrium.

\paragraph{Bottom row: module $3$ deviates.}
In this case, module $3$'s deviation affects the outcome for modules $1$ and $2$ because module $3$ is now dictating the bang-per-buck, which is approximately $0.643$.
Previously rejected modules, namely module $4$, have no incentive to deviate since they were already unwilling to accept the previously higher bang-per-buck.
Newly rejected modules, namely modules $1$ and $2$ will now decrease their price until either (i) they hit (or slightly below) the current bang-per-buck or (ii) they hit their cost.
In this case, we see that module $1$ can decrease below the current bang-per-buck while module $2$ stops at a price of $2.1$ (equivalently, a bang-per-buck of $0.7 > 0.643$).

\newcommand{\plotAdditiveEquilibriumOne}[1]{
\begin{tikzpicture}[scale=#1]
  \begin{axis}[
    xlabel={$v$ (value)},
    ylabel={$p$ (price)},
    title={Initial equilibrium},
    xmin=0, xmax=8.1,
    ymin=0, ymax=8.1,
    axis x line*=bottom,
    axis y line*=left,
  ]
    \acceptedCoords{
      (1, 1)
      (3, 3)
    };

    \rejectedCoords{
      (7, 7.1)
      (1, 6)
    };

    \node[above] at (axis cs:1,1) {$1$};
    \node[below] at (axis cs:1,1) {$(1, 1); 0$};

    \node[above] at (axis cs:3,3) {$2$};
    \node[below] at (axis cs:3,3) {$(3, 3); 2.5$};

    \node[above] at (axis cs:7,7.1) {$3$};
    \node[below] at (axis cs:6.8,7.1) {$(7, 7 + \delta); 4.5$};

    \node[above] at (axis cs:1,6) {$4$};
    \node[below] at (axis cs:1,6) {$(1, 6); 4$};

    \addplot[dashed, domain=0:10] {x};
  \end{axis}
\end{tikzpicture}
}

\newcommand{\plotAdditiveEquilibriumTwo}[1]{
\begin{tikzpicture}[scale=#1]
  \begin{axis}[
    xlabel={$v$ (value)},
    ylabel={$p$ (price)},
    title={Module $3$ deviates},
    xmin=0, xmax=8.1,
    ymin=0, ymax=8.1,
    axis x line*=bottom,
    axis y line*=left,
  ]
    % \acceptedCoords{
    % };

    \rejectedCoords{
      (1, 1)
      (3, 3)
      (7, 7.1)
      (7, 4.5)
      (1, 6)
    };
    
    \addplot[color=red,->, ultra thick] coordinates {
      (7, 6.9)
      (7, 4.7)
    };

    % \node[above] at (axis cs:1,6) {$4$};
    \node[below] at (axis cs:1,6) {$6$};
    
    % \node[above] at (axis cs:1,1) {$1$};
    \node[below] at (axis cs:1,1) {$1$};
    
    % \node[above] at (axis cs:3,3) {$2$};
    \node[below] at (axis cs:3,3) {$3$};

    % \node[above] at (axis cs:7,7.1) {$3$};
    \node[below] at (axis cs:7,4.5) {$4.5$};

    \addplot[dashed, domain=0:10] {x * 4.5 / 7};
  \end{axis}
\end{tikzpicture}
}

\newcommand{\plotAdditiveEquilibriumThree}[1]{
\begin{tikzpicture}[scale=#1]
  \begin{axis}[
    xlabel={$v$ (value)},
    ylabel={$p$ (price)},
    title={Modules $1$ and $2$ update},
    xmin=0, xmax=8.1,
    ymin=0, ymax=8.1,
    axis x line*=bottom,
    axis y line*=left,
  ]
    \acceptedCoords{
      (1, 0.64)
    };

    \rejectedCoords{
      (1, 1)
      (3, 3)
      (3, 2.5)
      (1, 6)
      (7, 4.5)
    };
    
    \addplot[color=red,->, ultra thick] coordinates {
      (2.8, 3.1)
      (2.8, 2.5)
    };
    
    \addplot[color=red,->, ultra thick] coordinates {
      (0.8, 1.1)
      (0.8, 0.64)
    };

    % \node[above] at (axis cs:1,6) {$4$};
    \node[below] at (axis cs:1,6) {$6$};
    
    % \node[above] at (axis cs:1,1) {$1$};
    \node[right] at (axis cs:1,0.64) {$0.642$};
    
    % \node[above] at (axis cs:3,3) {$2$};
    \node[below] at (axis cs:3,2.5) {$2.5$};

    % \node[above] at (axis cs:7,7.1) {$3$};
    \node[below] at (axis cs:7,4.5) {$4.5$};

    \addplot[dashed, domain=0:10] {x * 0.64};
  \end{axis}
\end{tikzpicture}
}

%%%

\newcommand{\plotAdditiveEquilibriumBOne}[1]{
\begin{tikzpicture}[scale=#1]
  \begin{axis}[
    xlabel={$v$ (value)},
    ylabel={$p$ (price)},
    title={Initial equilibrium},
    xmin=0, xmax=8.1,
    ymin=0, ymax=8.1,
    axis x line*=bottom,
    axis y line*=left,
  ]
    \acceptedCoords{
      (1, 1)
      (3, 3)
    };

    \rejectedCoords{
      (7, 7.1)
      (1, 6)
    };

    \node[above] at (axis cs:1,1) {$1$};
    \node[below] at (axis cs:1,1) {$(1, 1); 0$};

    \node[above] at (axis cs:3,3) {$2$};
    \node[below] at (axis cs:3,3) {$(3, 3); 2.5$};

    \node[above] at (axis cs:7,7.1) {$3$};
    \node[below] at (axis cs:6.8,7.1) {$(7, 7 + \delta); 4.5$};

    \node[above] at (axis cs:1,6) {$4$};
    \node[below] at (axis cs:1,6) {$(1, 6); 4$};

    \addplot[dashed, domain=0:10] {x};
  \end{axis}
\end{tikzpicture}
}

\newcommand{\plotAdditiveEquilibriumBTwo}[1]{
\begin{tikzpicture}[scale=#1]
  \begin{axis}[
    xlabel={$v$ (value)},
    ylabel={$p$ (price)},
    title={Module $4$ deviates},
    xmin=0, xmax=8.1,
    ymin=0, ymax=8.1,
    axis x line*=bottom,
    axis y line*=left,
  ]
    \acceptedCoords{
      (1, 1)
      (3, 3)
    };

    \rejectedCoords{
      (7, 7.1)
      (1, 6)
      (1, 4)
    };
    
    \addplot[color=red,->, ultra thick] coordinates {
      (1, 5.9)
      (1, 4.1)
    };

    % \node[above] at (axis cs:1,6) {$4$};
    \node[below] at (axis cs:1,4) {$4$};
    
    % \node[above] at (axis cs:1,1) {$1$};
    \node[below] at (axis cs:1,1) {$1$};
    
    % \node[above] at (axis cs:3,3) {$2$};
    \node[below] at (axis cs:3,3) {$3$};

    % \node[above] at (axis cs:7,7.1) {$3$};
    \node[below] at (axis cs:7,7.1) {$7+\delta$};

    \addplot[dashed, domain=0:10] {x};
  \end{axis}
\end{tikzpicture}
}

\newcommand{\plotAdditiveEquilibriumBThree}[1]{
\begin{tikzpicture}[scale=#1]
  \begin{axis}[
    xlabel={$v$ (value)},
    ylabel={$p$ (price)},
    title={Final prices},
    xmin=0, xmax=8.1,
    ymin=0, ymax=8.1,
    axis x line*=bottom,
    axis y line*=left,
  ]
    \acceptedCoords{
      (1, 1)
      (3, 3)
    };

    \rejectedCoords{
      (7, 7.1)
      (1, 4)
    };
    
    % \node[above] at (axis cs:1,6) {$4$};
    \node[below] at (axis cs:1,4) {$4$};
    
    % \node[above] at (axis cs:1,1) {$1$};
    \node[below] at (axis cs:1,1) {$1$};
    
    % \node[above] at (axis cs:3,3) {$2$};
    \node[below] at (axis cs:3,3) {$3$};

    % \node[above] at (axis cs:7,7.1) {$3$};
    \node[below] at (axis cs:7,7.1) {$7+\delta$};

    \addplot[dashed, domain=0:10] {x};
  \end{axis}
\end{tikzpicture}
}
\begin{figure}[h]
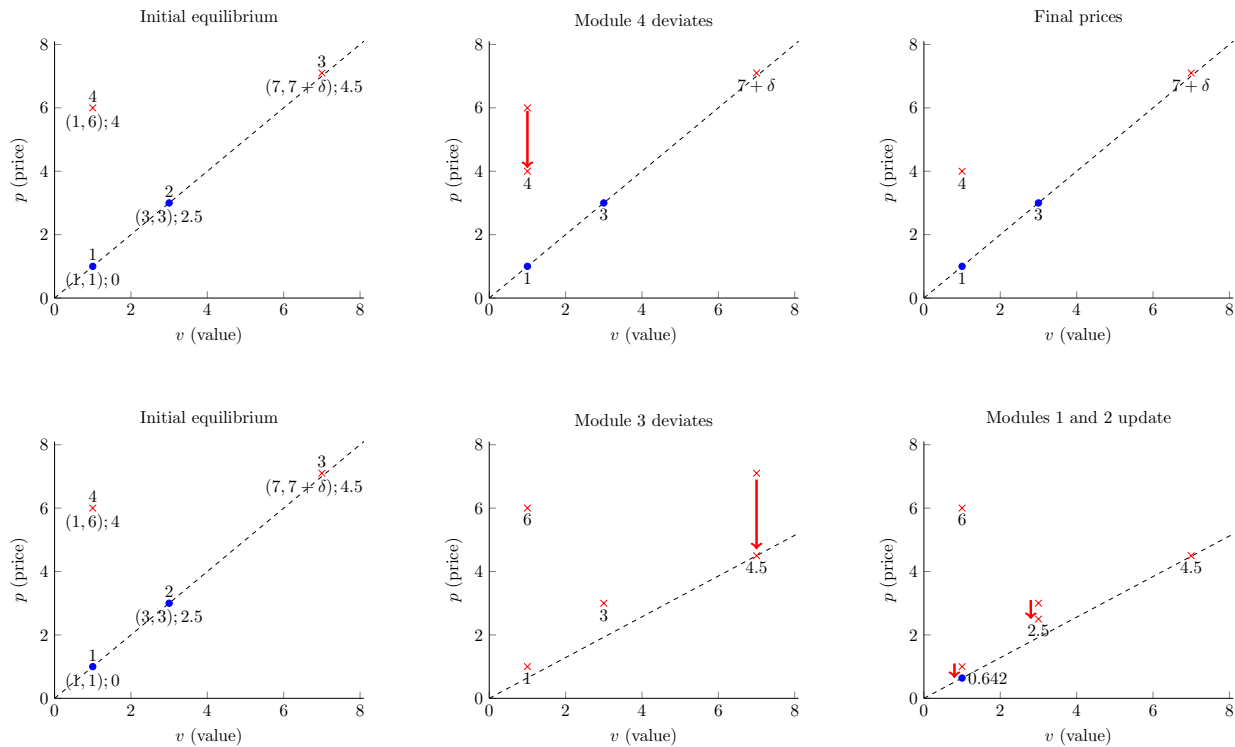

\centering
\begin{subfigure}{0.3\textwidth}
  \centering
  \plotAdditiveEquilibriumBOne{\tikzscale}
\end{subfigure}%
\hfill
\begin{subfigure}{0.3\textwidth}
  \centering
  \plotAdditiveEquilibriumBTwo{\tikzscale}
\end{subfigure}%
\hfill
\begin{subfigure}{0.3\textwidth}
  \centering
  \plotAdditiveEquilibriumBThree{\tikzscale}
\end{subfigure}%

\medskip

\begin{subfigure}{0.3\textwidth}
  \centering
  \plotAdditiveEquilibriumOne{\tikzscale}
\end{subfigure}%
\hfill
\begin{subfigure}{0.3\textwidth}
  \centering
  \plotAdditiveEquilibriumTwo{\tikzscale}
\end{subfigure}%
\hfill
\begin{subfigure}{0.3\textwidth}
  \centering
  \plotAdditiveEquilibriumThree{\tikzscale}
\end{subfigure}%
\caption{
This figure gives a pictorial sketch of how we go from an equilibrium where some rejected modules may be bidding above cost to a potentially worse equilibrium where strictly fewer rejected modules are bidding above cost.
}
\label{fig:worst_equil_sketch}
\end{figure}

\paragraph{Computing the Approximation Ratio.}
With the characterization above, we show the approximation result in two steps.
First, we show that running $\greedy$ on the price vector $\bar{\vecp}$ described above results in a constant approximation of what is achievable when running $\greedy$ on the price vector $\vecc$.
Second, we then relate the performance of $\greedy$ on the price vector $\vecc$ with the optimal solution given the price vector $\vecc$.

In particular, we show that if $\max_i \vecc(i)$ is small compared to the budget then the former is a $2$-approximation.
It is well-known that under this ``small-cost'' assumption that $\greedy$ gets very close to the optimal solution.
Thus overall, it shows that if we use $\greedy$ then any equilibrium is roughly a $2$-approximation.
As discussed in Subsection~\ref{subsec:model_and_results} above, we parameterize our results in terms of $\lambda = \max_i \vecc(i) / B$.

\subsection{Convergence to Equilibria}
Finally, we show that natural dynamics lead to convergence to an equilibrium in this pricing game.
In particular, we focus on the multiplicative weights update algorithm.   More formally, for any price competition game instance $\instance$ with $\I$ being a matroid constraint and budget $B=1$, we show that our dynamics convergences to an equilibrium that can be constructed via Algorithm~\ref{alg:equilibrium_dynamics_weighted}.

To begin with, we let $\barp$ be the equilibrium price computed by Algorithm~\ref{alg:equilibrium_dynamics_weighted} and $\SE$ be the set of selected modules at price $\barp$. For simplicity, we assume that $\barp(\SE) = 1$. We first, prove the following structural properties of the computed equilibrium price:
\paragraph{Weak-Dominance of $\barp$: }We show that if module $i\in \SE$ gets selected at some price vector $\vec p$ then it also gets selected at price $\barp(i)$. We emphasize that setting a price equal to $\barp(i)$ is not a dominating strategy as bidding higher than $\barp(i)$ can lead to higher utility.
\paragraph{Iterative Peeling of the Worst-Bid: }We show that given any price vector $\vec p$ satisfying $\forall i\in \SE: \vecp(i) \geq \barp(i)$ then the module with the module with worst-bang-per-buck at price $\vec p$ from the set $\SE$ does not get selected.
\paragraph{Stability of $\barp$: }Finally, we show that if all modules $i\in \SE$ set their price $\vec p(i) \in [\barp(i) - \delta , \barp(i) + \delta]$ then modules $i\in \SE$ can always deviate to price $\barp(i)$ and can get selected. In addition, no module $i\in \SE$ can get selected if they set their price significantly higher than $\barp$.    
Above, the second and the third property follows from the fact that $\barp(\SE) = 1$. When $\barp(\SE) < 1$, we prove such properties by leveraging the existence of a module that can swap module from $\SE$ that sets price greater than $\barp(i)$. The first property is slightly tricky and relies on the invariant of the equilibrium construction algorithm (Algorithm~\ref{alg:equilibrium_dynamics_weighted}) proved in Claim~\ref{claim:weighted_invarient}.

We assume that modules only place bids in a discretization of the unit interval, i.e.~$\{\delta, 2\delta, \ldots, 1\}$ for some small $\delta > 0$.
Due to Property~1, we first observe that for module $i\in \SE$, its cumulative utility obtained by any price $p <\barp(i)$ is strictly smaller than the cumulative utility of price $\barp(i)$. To steadily increase the difference in the cumulative utility $\sum_{t'\leq t} u_i(\barp(i),\vec p_{-i}^t) - u_i(p,\vec p_{-i}^t)$, we slightly distort the payments by giving extra $\delta^2\cdot p$ for any submitted price $p$. This essentially implies that after a sufficiently large number of rounds $\geq \poly \left( n, \frac{1}{\delta}\right)$, module $i \in \SE$ will stop submitting price sufficiently smaller than $\barp(i)$ with high probability. 

Next, we then condition on the event that all modules $i\in \SE$ sets their price at least $\barp(i) - O(\delta)$. Then we utilize the second property to iteratively rule out modules bidding that would lead to the worst bang-per-buck. More formally,  we iteratively define a pair of modules and their bid and order them from the worst to the best bang-per-buck, 
$$(  b_{(k)},   i_{(k)}): = \argmin_{\{(b,i)\in \bar{ \mathcal B}  \setminus \{(   b_{(1)},   i_{(1)}),\dots ,(   b_{(k-1)},   i_{(k-1)}) \} \}} \left \{v(i) / b \right \}.$$
Here, $(   b_{(1)},   i_{(1)}): = \argmin_{(b,i) \in \bar{\mathcal B}} \{ \vec v(i) / b \}$ and $\bar{\mathcal B}:=\bigcup_{i\in \SE'} \{(\barp(i) + \delta, i) , \dots , (1,i) \}$. We observe that once we condition on the event that all modules $i\in \SE$ sets their price at least $\barp(i) - O(\delta)$, we observe that module $i_{(1)}$ can not be selected at price $b_{(1)}$. Therefore, the gap in the cumulative utility of $\sum_{t'\leq t} u_{i_{(1)}}(\barp(i_{(1)}),\vec p_{-i_{(1)}}^t) - u_{i_{(1)}}(b_{(1)},\vec p_{-i_{(1)}}^t)$ keeps on growing. Again, to accelerate the convergence, we start incentivizing modules to lower their bid by paying extra $\frac{\delta^4}{b}$ for bid $b$.  
This essentially implies that after $\poly \left( n, \frac{1}{\delta}\right)$ many rounds, module $i_{(1)} \in \SE$ will stop submitting price $b_{(1)}$ in future with sufficiently high probability using the property of multiplicative weight update algorithm. This argument can be repeated until we are at convergence. Once modules in $\SE$ start submitting prices closer to their equilibrium price, Property 3 ensures that no module has an incentive to deviate from their convergent prices which are their constructed equilibrium prices.

\section{Warm-up: Price Competition for Additive Buyer}\label{sec:additive_price_equilibrium}

We first present the equilibrium quality analysis for the case when the buyer's value is an additive function over the set of modules.
We consider the greedy algorithm that is presented in Algorithm~\ref{alg:knapsack} which is just Algorithm~\ref{alg:greedy} for the unconstrained case, i.e.~$\I = 2^M$.
The algorithm first sorts all modules in terms of their ``bang-per-buck'' $\vecv(i) / \vecp(i)$ where $\vecv(i)$ is the value of module $i$ and $\vecp(i)$ is the submitted bid by module $i$.
It then selects modules in decreasing bang-per-buck order and exits as soon as it finds a module which does not fit into the budget without taking that module.

\begin{algorithm}
\caption{Greedy Algorithm for Knapsack Problem}
\label{alg:knapsack}
\begin{algorithmic}[1]
\State \textbf{Input}: price vector $\vec p$, values $\vec v$, budget $B$, \textbf{Output}: set $G$ of modules.
\State Initialize $G \gets \emptyset$.
\State Re-arrange modules such that $\frac{\vecv(1)}{\vecp(1)} \geq \frac{\vecv(2)}{\vecp(2)} \geq \dots \geq \frac{\vecv(n)}{\vecp(n)}$ with ties broken arbitrarily.
\For{$i=1,2 \dots ,n$}
\If{$\vec p(G) + \vec p(i) > B$}
\State \textbf{break}
\EndIf
\State $G \gets G \cup \{i\}$
\EndFor
\State \textbf{return} $G$
\end{algorithmic}
\end{algorithm}

In Subsection~\ref{subsec:eq_exists}, we show that $\eps$-equilibria exist when the buyer uses $\greedy$.
Note that some slack in the equilibrium notion is required.
In Subsection~\ref{subsec:need_eps_eq}, we show that exact equilibria can be very brittle and in fact, it may not exist.

In Subsection~\ref{subsec:eq_quality}, we then analyze the quality of all $\eps$-equilibria.
We show that all $\eps$-equilibria obtain at least half of the optimal solution when costs are known provided that (i) the cost of all modules relative to the budget is small and (ii) $\eps$ is small relatively to the cheapest module.
The first condition is similar to the ``small-bid'' assumption that appears in the literature (e.g.~\cite{mehta2007adwords}).
The second condition turns out to be necessary to prove an approximation result since there could be many very cheap and low value modules, which in aggregate, could contribute a significant amount of value.
Details of this necessity can be found in Subsection~\ref{subsec:need_cost_lb}.

\subsection{Existence of Equilibrium}
\label{subsec:eq_exists}
In this section, we prove the existence of an $\eps$-equilibrium.
To slightly simplify the exposition, we assume that the ratio $\frac{\vecv(i)}{\vecc(i)}$ is distinct for all modules.
If this is not the case we can perturb the value of each module by adding an independent random value drawn uniformly from the interval $[0, \eps]$ so that the aforementioned ratio is almost surely distinct for all modules.
In the argument below, this means that each module may actually be able to increase its utility by at most $\eps$ more so the argument below finds a $2\eps$-equilibrium instead of an $\eps$-equilibrium.
Henceforth, we assume $\frac{\vecv(i)}{\vecc(i)}$ are all distinct.
To simplify notation, in this section we sort modules in decreasing order of the bang-per-buck value w.r.t.~their cost, i.e.~$\frac{\vec v(1)}{\vec c(1)} > \frac{\vec v(2)}{\vec c(2)} > \dots > \frac{\vec v(n)}{\vec c(n)}$.
Let
\begin{equation}
    \label{eqn:cpv_additive}
    \cpv^* = \sup\left\{ \cpv \in \R \,:\, \sum_{i=1}^n \vecv(i) \cdot \cpv \cdot \indicator[\vecc(i) \leq \vecv(i) \cdot \cpv] \leq B \right\}.
\end{equation}
The quantity $\cpv^*$ is the smallest cost-per-value such that all modules that are taking all modules that are willing to accept the price given by this cost-per-value already exceeds our budget.
We claim that $\cpv^*$ essentially defines the equilibrium price.
\begin{theorem}
\label{theorem:additive_eq_exists}
Let $\cpv^*$ be as defined in Eq.~\eqref{eqn:cpv_additive}.
For any $\eps > 0$, there exists $\delta \geq 0$ such that the price vector $\vecp$ defined as $\vecp(i) \coloneqq \max\{ (\cpv^* - \delta) \cdot \vecv(i), \vecc(i) \}$ is an $\eps$-equilibrium.
\end{theorem}
\begin{proof}
Let $\cpv_i = \frac{\vecc(i)}{\vecv(i)}$.
Recall that we assume $\cpv_1 < \ldots < \cpv_n$.
We have two cases.
\paragraph{Case 1: $\cpv^* = \cpv_k$ for some $k \in [n]$.}
In this case, we choose $\delta$ satisfying (i) $\delta > 0$, (ii) $\delta < \cpv^* - \cpv_{k-1}$, and (iii) $\delta \vecv(i) \leq \eps / 2n$ for all $i \in [n]$.
Note that, by (ii), $\vecp(i) > \vecc(i)$ for all $i < k$ and $\vecp(i) = \vecc(i)$ for all $i \geq k$.
By definition of $\cpv^*$, $\vecp([k-1]) \leq B$ so modules $1$ through $k-1$ are accepted.

We need to show that if any module raises their bid by $\eps$ then it is rejected.
Let $i \in [n]$ and suppose module $i$ deviates to $\vecp'(i) = \vecp(i) + \eps$.
Note that
\[
    \vecp(i) + \eps \geq \cpv_k \cdot \vecv(i) - \delta \cdot \vecv(i) + \eps > \cpv_k \cdot \vecv(i)
\]
because $\delta \vecv(i) \leq \eps / 2n < \eps$.
So module $i$, with its deviation, comes after module $k$ (or is module $k$) in the cost-per-value ordering.
In that case, the budget utilization by $\greedy$ up to and including $i$ is at least
\begin{align*}
\vecp([k]) \geq \cpv_k \vec([k]) - \delta \vecv([k]) + \eps
\geq B - \eps / 2 + \eps > B.
\end{align*}
In particular, module $i$ is not accepted.
Thus any module $i < k$ cannot increase its bid by at least $\eps$ while being accepted and any module $i \geq k$ cannot choose a bid to achieve utility at least $\eps$.

\paragraph{Case 2: $\cpv^* \neq \cpv_k$ for all $k \in [n]$.}
Let $k = \max\{i \,:\, \cpv_i < \cpv^*\}$.
By definition of $\cpv^*$ and $k$, this means that for the price vector $\vecp$ given by $\vecp(i) = \max\{\cpv^* \vecv(i), \vecc(i)\}$, we have that $\greedy$ picks modules $1$ through $k$ and $\vecp([k]) = B$.
No accepted module $i \leq k$ can increase its bid by more than $\eps$ because doing so would make it go after $[k] \setminus \{i\}$ and including it would exceed the budget constraint.
On the other hand, no rejected module $i > k$ can decrease its bid because it is already bidding its cost.
We thus conclude that $\vecp$ is already an equilibrium price vector.
\end{proof}

\newcommand{\additiveiterstikzscale}{0.6}

\subsection{Equilibrium Quality}
\label{subsec:eq_quality}
Lemma~\ref{lemma:worst_equil} is the main structural result needed in this section which, at a high level, says that the worst equilibrium is achieved when rejected modules are bidding near their cost.
The high-level idea of the proof is illustrated in Figure~\ref{fig:worst_equil_sketch} which illustrates how rejected modules may lower their price to cause other modules to be rejected.
The proof is relegated to Appendix~\ref{app:proof_worst_equil}.
\begin{lemma}
\label{lemma:worst_equil}
Let $\vec p$ be a vector of $\eps$-equilibrium prices and let $S_{\vec p}$ be the set of modules selected by $\greedy$.
Fix an arbitrary $i \notin S_{\vec p}$.
There exists $\eps$-equilibrium prices $\vec{\bar{p}}$ such that (i) $\vec{\bar{p}}(i) \leq \vec{c}(i) + \eps$, (ii) $S_{\vec{\bar{p}}} \subseteq S_{\vec{p}}$, and (iii) $\bar{\vecp} \leq \vecp$ (coordinate-wise).
\end{lemma}
The following lemma shows that running $\greedy$ on any equilibrium yields roughly half of the value compared to running $\greedy$ on an instance where the costs are known. The proof is delegated to Appendix~\ref{proof_eps_additive}
\begin{lemma}\label{lem:additive_eps_eq_deviation}
Let $0 < m < M \leq B$ and $\lambda = M / B$.
Suppose that the cost of every module is in $[m, M-\eps]$.
Let $\vecp$ be any $\eps$-equilibrium and let $S_{\vecp}$ (resp.~$S_{\vecc}$) be the set returned when using $\greedy$ on $\vecp$ (resp.~$\vecc$).
Then $\frac{\vecv(S_{\vecc})}{\vecv(S_{\vecp})} \leq 2 + \frac{\eps}{m} + \left( 1 + \frac{\eps}{m} \right) \cdot \frac{\lambda}{1-\lambda}$.
\end{lemma}
In particular, if the largest cost $M$ is small compared to the budget then as $\eps \to 0$, the solution obtained in an $\eps$-equilibrium becomes a $2$-approximation to the solution obtained via $\greedy$ when the costs are known a priori.
Next, we show that $S_{\vecc}$ is itself a good approximation to $S_{\OPT}$.
The proof is standard and relegated to Appendix~\ref{app:additive_greedy_approx_with_cost}.
\begin{lemma}
\label{lemma:additive_greedy_approx_with_cost}
Let $S_{\OPT} \in \argmax \{ \vecv(S) \,:\, \vecc(S) \leq B \}$ be an optimal bundle.
Suppose that $\vecc(i) \leq M$ for every $i \in [n]$.
Let $\lambda = M / B$.
Then $\frac{\vecv(S_{\OPT})}{\vecv(S_{\vecc})} \leq 1 + \frac{\lambda}{1-\lambda}$.
\end{lemma}

\begin{theorem}
\label{thm:additive_approx}
Let $0 < m < M \leq B$ and let $\lambda = M / B$.
Suppose that the cost of every module is in $[m, M-\eps]$.
Let $\vecp$ be any $\eps$-equilibrium and let $S_{\vecp}$ be the set returned when using $\greedy$ on $\vecp$.
Let $S_{\OPT} \in \argmax \{ \vecv(S) \,:\, \vecc(S) \leq B \}$ be an optimal bundle.
Then
\[
    \frac{\vecv(S_{\OPT})}{\vecv(S_{\vecp})}
    \leq
    \left(2 + \frac{\eps}{m} + \left( 1 + \frac{\eps}{m} \right) \cdot \frac{\lambda}{1-\lambda}\right) \cdot \left(1 + \frac{\lambda}{1-\lambda} \right).
\]
\end{theorem}
As $\eps \to 0$, the above ratio becomes $\frac{2-\lambda}{(1-\lambda)^2}$.
Thus as $\lambda \to 0$ (i.e.~the cost of every module becomes small relative to the budget), we get a $2$-approximation.
\section{Price Competition for Buyer with Matroid Constraints}
\label{sec:weighted_matroid}
In this paper, we consider two extensions of the price competition game to buyers with matroid constraints.
In the first extension, we focus on the setting where the buyer's value $\vecv(i) = 1$ for all modules $i$.
In this case, we consider a natural generalization of Algorithm~\ref{alg:knapsack} in that a module is automatically rejected if it is infeasible in the current set selected thus far (see Algorithm~\ref{alg:modified_greedy}).
Note that this is not the same as Algorithm~\ref{alg:greedy} since Algorithm~\ref{alg:greedy} allows modules that come later in the bang-per-buck order to ``replace'' modules that have already been accepted.
However, for this setting we are able to prove that an $\eps$-equilibrium always exists (Theorem~\ref{thm:unweighted_eq_exist}) and and we obtain a $(2-\lambda) / (1-\lambda)^2$-approximation where $\lambda = \max_i \vecc(i) / B$ (Theorem~\ref{theorem:unweighted_approx}).
All the details are relegated to Appendix~\ref{app:price_comp_matroid}.
In the rest of the section, we consider the pricing game when the platform uses Algorithm~\ref{alg:greedy} for module selection for any matroid and value vector. 

\subsection{Existence of Equilibria}
To prove the existence of equilibria, we consider Algorithm~\ref{alg:equilibrium_dynamics_weighted} that construct an equilibria for any given price competition instance with a matroid constraint. The algorithm begins by initializing $\vecp$ as $\vec p = \vec c$ and an ordering $\pi^0$ which is decreasing in their bang-per-buck w.r.t.~$\vecp$ at round zero. We then define the set $A^\ind$ that denotes the set of modules that increase their price at iteration $\ind$ and $\pi^\ind$ that denotes the bang-per-buck order over modules at the price $\vec p^\ind$. 

We can describe the algorithm as follows.
First, we raise the prices of all modules in $A^{k-1}$ so that their bang-per-buck meets that of module $\pi^k(k)$, whose price is still at its cost.
Next, we check whether or not $\pi^k(k) \in \spn(A^{k-1} \cup T^{k-1})$.
If not we add it to $A^{k-1}$ to obtain $A^k$.
Otherwise, we find the (unique) circuit in $A^{k-1} \cup  \{\pi^k(k)\}$. Then we remove the lowest valued module from $C^k$ and replace it with the module $\pi^0(k)$. 

Finally, if adding $\pi^0(k)$ to the set $A^k$ increases the total price to exceed the budget then we remove the module $\pi^0(k)$ from the set $A^k$ and adjust the price of all the modules in $A^{k-1}$ such that it satisfies the budget constraint and no module in $A^{k-1}$ has a profitable deviation. Throughout the algorithm, we maintain the following invariants at iteration $\ind$ whose proof is delegated to Appendix~\ref{proof_of_weighted_invariant} 
\begin{claim}\label{claim:weighted_invarient}
For any iteration $k\leq k^*$ of Algorithm ~\ref{alg:equilibrium_dynamics_weighted}, the following invariants holds:
\begin{enumerate}
    \item $\pi^{k-1}([k-1]) = \pi^k([k-1])$ and $\pi^{k-1}(k) = \pi^{k}(k)$.
    \item when the prices are $\vecp^k$, the bang-per-buck of any module in $A^{k}$ is at least that of any module in $\pi^k([n]) \setminus \pi^k([k])$.
    \item $\vecp^k(A^{k}) \leq B$.
    \item for $k<k^*$, $\rank(A^k) = \rank(\pi([k]))$. 
    \item if $i \leq k<k^*$ and $\pi^{k}(i) \in \pi[k] \setminus A^k$ then $\pi^k(i)$ forms a circuit $C$ with $A^k$ where $\pi^k(i)$ has the lowest value on $C$.
    \item $S_{\bar{\vecp}} = A^{k^*}$. In addition, $\SE$ is a maximum weight independent set in matroid $\mathcal I$ restricted to $\pi^0(1), \dots , \pi^0(k^*)$
\end{enumerate}
\end{claim}
\label{sec:weighted_matroid_equilibrium}
\begin{algorithm}
\caption{Algorithm to Compute Equilibrium for Weighted Matroid}
\label{alg:equilibrium_dynamics_weighted} 
\begin{algorithmic}[1]
\State Initialize price $\vecp^0 \gets \vecc$ and $\pi^0$ such that $\frac{\vecv(\pi^0(i))}{\vecc(\pi^0(i))}\geq \frac{\vecv(\pi^0(j))}{\vecc(\pi^0(j))}$ for $i < j$.
\State Initialize $A^0 \gets \emptyset$.
\For{$\ind = 1, \ldots, n$}

\State Copy $\pi^{\ind-1}, \vec p^{\ind-1}$ into $\pi,\vec p$ for simplicity of notation.
\State Update prices: $\vecp^k(\pi(j)) = \frac{\vecp(\pi(k))}{\vecv(\pi(k))} \cdot \vecv(\pi(j))$ if $j \in A^{k-1} $, otherwise, $ \vecp^k(\pi(j)) = \vecp(\pi(j))$.
\If{$(\pi(\ind) \notin \spn( A^{k-1})$)} \label{line:if_weighted}
    \Comment{Accept $\pi(k)$.}
    \State $A^k \gets A^{k-1} \cup \{\pi(k)\}$
\Else
    \State Let $C^k$ be the unique circuit in $A^{k-1} \cup \{\pi(k)\}$.
    \State Let module $i \in C^k$ has the lowest value.
    \State $A^k \gets (A^{k-1} \cup \pi(k)) \setminus \{i\}$
\EndIf

\State Set $\pi^k$ as bang-per-buck ordering with respect to $\vecp^k$, breaking ties according to $\pi^0$.
\If{$\vecp^k(A^k) > B$} \label{line:budget_check_weighted}
\Comment{Reject last module, update prices, and terminate.}
\State  $A^k \gets A^{k-1}$.
\State Update $\vecp^k(i) \gets \min\left\{\vecv(i) \cdot \frac{B - \vecp^{k-1}(A^{k-1})}{\vecv(A^{k-1})}, \vecp^k(i) \right\}$ for $i \in A^{k-1}$. 
\State \textbf{return} $\vecp^k, \pi^k$
\EndIf
\EndFor

\State Update $\vecp^n(i) \gets \vecv(i) \cdot \frac{B - \vecp^n(A^{n})}{\vecv(T^n)}$ for $i \in A^n$.
\State \textbf{return} $\vecp^n, \pi^n$
\end{algorithmic}
\end{algorithm}
The above invariant leads to the main theorem of this sub section whose proof is delegated to Appendix~\ref{proof_weighted_matroid_exists}.
\begin{theorem}
\label{thm:weighted_matroid_eq_exists}
 Let $\vec {\bar p}$ be the prices computed by Algorithm~\ref{alg:equilibrium_dynamics_weighted} then $\bar{\vecp}$ are equilibrium prices given that ties are broken in favor of modules with higher $\frac{\vec v(\cdot )}{\vec c(\cdot)}$. 
 \end{theorem}
 
\subsection{Quality of Equilibria}
The following is the main theorem of this section. 
\begin{theorem}
\label{thm:weighted_matroid_approx}
Let $\lambda = \max_i \vecc(i) / B$ and $S_{\OPT} \in \argmax \{ S \,:\, \vecc(S) \leq B, S \in \I \}$.
For any $\eps$-equilibrium $\vecp$, we have
\[
    \vecv(S_{\OPT})
    \leq \left( 1 + \frac{2\lambda}{1-3\lambda} + \frac{1+\eps / \lambda B}{1-\lambda - \eps / B} \right)
    \left( 1 + \frac{\lambda}{1-\lambda} \right) \cdot \vecv(S_{\vecp}).
\]
\end{theorem}
In particular, as $\eps, \lambda \to 0$, the approximation ratio approaches $2$. To prove the above theorem, we first prove the following Lemma (proof delegated to Appendix~\ref{proof_of_lemma_poa_high_bpb}) that ensures that for any $\eps$-equilibrium price $\vec p$, there exists an equilibrium price such that all non-selected module with relatively high bang-per-buck sets their price equals to their cost. 

\begin{lemma}
\label{lemma:unselected_lt_k_bid_cost}
Let $\vecp$ be an $\eps$-equilibrium price.
Let $\frac{\vecv(1)}{\vecp(1)} \geq \ldots \geq \frac{\vecv(n)}{\vecp(n)}$ be the bang-per-buck order at $\vecp$ and $S_{\vec p}$ be the set of modules selected by $\greedy$.
If module $i < k^*$ is not selected then there exists an equilibrium price $\bar{\vecp}$ such that (i) $\bar{\vecp} \leq \vecp$, (ii) $\bar{\vecp}(i) \leq \vecc(i) + \eps$, and (iii) $\vecv(S_{\bar{\vecp}}) = \vecv(S_{\vecp})$.
\end{lemma}
Next, we state the key technical lemma that ensures that given any $\eps$-equilibrium price $\vec p$, there exists a non-selected module $q$ at price $\vec p$ with relatively lower bang-per-buck and equilibrium price $\bar p$ such that $\bar p(q) \leq \vec c(q)+\eps$, the price of the rest of modules $\bar p(i)\leq \vec p(i)$, and $v(S_{\vec{\bar p}})\leq v(S_{\vec p})$. The proof of the lemma is technical and delegated to Appendix~\ref{proof_of_key_lemma_poa}.
\begin{lemma}
\label{lemma:unselected_gt_k_to_cost}
Let $\vecp$ be an $\eps$-equilibrium.
Let $\frac{\vec v(1)}{\vec p(1)}\geq \dots \geq \frac{\vec v(n)}{\vec p(n)}$ be the bang-per-buck order at price $\vec p$ and $S_{\vec p}$ be the set of selected modules by $\greedy$ such that $k^*$ be the module with the worst bang-per-buck.
Suppose that $n \geq k^* + 2$ and that if $i < k^*$ is rejected then $\vecp(i) \leq \vecc(i) + \eps$.
Let $q$ be the largest index of an un-selected module such that $\vecp(q) > \vecc(q) + \eps$.
There exists equilibrium prices such that (i) $\bar{\vecp}(i) \leq \vecp(i)$ for $i \geq k^* + 1$, (ii) $\bar{\vecp}(q) \leq \vecc(q) + \eps$, (iii) $\bar{\vecp}(i) \leq \vecc(i) + \eps$ for un-selected modules $i < k^*$, and (iv) $\vecv(S_{\bar{\vecp}}) \leq \vecv(S_{\vecp})$.
\end{lemma}
The above lemma allows us to only focus on analysing the value of equilibria where all non-selected modules are setting their price within $\eps$ additive factor of their cost. This characterizes all the bad equilibria of the pricing game and allows us to complete the proof of the main theorem in Appendix~\ref{proof_of_quality_matroid}.  
\section{Convergence Dynamics of Learning Modules}
\label{sec:convergence_dynamics_of_learning}
Our main result in this section shows that the learning dynamics of the price competition game $\instance$ with multiplicative weight learning algorithm $\{\mathcal L(i):i\in M\}$ converges to the price $\bar \vecp$ computed by Algorithm~\ref{alg:equilibrium_dynamics} under mild assumption on $\vec c, \vec v$ and disretization $\delta$, when $\mathcal I$ is a matroid constraints and slightly distorted payments rule: for initial rounds  $t\leq T_0$, each selected module $i\in M$ gets payment of $\vec p^t(i) + \delta^2\cdot \vec p^t(i)$  and the rest of the modules get payment of $\delta^2 \cdot \vec p^t(i)$ and later after round $t > T_0$, all selected module $m_i\in M$ gets payment of $\vec p^t(i) + \frac{\delta^4}{ \vec p^t(i)}$  and the rest of the modules get payment of $\frac{\delta^4}{ \vec p^t(i)}$. 

The main reason for the distorted payment is the following:  note that at the equilibrium price $\vec{\bar p}$ computed by Algorithm~\ref{alg:equilibrium_dynamics_weighted}, all the modules which are selected at equilibrium prices $\SE$ have identical bang-per-buck $\optbpb$ while the not selected with $\frac{\vec v(i)}{\vec c(i)} < \optbpb$ at price $\bar p$ sets their price equals to their cost. Therefore to speed up the convergence, we initially incentivize modules in $\SE$ to bid higher, and later once modules in $\SE$ start bidding higher than their equilibrium prices, we then incentivize these modules to lower their price if they are bidding much higher than their equilibrium price.  As a result,  we distort the payment rule by an additive factor of $O(\delta^3)$. It is easy to observe that if $\vec S_\vec p$ is constructed via a greedy algorithm with budget $B-\delta$, and $\vec p$ is an $\delta$-equilibrium price then the price vector $\vec p$ is a $(\delta + \delta^2)$-equilibrium price with the new payment rule.

The following is the main theorem of this section. 
\begin{theorem}\label{thm:matroid_convergence}
Given the price competition instance $\instance$ with matroid constraint and one unit of budget, if the platform implements a greedy algorithm (Algorithm~\ref{alg:greedy}) for module selection with distorted payment rule and each module, unaware of market parameters, employs a multiplicative weight update-style price learning algorithm, then the induced price dynamics converges to the price $\bar \vecp$ computed by Algorithm~\ref{alg:equilibrium_dynamics_weighted}. More precisely, 
each module $i\in M$ submits their price $\vec p^T(i) \in [ \bar {\vec {p}}(i) -\sqrt \delta  , \bar {\vec p}(i) + \sqrt \delta]$ with probability at least $1 - \exp \left( - \poly \left(\frac 1 \delta  \right)\right)$ for $T \geq \Omega \left( \poly \left(n , \frac 1 \delta \right)\right)$, where $\bar {\vec p}$ is the equilibrium prices computed by Algorithm~\ref{alg:equilibrium_dynamics_weighted}.
\end{theorem}

\subsection*{Structural Properties of Algorithm~\ref{alg:equilibrium_dynamics_weighted} and Implications}
Before we start proving the main theorem, we make useful structural observations about the equilibrium prices computed by Algorithm~\ref{alg:equilibrium_dynamics_weighted} which will be the main ingredient to show convergence of learning dynamics to equilibrium prices. 
Throughout the section, we let $L := \{  i\in M\setminus \SE: \frac{\vec v(i)}{\vec c(i)} < \optbpb \}$ and $H:= \{  i\in M\setminus \SE: \frac{\vec v(i)}{\vec c(i)} \geq  \optbpb \}$. The following structural lemmas are the key ingredients of the convergence analysis. First, we show that if some module $i\in \SE$ gets selected at price vector $\vec p$ with $\vec p(i)< \barp (i)$ then it can also be selected at price $(\barp(i), \vec p_{-i})$. This implies that each module prefers bidding $\barp(i)$ rather than some price smaller than $\barp(i)$. Throughout this section, we let $\Delta_i = \barp(i) - 10\cdot \delta$ and assume that $\Delta_i <1 - 10 \cdot \delta$.

\begin{lemma}\label{lem:eqm_price_dominates}
Let $\barp$ be the equilibrium price computed by Algorithm~\ref{alg:equilibrium_dynamics_weighted} and $\SE$ be the set of the selected module at price $\barp$. Then if module $i$ is selected by the greedy algorithm at price $\vec p$ then module $i$ also gets selected by the greedy algorithm at price $\vec p' = (\barp(i), \vec p_{-i})$.
\end{lemma}

Next, we show that if all the modules in the set $\SE$ set their prices larger than their equilibrium prices then the module with the worst bang-per-buck ratio does not get selected by the greedy algorithm. This lemma essentially puts ``down pressure" on the module's prices once they start bidding larger than their equilibrium prices. 

\begin{lemma}\label{lem:module_with_worst_bpb_rejected}
Let $\barp$ be the equilibrium price computed by Algorithm~\ref{alg:equilibrium_dynamics_weighted} and $\SE$ be the set of the selected module at price $\barp$. We consider set $\SE':=\SE \cup \{\pi^0(k^*)\}$ where $k^*$ is the last iteration of Algorithm~\ref{alg:equilibrium_dynamics_weighted}. For any price vector $\vec p$ satisfying $\vec p(i) \geq  \barp(i)$  for all $i\in \SE$, we let $S_\vec p$ be the selected module at price $\vec p$. Then for module $i^*:=\argmin \left\{i\in \SE: \frac{\vec v(i)}{\vec p(i)} > \optbpb \right\}$, we have $i^* \notin S_\vec p$. In addition, we have $\vec p(i^* \cup S_\vec p)> B$. 
\end{lemma}

Next, finally, we show that if all modules bid close to the equilibrium price then no module can deviate significantly to improve their utility.  

\begin{lemma}\label{lem:stability_of_eqm_price}
Let $\barp$ be the equilibrium price computed by Algorithm~\ref{alg:equilibrium_dynamics_weighted} and $\SE$ be the set of the selected module at price $\barp$. We consider set $\SE':=\SE \cup \{\pi^0(k^*)\}$ where $k^*$ is the last iteration of Algorithm~\ref{alg:equilibrium_dynamics_weighted}. For any price vector $\vec p$ satisfying $\barp (i) + 10\cdot \delta \geq \vec p(i) \geq  \barp(i) - 10\cdot \delta$ for all $i\in \SE'$, we let $S_\vec p$ be the selected module at price $\vec p$. Then any module $i\in \SE$ gets selected at price $\barp(i)$ and does not get selected at price $\Delta_i + \sqrt \delta > \barp(i) \cdot (1 +\sqrt \delta)$. 
\end{lemma}

\subsection*{Proof Sketch of Theorem~\ref{thm:matroid_convergence}}
We first leverage Lemma~\ref{lem:eqm_price_dominates} to define an event stating that after a sufficiently large number of iterations, all modules in $\SE$ start setting their price higher than their equilibrium prices.  More formally, we let $T_0$ be sufficiently large the round (determined later) and define $\mathcal E^0 = \{ \forall t\geq T_0:\vec p^t(i) > \Delta_i,  \forall i \in \SE  \}$.

We first observe that given any price vector $\vec p$ and conditioned on the event $\mathcal E^0$, the module with the worst bang-per-buck in the set $\SE'$ defined as $\SE' :=\SE \cup \{\pi^0(k^*)\}$ where $k^*$ is the last iteration of Algorithm~\ref{alg:equilibrium_dynamics_weighted}, will not be selected due to Lemma~\ref{lem:module_with_worst_bpb_rejected}.
Our overall proof approach is to show that conditioned on the event $\mathcal E^0$, the modules in $\SE'$ stop posting the prices that lead to smaller bang-per-buck for the platform since it will not be selected at such a higher price. In order to demonstrate that, we define an order over the possible prices w.r.t.~their bang-per-buck values for the buyer. Next, we define an order over the prices with respect to their bang-per-buck value. We iteratively define 
$$(  b_{(k)},   i_{(k)}): = \argmin_{\{(b,i)\in \bar{ \mathcal B}  \setminus \{(   b_{(1)},   i_{(1)}),\dots ,(   b_{(k-1)},   i_{(k-1)}) \} \}} \left \{v(i)/b \right \}.$$
Here, $(   b_{(1)},   i_{(1)}): = \argmin_{(b,i) \in \bar{\mathcal B}} \{v(i)/b \}$ and $\bar{\mathcal B}:=\bigcup_{i\in \SE'} \{(\Delta_i + \delta, i) , \dots , (1,i) \}$.  First, we prove the following claim that shows that module $i_{(1)}$ will never set price $b_{(1)}$ with high probability after $\poly(1/\delta)$ many rounds as it will never be selected.

Next, we extend this argument and iteratively show that the modules will stop posting prices larger than their respective $\Delta_i$'s. We let $T_1 < T_2 <\dots <T_K$ and $T^*_1 < T^*_2 <\dots <T^*_K$ such that $T_k \leq T_{k}^*$ for some $K < \frac{n}{\delta}$ and the events $$\mathcal E^k:= \{ \vecp^s(i(k)) < b_{(k)} \text{ for all } s \geq T_k\},$$
we have that the $T_k$'s and $T_k^*$'s satisfy the following conditions.
\begin{enumerate}
	\item We inductively define $T_k^*$ as follows. Suppose we are given $T_1,\dots , T_{k-1}$; $T^*_0,\dots , T^*_{k-1}$. Conditioned on the events $\mathcal E^1,\dots , \mathcal E^{k-1}$, we let $T_k^* > T^*_{k-1}$ be the smallest round (if it exists) such that, 
	\begin{equation}\label{condition1_weighted}
		\sum_{s\leq T^*_k} u_{i_{(k)},s}(\Delta_{i_{(k)}},\vec p^s_{-i_{(k)}}) - u_{i_{(k)},s}(b,\vec p^s_{-i_{(k)}}) \geq \delta^6 \cdot T_k^*
	\end{equation}
	for all $b \geq \Delta_{i(k)} + \delta$ (and in the discretization).
	\item Next, we define $T_k$ as follows. We let $T_k \leq T_k^*$ be the smallest stage such that there exists some
    $b_{i_{(k)}}' < b_{i_{(k)}}$
	for all $s \in (T_k,T_k^*]$, such that, 
	\begin{equation}\label{condition2_weighted}
		\sum_{t\leq s} u_{i_{(k)},t} (b'_{i_{(k)},t},\vec p^s_{-i_{(k)}}) - u_{i_{(k)},t}(b_{i_{(k)}}, \vec p^s_{-i_{(k)}}) \geq \frac{C}{\gamma_s \cdot \delta}.
	\end{equation}
\end{enumerate}
Above, $\mathcal E^k$ captures the event where the module $i_{(k)}$ sets their price with higher bang-per-buck than $k$-th lowest bang-per-buck in the set of bids. In addition, we can observe that by the definition of $T_k$, we have that the module $i_{(k)}$ begins setting the price lower than $b_{(k)}$ with high probability due to the existence of some lower bid $b$ which has significantly higher cumulative utility. Therefore, if we show an existence of bounded $T_k$'s then we can essentially show that any module $i\in \SE'$ starts setting price lower than their corresponding $\Delta_i + 2\cdot \delta$ after finitely many rounds while event $\mathcal E^0$ ensures that module $i\in \SE'$ is bidding higher than $\Delta_i$. 

Above, we note that if Condition~\eqref{condition1_weighted} is satisfied for some $k$ then for $T_k = T_k^*$, Condition~\eqref{condition2_weighted} is trivially satisfied. We
first prove an upper bound on $T_i$ which shows an existence of desired $T_i$s which is one of the most crucial steps in proving Theorem~\ref{thm:matroid_convergence}. We show that (Lemma~\ref{lem:bounding_iterations}) that $T_k \leq \poly \left (\frac{1}{\delta}, k \right)$. Finally, once we have that all the modules stop posting prices very far from their computed equilibrium prices, using Lemma~\ref{lem:stability_of_eqm_price} and standard probability calculations, we show that $\Pr\left [\bigcap_{k=1}^K \mathcal E^k\right] \geq 1 - O(\poly(1/\delta)) \cdot \exp(-1/\delta)$. This ensures that after $\poly(1/\delta)$ many rounds, all the sellers converges to their equilibrium price.

\newpage
\appendix
\section{Technical Assumptions}
\subsection{Necessity of \texorpdfstring{$\eps$}{eps}-equilibrium}
\label{subsec:need_eps_eq}
In this brief sub-section, we argue the necessity of studying $\eps$-equilibria instead of exact equilibrium.
Consider the following example with three modules $A, B, C$.
The values are $v_A = v_B = v_C = 1$ so that the greedy algorithm simply selects modules in order of prices and a tie-breaking rule.
The costs are $c_A = 2, c_B = 3, c_C = 4$ while the budget is $10$.
Suppose that ties are broken in favor of modules $C$ then $A$ then $B$.
We consider a few cases.

\paragraph{Case 1: both $A$ and $B$ bid exactly $4$.}
We consider two smaller subcases.
Suppose first that $C$ also bids $4$.
Then the winners are $C$ and $A$. But this means $B$ can deviate, to $3.5$ for example, to increase their utility.
On the other hand, if $C$ bids strictly more than $4$ then both $A$ and $B$ are selected.
However, both of these modules can slightly increase their bid to ensure they remain selected while being paid slightly more.
So there is no equilibrium where both $A$ and $B$ bid exactly $4$.

\paragraph{Case 2: at least one of $A$ and $B$ bid strictly more than $4$.}
In this case, any equilibrium must include $C$ as a winner because if $C$ is not selected then it could always choose some bid strictly above $4$ to be selected and thus obtain positive utility.
In general (not just in this case), any equilibrium must also include $A$ and at a price of at least $4-\eps$ because if not, it could just bid $4-\eps$ for any $\eps > 0$.
Since $C$ can never bid below $4$, the budget utilization when $A$ is being considered is at most $4-\eps$.
Thus, module $B$ must be rejected but this is not an equilibrium since $B$ can deviate to $3.5$ to ensure it is selected at strictly above its price.

\paragraph{Case 3: at least one of $A$ and $B$ bids strictly less than $4$.}
If either $A$ or $B$ bids strictly less than $4$ then it is certainly accepted since $C$ cannot bid less than $4$.
But if they do bid strictly less than $4$ then they can slightly raise their bid to gain slightly more utility.

\subsection{Necessity for a lower bound on cost}
\label{subsec:need_cost_lb}
Here, we show that if we analyze $\eps$-equilibria then it is necessary to assume a lower bound on the cost of each module.
Given any fixed $\eps > 0$, consider the following example where the budget is $1$.
There are $n+1$ modules all with cost $0$.
Module $1$ has value $1$ while all other modules have value $\eps / 2$.
The optimal solution is to take all modules for a total value of $1 + n\eps / 2$.
On the other hand, an $\eps$-equilibrium is for all modules to set a price of $1$.
Module $1$ is taken because it has the highest bang-per-buck.
For all the other modules, they need to set a bid of at most $\eps / 2$ in order to be selected.
Thus, they cannot deviate to gain more than $\eps / 2$ utility.
Note that this equilibrium receives only value $1$.
As $n \to \infty$, the efficiency of this $\eps$-equilibrium becomes arbitrarily worse.

\section{Matroid Theory Preliminaries} \label{appendix:matroid_prelims}
A \emph{matroid} $\M=(E,\I)$ is a structure with elements $E$ and a family of independent sets $\I \subseteq 2^E$ satisfying the three \emph{matroid axioms}: (i) $\emptyset \in \I$, (ii) if $A \subseteq B$ and $B \in \I$ then $A \in \I$, and (iii) if $A, B \in \I$ and $|A| < |B|$ then there exists $x \in B \setminus A$ such that $A \cup \{x\} \in \I$.
A \emph{weighted matroid} incorporates a matroid $\M=(E, \I)$ with weights $w \in \R^E$ for its elements.

The rank function of a matroid $\M=(E, \I)$ is denoted by $\rank^\M$, where $\rank^\M(S) \coloneqq \max\{|T| : T\subseteq S, T \in \I\}$. The weighted version of the rank function $\rank_w^\M$ is defined for weighted matroids $(\M,w)$ as $\rank^\M_w(S) = \max\{w(T): T\subseteq S, T \in \I\}$. The span function of a matroid $\M$ is defined as $\spn^\M(S) \coloneqq \{e \in E: \rank^\M(S \cup \{e\})=\rank^\M(S)\}$.

Our proofs will make use of the following basic fact that for a matroid if $S, A$ are sets and $x$ is an element such that (i) $S \in \I$, (ii) for every $a \in A$, $S \cup \{a\} \notin \I$, and (iii) $x \not in \spn(S)$, then $x \notin \spn(S)$.

\begin{claim}
\label{claim:matroid_first_k_opt}
For any $G_k$ be the set of selected elements by Algorithm~\ref{alg:greedy} that satisfies $\vec p(G'_k)<B$ then $G_k$ is maximum weighted independent set in matroid $\mathcal I$ restricted on $\pi[k]$.  
\end{claim}
\begin{proof}
We prove this claim via induction on $k$. The claim trivially holds for $k=1$. Suppose, the claim holds until iteration $k-1$. In this case, we observe that $G^{k-1}$ is a full rank set in matroid $\mathcal I$ restricted on $\pi[k-1]$. If $\pi(k)\notin \spn(G^{k-1})$ then maximum weight independent set in $\pi[k]$ is $G^{k-1} \cup \pi[k]$ which is precisely the set $G_{k+1}$. 

In other case when $\pi(k)\in \spn(G^{k-1})$ then due to \cite[Proposition 1.1.6]{oxley2006matroid}, $\pi(k)$ forms a unique circuit $C$ with $G^{k-1}$. Now, we consider the classic greedy algorithm on matroid that orders elements in the decreasing order of weights and selects an element if it is not spanned by higher weight element. We note that minimum weight element $i$ on circuit $C$ will not be the part of optimal set as it is spanned by the set of higher weight elements. If $\pi(k)$ is the minimum weight element on $C$ then $\pi(k)$ can not be a part of the optimal solution and the current set $G^{k-1}$ remains optimal. In the other case, the optimal greedy algorithm selects the same set of elements in $G^k$ as it selected in $G^{k-1}$ whose weight is higher than weight of $\pi(k)$. Therefore, while making decision about $\pi^k$, the optimal algorithm selects $\pi(k)$ as it does not form a cycle with $G^{k-1}$ on which $\pi(k)$ has the lowest weight.
\end{proof}

\section{Missing Proofs From Section~\ref{sec:additive_price_equilibrium}}
\label{app:proofs_additive}
\subsection{Proof of Lemma~\ref{lem:additive_eps_eq_deviation}} \label{proof_eps_additive}
Since we want to bound the result of the worst $\eps$-equilibrium, we assume that $\vecp(i) \leq \vecc(i) + \eps$ for any module $i$ that is not accepted (by Lemma~\ref{lemma:worst_equil}).

Fix any $\eps$-equilibrium $\vecp$ and let $k$ denote the first module in the bang-per-buck order, according to $\vecp$, that was not selected by $\greedy$.
Since it was not selected by $\greedy$, it must be that $\greedy$ had spent $B - \vecp(k)$ up to this point.
So in $\vecp$, $\greedy$ obtained value $\vecv(S_{\vecp}) \geq \frac{\vecv(k)}{\vecp(k)} \cdot (B - \vecp(k))$.

Now, let $k'$ denote the module with the largest bang-per-buck, according to $\vecc$, that was not selected by $\greedy$ in the equilibrium with respect to $\vecp$.
The difference in value between $\greedy$ that knows the cost and $\greedy$ under the equilibrium is upper bounded by $\frac{\vecv(k')}{\vecc(k')} \cdot B$.

We have $B - \vecp(k) \geq B - M$ and
\[
    \frac{\vecv(k)}{\vecp(k)}
    \geq \frac{\vecv(k')}{\vecp(k')}
    \geq \frac{\vecv(k')}{\vecc(k')+\eps}
    \geq \frac{\vecv(k')}{(1+\eps / m) \cdot \vecc(k')}.
\]
In other words, $\vecv(S_{\vecp}) \geq \frac{\vecv(k')}{(1+\eps / m)\vecc(k')} \cdot \frac{B - M}{B} \cdot B$
Thus,
\[
    \vecv(S_{\vecc}) - \vecv(S_{\vecp})
    \leq \frac{\vecv(k')}{\vecc(k')} \cdot B
    \leq (1+\eps / m) \cdot \frac{B}{B - M} \cdot \vecv(S_{\vecp}).
\]
Rearranging, we conclude that $\vecv(S_{\vecc}) \leq \left[ 2 + \frac{\eps}{m} + \left( 1 + \frac{\eps}{m} \right) \cdot \frac{M}{B-M} \right] \cdot \vecv(S_{\vecp})$.
Observing that $\frac{M}{B-M} = \frac{\lambda}{1-\lambda}$ gives the result.

\subsection{Proof of Lemma~\ref{lemma:worst_equil}}
\label{app:proof_worst_equil}
\begin{proof}[Proof of Lemma~\ref{lemma:worst_equil}]
Let $i$ be as in the statement of the lemma and suppose that $\vec{p}(i) > \vec{c}(i) + \eps$ (otherwise the claim is trivial).
We consider a two-step process.
First, we modify the price of $i$ so that $\vec{p}'(i) = \vec{c}(i) + \eps$.
Next, define the price vector
\begin{equation}
    \label{eqn:vecp_def_additive}
    \bar{\vec{p}}(j) = \begin{cases}
        \max\left\{\vecc(j), \frac{\vecc(i) + \eps}{\vec{v}(i)} \cdot \vec{v}(j) - \frac{\eps}{2} \right\} & j \in S_{\vec p} \setminus S_{\vecp'} \\
        \vec{p}'(j) & \text{otherwise}
    \end{cases}.
\end{equation}

We first check the last three assertions in the lemma.
The first assertion is true by construction.
The third assertion follows from the following claim.
The second assertion follows from Claim~\ref{claim:Sbarvecp_subset_Svecp}.
\begin{claim}
\label{claim:barvecp_leq_vecp}
We have $\bar{\vecp} \leq \vecp' \leq \vecp$.
\end{claim}
\begin{proof}
The fact that $\vecp' \leq \vecp$ is trivial since the only difference between $\vecp'$ and $\vecp$ is to lower module $i$'s price to $\vecc(i) + \eps$.
We now need to show that for $j \in S_{\vecp} \setminus S_{\vecp'}$, we have $\bar{\vecp}(j) \leq \vecp(j)$.
Clearly, $\vecc(j) \leq \vecp(j)$ since no module sets a price below their cost.
It remains to show that $\frac{\vecc(i) + \eps}{\vecv(i)} \cdot \vecv(j) - \frac{\eps}{2} \leq \vecp(j)$.
Let $A_j$ be the set of modules that $\greedy$ considered before $j$ when the price vector is $\vecp$.
We have that $\vecp(A_j) + \vecp(j) \leq B$ because $j \in S_{\vecp}$.
Also note that $i \notin A_j$ otherwise $i \in S_{\vecp}$.
Thus, $\vecp'(A_j) + \vecp'(j) = \vecp(A_j) + \vecp(j) \leq B$.
Since $j \in S_{\vecp} \setminus S_{\vecp'}$ and the only difference is module $i$'s price, it must be that $\greedy$ considers module $i$ before module $j$ when the price vector is $\vecp'$.
This implies that $\frac{\vecp(j)}{\vecc(j)} \geq \frac{\vecc(i)+\eps}{\vecv(i)}$, as desired.
\end{proof}

\begin{claim}
\label{claim:i_notin_Sbarp}
We have that $i \notin S_{\bar{\vecp}}$.
\end{claim}
\begin{proof}
Let $A_i'$ (resp.~$\bar{A}_i$) be the set of modules that $\greedy$ considers before module $i$ when the price vector is $\vecp'$ (resp.~$\bar{\vecp}$).
Note that (i) $A_i' \subseteq \bar{A}_i$ and (ii) $\bar{\vecp}(A_i') = \vecp'(A_i')$.
The first assertion follows from Claim~\ref{claim:barvecp_leq_vecp} since module $i$ has the same price in both $\vecp'$ and $\bar{\vecp}$.
The second assertion is because $\greedy$ makes exactly the same decisions for $A_i'$ under $\vecp$ and $\vecp'$.
In particular, $A_i' \cap (S_{\vecp} \setminus S_{\vecp'}) = \emptyset$.
We thus conclude, using the definition of $\bar{\vecp}$ (Eq.~\eqref{eqn:vecp_def_additive}), that $\bar{\vecp}(j) = \vecp'(j) = \vecp(j)$ for all $j \in A_i'$.

The hypothesis of the lemma means that module $i$ cannot modify its price to $\vec{c}(i) + \eps$ and be in the accepted set.
In particular,
\[
    B < \vecp(A_i') + \vecc(i) + \eps = \vecp'(A_i') + \vecp'(i) = \bar{\vecp}(A_i') + \bar{\vecp}(i)
    \leq \bar{\vecp}(\bar{A}_i) + \bar{\vecp}(i),
\]
where the first equality is because $\vecp(j) = \vecp'(j)$ for $j \neq i$,
the second equality is by (ii) above,
and the last equality is by (i) above.
We conclude that $i \notin S_{\bar{\vecp}}$ since including it when it is considered by $\greedy$ would violate the budget constraint.
\end{proof}

\begin{claim}
\label{claim:Sbarvecp_subset_Svecp}
$S_{\bar{\vecp}} \subseteq S_{\vecp}$.
\end{claim}
\begin{proof}
We will prove the contra-positive: if $j \notin S_{\vecp}$ then $j \notin S_{\bar{\vecp}}$.
If $j = i$ then this follows from Claim~\ref{claim:i_notin_Sbarp}.
So now assume $j \neq i$.
Let $A_j'$ (resp.~$\bar{A}_j$) be the set of modules that $\greedy$ considers before module $j$ when the price vector is $\vecp'$ (resp.~$\bar{\vecp}$).
Note that the ordering, according to $\greedy$, between $j$ and $i$ is the same in both the price vectors $\vecp'$ and $\bar{\vecp}$ since both of their prices are unchanged between the two price vectors.
If $i \in A_j'$ (and hence $i \in \bar{A}_j$) then $j \notin S_{\bar{\vecp}}$) since $\greedy$ terminates after considering $i$.
Henceforth, we assume that $j \notin A_j'$.

In this case, the proof is very similar to Claim~\ref{claim:i_notin_Sbarp}.
We again have that (i) $A_j' \subseteq \bar{A}_j$ and (ii) $\bar{\vecp}(A_j') = \vecp'(A_j')$.
The first assertion follows from Claim~\ref{claim:barvecp_leq_vecp} since module $j$ has the same price in both $\vecp'$ and $\bar{\vecp}$.
The second assertion is because $\greedy$ makes exactly the same decisions for $A_j'$ under $\vecp$ and $\vecp'$.
In particular, $A_j' \cap (S_{\vecp} \setminus S_{\vecp'}) = \emptyset$.

Since $j \notin S_{\vecp}$, we have that
\[
    B < \vecp(A_j') + \vecp(j) = \vecp'(A_j') + \vecp'(j) = \bar{\vecp}(A_j') + \bar{\vecp}(j)
    \leq \bar{\vecp}(\bar{A}_j) + \bar{\vecp}(j),
\]
where the first equality is because $\vecp(j) = \vecp'(j)$ for $j \neq i$,
the second equality is by (ii) above and Eq.~\eqref{eqn:vecp_def_additive},
and the last equality is by (i) above.
We conclude that $j \notin S_{\bar{\vecp}}$ since including it when it is considered by $\greedy$ would violate the budget constraint.
\end{proof}

Finally, we check that $\bar{\vecp}$ is still an $\eps$-equilibrium which follows from the following case analysis.
\paragraph{Case 1: $j \in S_{\vec p} \cap S_{\vec p'}$.}
In this case, we have $\vec p(j) = \bar{\vecp}(j)$.
Since $\vecp$ is an $\eps$-equilibrium price, we have that module $j$ is rejected in the price vector $\vec{q} = (\vecp(j) + \eps, \vecp_{-j})$.
Hence, $j$ is also rejected in the price vector $\vec{q}' = (\vecp(j) + \eps, \vecp'_{-j})$ since all modules that came before $j$ in $\vec{q}$ also come before $j$ in $\vec{q}'$ (module $i$ is the only module with a different price between $\vec{q}$ and $\vec{q}'$).
Finally, consider $\bar{\vec{q}} = (\vecp(j) + \eps, \bar{\vecp}_{-j})$.
If $j$ comes before $i$ according to $\greedy$ in $\bar{\vec{q}}$ then $\greedy$ coincides on all modules up to and including module $j$ on $\bar{\vec{q}}$ and $\vec{q}'$.
So $j$ is rejected.
On the other hand, if $j$ comes after $i$ according to $\greedy$ then $j$ must be rejected.
This is because (i) the set of modules inspected by $\greedy$ up to and including $i$ in $\bar{\vecp}$ is a subset of the module inspected by $\greedy$ up to and including $j$ in $\bar{\vec{q}}$ and (ii) $\bar{\vecp} \leq \bar{\vec{q}}$.
We conclude that $j$ cannot raise its price by at least $\eps$ and still be included in $S_{\bar{\vecp}}$.

\paragraph{Case 2: $j \in S_{\vec p} \setminus S_{\vec p'}$.}
Since $\bar{\vec{p}}(j) \geq \frac{\vec{c}(i) + \eps}{\vec{v}(i)} - \frac{\eps}{2}$ then module $j$ cannot deviate to gain more than $\eps$ since it would have to increase its bid to at least $\frac{\vec{c}(i) + \eps}{\vec{v}(i)} + \frac{\eps}{2}$ to do so.
However, it would come after module $i$ in the bang-per-buck order and module $i$ is not accepted so neither would module $j$.

\paragraph{Case 3: $j \notin S_{\vec p}$.}
For $j = i$, we established above (Claim~\ref{claim:i_notin_Sbarp}) that $i \notin S_{\bar{\vecp}}$ and any possible deviation gives module $i$ utility strictly less than $\eps$ since $\bar{\vecp}(i) = \vecc(i) + \eps$.

Now suppose $j \neq i$.
Consider the price vector $\vec{q} = \left( \vecc(j) + \eps, \vecp \right)$.
Module $j$ must be rejected because otherwise $\vecp$ is not an $\eps$-equilibrium.
Now consider the price vector $\vec{q}' = \left( \vecc_j, \vecp' \right)$.
In this case, $j$ is still rejected because (i) if $j$ comes before $i$ according to $\greedy$ in $\vec{q}'$ then its outcome under $\greedy$ in $\vec{q}$ is the same as in $\vec{q}'$ or (ii) if $j$ comes after $i$ according to $\greedy$ then $j$ is rejected.
Finally, consider the price vector $\bar{\vec{q}} = (\vecc_j, \bar{\vecp})$.
If $j$ comes before $i$ then there may now be additional modules before $j$ in $\bar{\vec{q}}$ compared to $\vec{q}'$.
So module $j$ remains rejected.
If $j$ comes after $i$ then $j$ is rejected because $i$ is rejected.

Since module $j$ is rejected in $(\vecc(j) + \eps, \bar{\vecp}_{-j})$, it must also be rejected if it modifies its bid to any bid strictly more than $\vecc(j) + \eps$.
On the other hand, any bid lower than $\vecc(j) + \eps$ provides $j$ with strictly less than $\eps$ utility.
So any bid placed by $j$ results in utility strictly less than $\eps$.
\end{proof}

\subsection{Proof of Lemma~\ref{lemma:additive_greedy_approx_with_cost}}
\label{app:additive_greedy_approx_with_cost}
\begin{proof}[Proof of Lemma~\ref{lemma:additive_greedy_approx_with_cost}]
By rearranging, we can assume that $\frac{\vecv(1)}{\vecc(1)} > \ldots > \frac{\vecv(n)}{\vecc(n)}$.
This means that $S_{\vecc} = [k]$ for some $k \in [n]$.
An upper bound on $\vecv(S_{\OPT})$ is to take the first $k + 1$ modules.
Thus
\[
    \vecv(S_{\OPT}) - \vecv(S_{\vecc}) \leq \vecv(k+1) \leq M \cdot \frac{\vecv(k+1)}{\vecc(k+1)}.
\]
Since $k+1 \notin S_{\vecc}$, we have $\vecc([k]) \geq B - \vecc([k+1]) \geq B - M$.
Thus
\[
    \vecv(S_{\vecc}) \geq (B - M) \cdot \frac{\vecv(k+1)}{\vecc(k+1)}
\]
Combining the above two inequalities, we conclude that
\[
    \vecv(S_{\OPT}) \leq \left( 1 + \frac{M}{B - M} \right) \vecv(S_\vecc) = \left( 1 + \frac{\lambda}{1-\lambda} \right) \vecv(S_{\vecc}). \qedhere
\]
\end{proof}
\section{Price Competition for Unweighted Matroid}
\label{app:price_comp_matroid}
In this section, we present the equilibrium analysis for the case when the buyer's constraint is a matroid constraint.
This section looks at an algorithm (Algorithm~\ref{alg:modified_greedy}) which is a natural generalization of Algorithm~\ref{alg:knapsack} in that it skips over any element that is not feasible in the bang-per-buck order given the modules selected so far.
Note that Algorithm~\ref{alg:modified_greedy} is equivalent to Algorithm~\ref{alg:greedy} when $\vecv(i) = 1$ for all $i$.
In this section, we show that for all $\vecv$, an equilibrium exists but we only prove an equilibrium quality result for the case where $\vecv(i) = 1$ for all $i$.

\begin{algorithm}
\caption{Modified Greedy Algorithm}
\label{alg:modified_greedy}
\begin{algorithmic}[1]
\State \textbf{Input}: price vector $\vec p$, values $\vec v$, budget $B$, feasible sets $\I \in 2^{[n]}$.
\State \textbf{Output}: set $G$ of modules.
\State Initialize $G \gets \emptyset$.
\State Re-arrange modules such that $\frac{\vecv(1)}{\vecp(1)} \geq \frac{\vecv(2)}{\vecp(2)} \geq \dots \geq \frac{\vecv(n)}{\vecp(n)}$ with ties broken arbitrarily.
\For{$i=1,2 \dots ,n$}
\If{$G \cup \{i\} \notin \I$}
\State \textbf{continue}
\EndIf
\If{$\vec p(G) + \vec p(i) > B$}
\State \textbf{break}
\EndIf
\State $G \gets G \cup \{i\}$
\EndFor
\end{algorithmic}
\end{algorithm}
We divide the proof of the main theorem of this section into two parts. First, in Subsection~\ref{sec:matroid_price_eqm_existance}, we show the existence of $\eps$-equilibrium prices and later in Subsection~\ref{sec:matroid_eq_analysis} we analyze the quality of all $\eps$-equilibrium prices.

\subsection{Existence of Price Equilibrium}\label{sec:matroid_price_eqm_existance}
We begin our analysis by showing the existence of $\eps$-equilibrium prices. To this end, we provide an algorithm that explicitly constructs an $\eps$-equilibrium price vector for the pricing game between the modules.

Algorithm~\ref{alg:equilibrium_dynamics} shows how to construct an equilibrium.
In addition, to simplify the notations and reduce clutter, we assume that the buyer breaks ties in favor of modules with higher $\frac{\vecv(i)}{\vecc(i)}$ and as a result, we show that Algorithm~\ref{alg:equilibrium_dynamics} converges to an exact equilibrium.
More specifically, we assume that the buyer breaks ties in the initial permutation $\pi^0$ defined in Algorithm~\ref{alg:equilibrium_dynamics}.

We note that since the buyer has no information about the module's private cost, the buyer cannot implement such a tie-breaking rule while selecting modules.
In order to simulate the desired tie-breaking rule, we can lower the price of accepted modules by an infintesimal amount which results in an $\eps$-equilibrium for any small $\eps > 0$.

The algorithm begins by initializing $\vecp$ as $\vec p = \vec c$ and an ordering $\pi^0$ which is decreasing in their bang-per-buck w.r.t.~$\vecp$ at round zero. We then define the set $T^\ind$ that denotes the set of modules that increase their price at iteration $\ind$ and $\pi^\ind$ that denotes the bang-per-buck order over modules at the price $\vec p^\ind$. Throughout the algorithm, we maintain the following invariants at iteration $\ind$: (a) the price of modules $\pi^0(\ind), \dots, \pi^0(n)$ are unchanged and equal to their cost, (b) the bang-per-buck of the modules higher than $ \frac{\vec v(\pi(\ind))}{\vec c(\pi(\ind))}$ remains higher than $ \frac{\vec v(\pi(\ind))}{\vec c(\pi(\ind))}$ at price $\vec p^\ind$,  (c) for every module in $i\in T^k$, we have $\frac{\vec v(i)}{\vec p^\ind(i)} = \frac{\vec v(\pi^0(\ind))}{\vec c(\pi^0(\ind ))}$.  

We can describe the algorithm as follows.
First, we raise the prices of all modules in $T^{k-1}$ so that their bang-per-buck meets that of module $\pi^k(k)$, whose price is still at its cost.
Next, we check whether or not $\pi^k(k) \in \spn(A^{k-1} \cup T^{k-1})$.
If not we add it to $T^{k-1}$ to obtain $T^k$.
Otherwise, we find the (unique) circuit in $A^{k-1} \cup T^{k-1} \cup \{\pi^k(k)\}$ and add it, without $\pi^k(k)$, to $A^{k-1}$ to get $A^k$.
Doing so freezes the price of these modules.
We then remove $A^{k-1}$ from $T^{k-1}$ to get $T^k$, the set of modules which can still raise their prices.
Next, we check whether or not $A^{k} \cup T^{k}$ can fit within the budget at the price $\vecp^k$.
If not then, if necessary, we lower the prices for $T^{k-1}$ (whose prices we had just raised) so that $A^{k-1} \cup T^{k-1}$ fits within the budget.
In effect, this ensures that $\greedy$ rejects $\pi^k(k)$.

If the algorithm does manage to go through every module and $T^n \neq \emptyset$ then we increase the price of every module in $T^n$ until the price of all accepted modules in $A^n \cup T^n$ is equal to the budget.

\begin{algorithm}
\caption{Equilibrium Construction}
\label{alg:equilibrium_dynamics}
\begin{algorithmic}[1]
\State Initialize price $\vecp^0 \gets \vecc$ and $\pi^0$ such that $\frac{\vecv(\pi^0(i))}{\vecc(\pi^0(i))}\geq \frac{\vecv(\pi^0(j))}{\vecc(\pi^0(j))}$ for $i < j$.
\State Initialize $A^0 \gets \emptyset, T^0 \gets \emptyset$.
\For{$\ind = 1, \ldots, n$}

\State Copy $\pi^{\ind-1}, \vec p^{\ind-1}$ into $\pi,\vec p$ for simplicity of notation.
\State Update prices:
\[
\vecp^k(\pi(j)) =
\begin{cases}
\frac{\vecp(\pi(k))}{\vecv(\pi(k))} \cdot \vecv(\pi(j)) & j \in T^k \\
\vecp(\pi(j)) & j \notin T^k
\end{cases}
\]

\If{$\pi(\ind) \notin \spn(A^{k-1} \cup T^{k-1}$)}
    \Comment{Accept $\pi(k)$.}
    \State $A^k \gets A^{k-1}$
    \State $T^k \gets T^{k-1} \cup \{\pi(k)\}$
\Else
    \Comment{Reject $\pi(k)$ and freeze modules in the circuit.}
    \State Let $C^k$ be the unique circuit in $A^{k-1} \cup T^{k-1} \cup \{\pi(k)\}$. \label{line:find_circuit} 
    \State Let $A^k \gets A^{k-1} \cup C^k \setminus \{\pi(k)\}$.
    \State Let $T^k \gets T^{k-1} \setminus A^k$.
\EndIf

\State Set $\pi^k$ as bang-per-buck ordering with respect to $\vecp^k$, breaking ties according to $\pi^0$.
\If{$\vecp^k(A^k \cup T^k) > B$} \label{line:budget_check}
\Comment{Reject last module, update prices, and terminate.}
\State $T^{k} \gets T^{k-1}$, $A^k \gets A^{k-1}$.
\State Update $\vecp^k(i) \gets \min\left\{\vecv(i) \cdot \frac{B - \vecp^{k-1}(A^{k-1})}{\vecv(T^{k-1})}, \vecp^k(i) \right\}$ for $i \in T^{k-1}$. \label{line:lower_price}
\State \textbf{return} $\vecp^k, \pi^k$
\EndIf
\EndFor

\State Update $\vecp^n(i) \gets \vecv(i) \cdot \frac{B - \vecp^n(A^{n})}{\vecv(T^n)}$ for $i \in T^n$.
\State \textbf{return} $\vecp^n, \pi^n$
\end{algorithmic}
\end{algorithm}

Note that Algorithm~\ref{alg:equilibrium_dynamics} terminates in at most $n$ iterations.
In each iteration, the algorithm needs to be able to test feasibility of a set and to be able to find circuits.

Let $\bar{\vec  p}$ be the price computed by Algorithm~\ref{alg:equilibrium_dynamics}. We first characterize $S_{\bar{\vecp}}$, the set of selected modules at the price, and obtain a structural decomposition of the set $S_{\bar {\vec p}}$ in the given underlying matroid $\mathcal M$. Throughout the analysis, we let $k^*$ be the last iteration of the dynamics.

First, we show that after the price updates at iteration $k$, the selected modules at round $k-1$ and round $k$ stays identical among the modules $\pi^{k-1}(1), \dots, \pi^{k-1}(k)$ with top $k$ bang-per-buck values. This property of our dynamics ensures that the modules that increase their price remain selected by the greedy allocation rule. Hence, the modules that update price at round $k$ increase their utility. 
  
\begin{claim}\label{claim:winners_stays_winners}
For any iteration $k \leq k^*$, we have $S_{\vecp^{k-1}} \cap \pi^{k-1}([k-1]) = S_{\vecp^{k}} \cap \pi^{k}([k-1]) = A^{k-1} \cup T^{k-1}$.
\end{claim}
\begin{proof}

Define $R^k = \pi^k([k]) \setminus (A^k \cup T^k)$.
We think of $R^k$ as the set of modules that $\greedy$ is sure to reject.
We show that the following invariants hold:
\begin{enumerate}
    \item $\pi^{k-1}([k-1]) = \pi^k([k-1])$ and $\pi^{k-1}(k) = \pi^{k}(k)$.
    \item when the prices are $\vecp^k$, the bang-per-buck of any module in $A^{k-1}$ is at least that of any module in $T^{k-1}$.
    \item when the prices are $\vecp^k$, the bang-per-buck of any module in $T^{k-1}$ is at least that of any module in $\pi^k([n]) \setminus \pi^k([k])$.
    \item $\vecp^k(A^{k-1} \cup T^{k-1}) \leq B$.
    \item if $i \leq k$ and $\pi^{k}(i) \in R^k$ then $\pi^k(i) \in \spn(A^k \cap \pi^k([i-1]))$.
\end{enumerate}
These invariants show that if $\greedy$ follows the order of $\pi^k$ then it accepts every module in $A^{k-1} \cup T^{k-1}$ when the price vector is $\vecp^k$.
In other words, $S_{\vecp^{k-1}} \cap \pi^{k-1}([k-1]) = S_{\vecp^{k}} \cap \pi^{k}([k-1])$.

The invariants are trivial if $k = 1$ since $A^{k-1} = T^{k-1} = R^{k-1} = \emptyset$.
We now assume that the invariant holds at iteration $k-1$ and prove that it remains true at iteration $k$.

For the first invariant, the only modules that increase their price is $T^{k-1} \subseteq \pi^{k-1}([k-1])$.
Since their bang-per-buck is at most that of module $\pi^{k-1}([k])$ this means that all modules in $\pi^{k-1}([k-1])$ still come before module $\pi^{k-1}([k])$.
So $\pi^{k-1}([k-1]) = \pi^k([k-1])$.
The fact that $\pi^{k-1}(k) = \pi^{k}(k)$ now easily follows because ties are broken according to $\pi^0$.

For the second invariant, we need the following claim.
\begin{claim}
  \label{claim:T_k_increasing}
  For all $i \in T^{k-1}$, we have $\vecp^k(i) \geq \vecp^{k-1}(i)$.
\end{claim}
\begin{proof}
Note that $\vecp^{k-1}(i) = \frac{\vecp(\pi^{k-1}(k-1))}{\vecv(\pi^{k-1}(k-1)))} \cdot \vecv(\pi^{k-1}(i))$.
We have that
\[
    \vecp^k(i) \geq \min\left\{
        \vecv(i) \cdot \frac{B - \vecv^{k-1}(A^{k-1})}{\vecv(T^{k-1})},
        \frac{\vecp(\pi^{k}(k))}{\vecv(\pi^{k}(k)))} \cdot \vecv(\pi^{k}(i))
    \right\}.
\]
The second term in the minimum is more than $\vecp^{k-1}(\pi^{k-1}(i))$ since $\pi^{k-1}(k-1)$ comes before $\pi^{k}(k)$ in the bang-per-buck order.
For the first term, we know that $A^{k-1} \cup T^{k-1}$ is budget-feasible (otherwise we would have terminated) so $B - \vecp^{k-1}(A^{k-1}) \geq \vecp^{k-1}(T^{k-1})$.
Thus $\vecv(i) \cdot \frac{B - \vecv^{k-1}(A^{k-1})}{\vecv(T^{k-1})} \geq \vecv(i) \cdot \frac{\vecp^{k-1}(T^{k-1})}{\vecv(T^{k-1})} = \vecp^{k-1}(i)$.
\end{proof}

Since the bang-per-buck of any module in $A^{k-1}$ is at most that of any module in $T^{k-2}$ when the price vector is $\vecp^{k-1}$ and the prices in $A^{k-1}$ do not change, we conclude, using Claim~\ref{claim:T_k_increasing} that the same is true when the price vector is $\vecp^k$.
This shows that the second invariant holds.

The third invariant is because prices of modules in $\pi^k([k-1])$ are only raised to meet the bang-per-buck of module in $\pi^k(k)$.

The fourth invariant is trivial if Line~\ref{line:budget_check} is false.
Otherwise, it is true because (i) $\vecp^{k-1}(A^{k-1}) + \vecp^{k-1}(T^{k-1})$ by the fact that the invariant holds for $k-1$, (ii) $\vecp^{k}(A^{k-1}) + \vecp^{k}(A^{k-1})$, and (iii) $\vecp^k(T^{k-1}) \leq B - \vecp^{k-1}(A^{k-1}) \leq \vecp^{k-1}(T^{k-1})$ which is just a re-arrangement of (i).

Finally, for the fifth invariant, first suppose that $i \leq k-1$.
We have that $\pi^{k-1}(i)$ forms a circuit with $A^{k-1} \cap \pi^{k-1}([i-1])$ which we call $C_i$.
When we raise prices in $T^{k-1}$, invariants $2$ and $3$ as well as the fact that $\pi^{k-1}, \pi^k$ are are consistent with $\pi^0$, imply that the relative ordering in $C_i$ does not change.
So $\pi^k(i)$ still forms a circuit with $A^{k-1} \cap \pi^{k-1}([i-1])$.
The case $i = k$ is straightforward since in this case Line~\ref{line:find_circuit} finds the circuit in $A^{k-1} \cup T^{k-1}$ and adds it to $A^k$.
In addition, observe that $A^k \cap \pi^k([k-1]) = A^k$.
\end{proof}

Let $\bar{\vecp}$ be the price computed by Algorithm~\ref{alg:equilibrium_dynamics}.
\begin{claim}
    \label{claim:matroid_greedy_output}
    $S_{\bar{\vecp}} = A^{k^*} \cup T^{k^*}$.
\end{claim}
\begin{proof}
    Claim~\ref{claim:winners_stays_winners} shows that $S_{\bar{\vecp}} \cap \pi^{k^*}([k^*-1]) = A^{k^*-1} \cup T^{k^*-1}$.
    First, suppose that $k^* < n$.
    If $A^{k^*} = A^{k^*-1}$ and $T^{k^*} = T^{k^*-1}$ before Line~\ref{line:budget_check} then Line~\ref{line:lower_price} ensures that $\vecp^{k^*}(A^{k^*-1} \cup T^{k^*-1}) = B$.
    So $S_{\bar{\vecp}} = A^{k^*-1} \cup T^{k^*-1}$ since any module not in $\pi^{k^*}([k^*-1])$ is inspected by $\greedy$ after $A^{k^*-1} \cup T^{k^*-1}$.
    
    Now suppose that $k^* = n$.
    Since Line~\ref{line:budget_check} never happens we have that $S_{\bar{\vecp}} \cap \pi^{k^*}([k^*-1]) = A^{k^*-1} \cup T^{k^*-1}$ and $A^{k^*} \cup T^{k^*}$ is feasible for $\greedy$ (and includes $A^{k^*-1} \cup T^{k^*-1}$).
    So $\greedy$ selects $A^{k^*} \cup T^{k^*}$.
\end{proof}

\begin{lemma}
We can decompose
\[
    S_{\bar{\vecp}} = T_{k^*} \cup \bigcup_{k < k^*} A_{k+1} \setminus A_k = T_{k^*} \cup \bigcup_{k < k^*} T_k \setminus (T_{k+1} \cup \pi^k(k)).
\]
\end{lemma}
\begin{proof}
The first equality follows from Claim~\ref{claim:matroid_greedy_output} and that $A_{k^*} = \bigcup_{k < k^*} A_{k+1} \setminus A_k$.
The second inequality is because $A^{k+1} \setminus A^k = T^{k} \setminus T^{k+1}$ in Algorithm~\ref{alg:equilibrium_dynamics}.
\end{proof}

\begin{theorem}
\label{thm:unweighted_eq_exist}
The vector $\bar{\vecp}$ is an equilibrium price.
\end{theorem}
\begin{proof}
If $i \notin A^{k^*} \cup T^{k^*} \cup R^{k^*} = \pi^{k^*}([k^*])$ then by Claim~\ref{claim:matroid_greedy_output}, module $i$ is not inspected by $\greedy$ before the budget constraint is met. Since module $i$ is already bidding its cost, it cannot deviate to get selected.

Suppose that $i \in R^{k^*}$.
By invariant $5$ in the proof of Claim~\ref{claim:winners_stays_winners}, we have that $\pi^{k^*}(i) \in \spn(A^{k^*} \cap \pi^{k^*}([i-1]))$.
So $i$ forms a circuit with $A^{k^*} \cap \pi^{k^*}([i-1])$ and is inspected by $\greedy$ after every other module in this circuit.
Since it is already bidding its cost, it cannot deviate to get selected by $\greedy$.

Now suppose $i \in A^{k^*}$.
This means there is some $k \in [k^*]$ such that $i \in A_{k} \setminus A_{k-1}$.
At this point, there was a circuit $C \subseteq A^{k} \cup \{\pi^k(k)\}$ that contained $i$ and after iteration $k$ all prices in $C$ were never raised.
If module $i$ raises its price then it would be inspected by $\greedy$ after every other module in $C$ and so it would be rejected. So $i$ has no incentive to deviate.

Finally, suppose $i \in T^{k^*}$.
By construction all modules in $T^{k^*}$ have the same bang-per-buck and have the lowest bang-per-buck among all modules selected by $\greedy$.
There are two cases.
If $\bar{\vecp}(A^{k^*} \cup T^{k^*}) = B$ then if $i$ increases its price, it would be rejected by $\greedy$ since there is no budget remaining when $i$ is inspected (it could come after every other module in $A^{k^*} \cup T^{k^*}$).
If $\bar{\vecp}(A^{k^*} \cup T^{k^*}) < B$ then we must have entered the if condition in Line~\ref{line:budget_check}.
Now in Line~\ref{line:lower_price}, we must set $\vecp^{k^*}$ to be the second term in the minimum.
If it was the first term, the budget constraint would be exactly met with modules $A^{k^*} \cup T^{k^*}$.
Note that $\pi^{k^*}(k^*) \in \spn(A^{k^*} \cup T^{k^*})$.
To see this, note that if not, we would have $A^{k^*} \cup T^{k^*} = A^{k^*-1} \cup T^{k^*-1}$.
Since in Line~\ref{line:lower_price}, we must set $\vecp^{k^*}$ to be the second term in the minimum, this implies that $\vecp^{k^*}(A^{k^*} \cup T^{k^*}) > B$.
We conclude then that $\bar{\vecp}(i) = \vecp^{k^*}(i)$.
In particular, this means that $i$ has the same bang-per-buck as module $\pi^{k^*}(k^*)$.
Now if $i$ raises its price then it would be inspected by $\greedy$ after $\pi^{k^*}(k^*)$.
Since $\pi^{k^*}(k^*)$ does not form a circuit with $A_{k^*} \cup T_{k^*}$, $\greedy$ would accept $\pi^{k^*}(k^*)$ but reject $i$.
So we conclude that $i$ cannot deviate.
\end{proof}

\subsection{Quality of Equilibria} \label{sec:matroid_eq_analysis}
First, we prove the following lemma that reduces our quality of equilibrium analysis to the case when the modules which are not selected at the equilibrium price sets their price equals to their respective cost.  
In this section, we will assume that $\frac{\vecv(i)}{\vecc(i)}$ are all distinct.

\begin{lemma}
Let $\vec p$ be an $\eps$-equilibrium price and $S_{\vec p}$ be the set of selected modules by $\greedy$. Then for any module $i\notin S_\vec p$, there exists an $\eps$-equilibrium prices $\vec {\bar p}$ such that (i) $\bar{\vecp} \leq \vecc(i) + \eps$, (ii)
$S_{\vec {\bar p}} \subseteq S_{\vec p}$, and (iii) $\bar{\vecp} \leq \vecp$.
\end{lemma}
\begin{proof}
Let $\pi$ be the ordering over the modules according to their bang-per-buck at price $\vec p$ with ties broken in lexicographic order. Fix a module $i \notin S_{\vecp}$ such that $\vecp(i) > \vecc(i) + \eps$.
Since $\vecp$ is an $\eps$-equilibrium price, we have that $u_i(\vecc(i) + \eps, \vecp_{-i}) = 0$. In particular, module $i$ is not accepted even if it lowered its price to $\vecc(i) + \eps$.
We now consider the deviation price vector $\vecp'$ where $\vecp'(i) = \vecc(i) + \eps$ and $\vecp'(k) = \vecp(k)$ for $k \neq i$.

Let $\pi'$ be the new ordering of modules according to their bang-per-buck at price $(\vec c(i)+\varepsilon,\vec p_{-i})$ and $j' = \pi'^{-1}(i)$ be the new position of module $i$ in the new order $\pi'$ (again, breaking ties in lexicographic order). We observe that since $i \notin S_\vec p$, either (1) $i \in \spn (\{\pi'(1),\dots,\pi'(j'-1)\})$ or (2) $i \notin \spn (\{\pi'(1),\dots,\pi'(j'-1)\})$ and $\vec c(i) + \eps + \vecp(S_{\vec p}\cap \{\pi'(1),\dots,\pi'(j'-1)\}) > B$. We analyze both cases separately. 
In both cases, by definition, $i \notin S_{\vec c(i) + \varepsilon , \vec p_{-i}}$.

\paragraph{Case 1: $i \in \spn (\{\pi'(1),\dots,\pi'(j'-1)\})$.} In this case, we claim that $S_{\vec p} = S_{\vec c(i) + \varepsilon , \vec p_{-i}}$.  We prove the claim by showing that $S_\vec p\cap \{\pi(1),\dots , \pi(n)\} = S_{\vec c(i)+\varepsilon, \vec p_{-i}} \cap \{\pi'(1),\dots , \pi'(n)\}$ where $n$ is the size of the set $M$.

We first observe that $\pi(k) = \pi'(k)$ for $k < j'$.
Therefore, we have $S_\vec p\cap \{\pi(1),\dots , \pi(k)\} = S_{\vec c(i)+\varepsilon, \vec p_{-i}} \cap \{\pi'(1),\dots , \pi'(k)\}$.
However, the modules $\pi(j'),\dots, \pi(j-1)$ have their position shifted up by one in the new ordering $\pi'$.
Then we show that for $j' \leq  k< j$, $S_\vec p \cap \{\pi(1),\dots , \pi(k)\} = S_{\vec c(j) + \varepsilon ,\vec p_{-j}} \cap \{\pi'(1),\dots , \pi'(k+1)\} $.
We prove this claim via induction on $k$.
We start with the base case where $k = j'$.
Our assumption here is that $i \in \spn (\{\pi'(1),\dots,\pi'(j'-1)\}$.
Now consider the module $\pi(j')= \pi'(j'+1)$.
We have that
\[
    S_{\vecp} \cap \{\pi(1), \ldots, \pi(j'-1)\} = S_{\vec c(i) + \varepsilon , \vec p_{-i}} \cap \{\pi'(1) \, \ldots, \pi'(j')\}.
\]
In other words, when $\greedy$ inspects module $\pi(j') = \pi'(j'+1)$, the set selected thus far is identical irrespective of whether the order $\pi$ or $\pi'$ is used.
We conclude that $\pi(j') \in S_{\vecp}$ if and only if $\pi'(j'+1) \in S_{\vec c(i) + \varepsilon ,\vec p_{-i}}$.

Next, consider any $k \in \{j'+1, \ldots, j-1\}$ and assume the induction hypothesis
\[
    S_{\vecp} \cap \{\pi(1), \ldots, \pi(k-1)\} = S_{\vec c(i) + \varepsilon , \vec p_{-i}} \cap \{\pi'(1) \, \ldots, \pi'(k)\}.
\]
The argument here is identical to that in the previous paragraph.
When $\greedy$ inspects module $\pi(k) = \pi'(k+1)$, the selected thus far is identical irrespective of whether the order $\pi$ or $\pi'$ is used.
We conclude that $\pi(k) \in S_{\vecp}$ if and only if $\pi'(k+1) \in S_{\vec c(i) + \varepsilon ,\vec p_{-i}}$.

At this point, we have established that
\[
    S_{\vecp} \cap \{\pi(1), \ldots, \pi(j-1)\} = S_{\vec c(i) + \varepsilon , \vec p_{-i}} \cap \{\pi'(1) \, \ldots, \pi'(j)\}.
\]
Our initial hypothesis of the lemma is that $j \notin S_{\vecp}$, whence we conclude that
\[
    S_{\vecp} \cap \{\pi(1), \ldots, \pi(j)\} = S_{\vec c(i) + \varepsilon , \vec p_{-i}} \cap \{\pi'(1) \, \ldots, \pi'(j)\}.
\]
For $k > j$, we have $\pi(k) = \pi'(k)$ and so $\greedy$ behaves identically irrespective of whether the ordering is $\pi$ or $\pi'$.

It remains to establish that $\vecp' = (\vec c(i)+\varepsilon,\vec p_{-i})$ is an $\eps$-equilibrium price.
First, consider a module $i' \notin S_{\vecp'}$ and $i' \neq i$ with $\vecp'(i') > \vecc(i') + \eps$.
We claim that $i'$ would remain rejected by $\greedy$ if it lowered its price to $\vecc(i') + \eps$ when the price vector of the remaining modules is $\vecp'_{-i'}$.
As we established above, $\greedy$ makes identical decisions irrespective of whether the ordering is given by $\pi$ or $\pi'$.
In particular, the set of modules selected by $\greedy$ when it inspects module $i'$ is identical when the price vector is $(\vecc(i') + \eps, \vecp_{-i'})$ or $(\vecc(i') + \eps, \vecp'_{-i'})$.
Thus, since $i'$ did not have a profitable deviation to $\vecc(i') + \eps$ in $\vecp$, we conclude that it also does not have a profitable deviation to $\vecc(i') + \eps$ in $\vecp'$.

Now consider a module $i' \in S_{\vecp'}$.
Suppose that module $i'$ has a profitable deviation to $\vecp(i') + \eps$ when the price vector is $\vecp'$ (recall that $\vecp(i') = \vecp'(i')$ for all $i' \neq i$). We claim that this is only a profitable deviation for $i'$ when the price vector is $\vecp$ which contradicts that $\vecp$ is an $\eps$-equilibrium.
To see this, note that $\vecp' \leq \vecp$.
Thus, module $i'$ with its deviation, can only come earlier in the order inspected by $\greedy$ in $\vecp$ than in $\vecp'$ and the set of accepted modules so far in $\vecp$ must be a subset of the set of accepted modules so far in $\vecp'$.
Thus, module $i'$'s deviation must also have been profitable in $\vecp$.

\paragraph{Case 2: $\vec c(i) + \vecp(S_{\vec p}\cap \{\pi'(1),\dots,\pi'(j'-1)\}) > B$.} Let $T \coloneqq S_{\vec p} \setminus S_{\vec c(i) + \varepsilon, \vec p_{-i}}$ be the set of modules which are dropped from the selected set of modules after module $i$ updated its price to $\vecc(i) + \varepsilon$. 
Here, similar to the first case,  we observe that $\pi(k') = \pi'(k')$ for all $k' < j'$. Hence, we have $S_\vec p \cap \{\pi(1),\dots, \pi(j-1)\} = S_{\vec c(i) + \varepsilon , \vec p_{-i}} \cap \{\pi(1),\dots, \pi(j-1)\}$.  Therefore, $T \subseteq M \setminus \{\pi(1),\dots, \pi(j-1)\}$.
Define the price vector $\bar{\vecp}$ as
\[
\bar{\vecp} = \begin{cases}
\max \left \{  \frac{\vec p(i) + \varepsilon}{\vec v(i)}\cdot \vec v(k) , \vec c(k)  \right\} & k \in T \\
\vecc(i) + \eps & k = i \\
\vecp(k) & k \notin T \cup \{i\}
\end{cases}.
\]
We claim that the price $\vec {\bar p} $ is an $\eps$-equilibrium price. 

We let $\bar \pi$ be the bang-per-buck ordering of modules w.r.t.~the price $\bar{\vec p}$. First, we notice that $\vec {\bar p} \leq \vec p$ point-wise. We further observe that $\pi(k) = \bar \pi(k)$ for all $k =1,\dots,  j'-1$. We consider any module $\bar i$ with $\frac{\vec v(i)}{\bar {\vec p}(i)} < \frac{\vec v(\bar i)}{\bar {\vec p}(\bar i)}$.
By construction of $\bar p$, we have that $\bar{\pi}(\bar{i}) < \bar{\pi}(i)$.
Therefore, we have that $\bar i\in S_\vec p$ if and only if $\bar i \in S_{\bar {\vec {p}}}$ as $\{\pi(1), \dots, \pi(j'-1)\} = \{\bar \pi(1), \dots, \bar \pi(j'-1)\}$. Therefore, if there exists a profitable deviation of at least $\eps$ for module $\bar i$ at price $\bar {\vec p}$ then the same price deviation is profitable of at least $\eps$ for module $\bar i$ at price $\vec p$. This implies that there is no profitable deviation of at least $\eps$ for module $\bar i$ at price $\vec {\bar p}$. 

Next, we consider module $\bar i$ such that $\frac{\vec v(i)}{\bar {\vec p}(i)}> \frac{\vec v(\bar i)}{\bar {\vec p}(\bar i)}$ and $\bar i\notin S_{\vec p}$. Since $\vec c(i) + \vecp(S_{\vec p}\cap \{\pi'(1),\dots,\pi'(j'-1)\}) > B$ and  $\pi(k) = \bar \pi(k)$ for all $k =1,\dots,  j'-1$, we have that $\bar i \notin S_{\vec {\bar p}}$. Since $\bar {\vec p} \leq  \vec p$ point-wise, we have that  if there exists a profitable deviation (of at least $\eps$) for module $\bar i$ at price $\bar {\vec p}$ then the same price deviation is profitable for module $\bar i$ at price $\vec p$ (of at least $\eps$). This implies that there is no profitable deviation for module $\bar i$ at price $\vec {\bar p}$ of at least $\eps$.

We then consider the module $\bar i$ such that $\frac{\vec v(i)}{\bar {\vec p}(i)}> \frac{\vec v(\bar i)}{\bar {\vec p}(\bar i)}$ and $\bar i\in S_{\vec p}$. We observe that $\bar i\in T$ and since $\frac{\vec v(i)}{\bar {\vec p}(i)}> \frac{\vec v(\bar i)}{\bar {\vec p}(\bar i)}$, it must be the case that $\vec {\bar {p}}(\bar i) = \vec c(\bar i)$. 

We finally consider the module $\bar i$ such that $\frac{\vec v(i)}{\bar {\vec p}(i)} = \frac{\vec v(\bar i)}{\bar {\vec p}(\bar i)}$. We first observe that $\bar i\in T$ and $\frac{\vec v(\bar i)}{\vec c(\bar i)} > \frac{\vec v(\bar i)}{\vec {\bar p}(\bar i)}  = \frac{\vec v (i)}{\bar {\vec p(i)}}$. Therefore, module $\bar i$ comes before module $i$ in the ordering $\bar \pi$. We claim that the module $\bar i \in S_{\vec {\bar p}}$. 

We observe that $\bar i$ cannot form a circuit with modules with higher bang-per-buck at price $\bar p$ than $\frac{\vec v(\bar i)}{\vec {\bar p}(\bar i)}$. For the sake of contradiction, let $\bar i$ forms a circuit $C$ with modules with higher bang-per-buck  than $\frac{\vec v(\bar i)}{\vec {\bar p}(\bar i)}$. We note that $C \subseteq \{\bar \pi_{1}, \dots, \bar \pi(j'-1)\}\cup T$. Let $i^*:= \argmin_{k\in T\cap C} \frac{\vec v(k)}{\vec p(k)}$. Since, $i^*$ forms a circuit $C$ with modules with higher bang-per-buck  than $\frac{\vec v( i^*)}{\vec { p}(i^*)}$, $i^*\notin S_{\vec {\bar p}}$ which is a contradiction. 

Let $k$ be the index of module $\bar i$ on the order $\bar \pi$. We observe that $\{\bar \pi(1),\dots , \bar \pi(k-1)\}\subseteq \{\bar \pi(1),\dots , \bar \pi(j'-1)\}\cup T$. Next, we claim that $\vec {\bar p}(S_{\vec {\bar p}} \cap \{\bar \pi(1),\dots , \bar \pi(k-1)\}) + \vec {\bar p}(k)\leq B$. This follows since, $$(S_{\vec {\bar p}} \cap \{\bar \pi(1),\dots , \bar \pi(k-1)\})\cup \bar i \subseteq S_{\vec p} \cap \{\bar \pi(1),\dots , \bar \pi(j'-1)\}\cup T \subseteq S_{\vec p},$$ we conclude that $\vec {\bar p}(S_{\vec {\bar p}} \cap \{\bar \pi(1),\dots , \bar \pi(k-1)\}) + \vec {\bar p}(k)\leq B$. This implies that $\bar i \in S_{\vec {\bar p}}$. 

Next, we want to show that there is no profitable deviation for the module $\bar i$ at price $\bar p$. If module $\bar i$ increases its price to $\vec {\bar p}(i) + \eps$ then $\frac{\vec v(\bar i)}{\vec {\bar p}(\bar i) + \eps} < \frac{\vec v( i)}{\vec {\bar p}(i) }$. Since $\vec c(i) + \vecp(S_{\vec p}\cap \{\pi'(1),\dots,\pi'(j'-1)\}) > B$, $\greedy$ will not select module $\bar i$. This concludes that there is no profitable deviation for the module $\bar i$. Combining above arguments, we conclude that $\vec {\bar p}$ is an $\eps$-equilibrium price.  

In both cases, we observe that $ S_{\bar{\vec  p}}\subseteq  S_{\vec p}$ where $\vec{ \bar{p}}(i) \leq \vec c(i) + \varepsilon$. This completes the proof of the lemma. 
\end{proof}

Due to the above lemma, we immediately obtain the following corollary. 
\begin{corollary}
Let $\vec p$ be an $\eps$-equilibrium price for the price competition game instance $\instance$.
There exists an $\eps$-equilibrium price $\vec{\bar p}$ such that $\vec {\bar {p}}(i) \leq \vec c(i)+\varepsilon$ for all $i \notin S_{\bar{\vecp}}$ and $V(S_{\bar {\vec p}}) \leq V(S_\vec p)$.
\end{corollary}
Due to the above corollary, in order to bound the approximation ratio, it is sufficient to focus on the case where all non-selected modules at equilibrium price bid their private cost.
Next, we can characterize the equilibrium prices for the matroid rank valuations in the following lemma and show that all equilibrium of the pricing game have a constant approximation ratio.

\begin{lemma}
\label{lemma:unweighted_greedy_approx}
Fix an instance $\instance$ of the price competition game where $\I$ denotes a matroid.
Let $\lambda = \max_i \vecc(i) / B$ denote the maximum normalized cost of any module.
There is a sufficiently small $\eps_0 > 0$ (depending on $\vecc, \vecp$) such that for all $\eps \in (0, \eps_0)$, the following holds.
Let $\vecp$ be an $\eps$-equilibrium price and let $S_{\vecp}$ be the set of modules selected by $\greedy$ when $\vecp(i) \leq \vecc(i) + \eps$ for all $i \notin S_{\vecp}$.
Then
\begin{enumerate}
\item $S_{\vecp} \subseteq S_{\vecc}$;
\item there exists $i \in S_{\vecc}$ such that $S_{\vecp} = S_{\vecc} \cap \pi_{\vecc}[\pi_{\vecc}^{-1}(i)]$; and 
\item $\vecv(S_{\vecc}) \leq \left( 1 + \frac{B}{B-\lambda B-\eps} \cdot \frac{\lambda B+\eps}{\lambda B} \right) \vecv(S_{\vecp}).$
\end{enumerate}
\end{lemma}
The first assertion of Lemma~\ref{lemma:unweighted_greedy_approx} shows that the set of modules selected by $\greedy$ in $\vecp$ is a subset of the set of modules selected by $\greedy$ in $\vecc$.
This is used as a stepping stone to prove the second assertion: that $\greedy$ in $\vecp$ actually selects a \emph{prefix} of the modules selected by $\greedy$ in $\vecc$.
Using this, we will prove that $\greedy$ in $\vecp$ gets, approximately, a $2$-approximation to $\greedy$ in $\vecc$.
\begin{proof}
We will choose $\eps_0$ such that for all $i \in [n-1]$, we have
\[
    \frac{\vecv(i)}{\vecc(i) + \eps_0} > \frac{\vecv(i+1)}{\vecc(i+1)}.
\]
Such an $\eps_0 > 0$ exists since we assume that $\frac{\vecv(i)}{\vecc(i)} > \frac{\vecv(i+1)}{\vecc(i+1)}$ for all $i \in [n-1]$.

Now fix $\eps \in (0, \eps_0)$. We begin with the first assertion that $S_{\vecp} \subseteq S_{\vecc}$.
For the sake of contradiction, suppose that $S_{\vecp} \setminus S_{\vecc} \neq \emptyset$ and let $i \in S_{\vecp} \setminus S_{\vecc}$.
We consider two cases and show that neither case can happen.
\paragraph{Case 1: $\vecc(S_{\vecc} \cap \pi_{\vecc}[\pi_{\vecc}^{-1}(i) - 1]) + \vecc(i) > B$.}
Observe that we have $\vecp(S_{\vecp} \cap \pi_\vecp[\pi_\vecp^{-1}(i)]) \leq B$ since we know that $i \in S_{\vecp}$.
For notation, let us write $S_{\vecc}' = S_{\vecc} \cap \pi_{\vecc}[\pi_{\vecc}^{-1}(i)-1]$ and $S_{\vecp}' = S_{\vecp} \cap \pi_{\vecp}[\pi_{\vecp}^{-1}(i) - 1]$.
We claim that $S_{\vecc}' \setminus S_{\vecp}' \neq \emptyset$.
If not then we would have $\vecc(S_{\vecc}' \cup \{i\}) \leq \vecp(S_{\vecp}' \cup \{i\}) \leq B$ which contradicts the premise of the case.
Choose an arbitrary module $i' \in S_{\vecc}' \setminus S_{\vecp}'$.
Note that module $i'$ comes before module $i$ in the bang-per-buck order according to $\vecc$ and if module $i'$ price is in $[\vecc(i'), \vecc(i') + \eps]$ then this ordering remains unchanged.
Further note that $\pi_{\vecc}[\pi_{\vecc}^{-1}(i')-1] \supseteq \pi_{\vecp}[\pi_{\vecp}^{-1}(i') - 1]$ as $\vecp \geq \vecc$ coordinate-wise.
As $i' \in S_{\vecc}'$, we know that $i' \notin \spn(\pi_{\vecc}[\pi_{\vecc}^{-1}(i')-1])$ and thus  $i' \notin \spn(\pi_{\vecp}[\pi_{\vecp}^{-1}(i')-1])$.
This means that $i'$ is matroid-feasible in $\vecp$ when inspected by $\greedy$ but as it was not selected, the budget constraint must have already been violated which must have caused $\greedy$ to terminate. As module $i$ comes after module $i'$, according to both $\pi_{\vecc}$ and $\pi_{\vecp}$, this contradicts that $i \in S_{\vecp}$.
So the premise of the present case is impossible.

\paragraph{Case 2: $\vecc(S_{\vecc} \cap \pi_{\vecc}[\pi_{\vecc}^{-1}(i) - 1]) + \vecc(i) \leq B$.}
In this case, module $i \in \spn(S_{\vecc} \cap \pi_{\vecc}[\pi_{\vecc}^{-1}(i)-1])$.
Let $C$ be the unique circuit in $S_{\vecc} \cap \pi_{\vecc}[\pi_{\vecc}^{-1}(i) - 1] \cup \{i\}$.
Note that module $i$ has the lowest bang-per-buck, according to $\vecc$, in $C$.
Since $i \in S_{\vecp}$ there must be another module $i' \in C$ such that $i' \notin S_{\vecp}$.
We know that $\vecp(i') \leq \vecc(i') + \eps$ and $\frac{\vecv(i')}{\vecp(i')} \geq \frac{\vecv(i')}{\vecc(i') + \eps} > \frac{\vecv(i)}{\vecc(i)}$.
Thus $i'$ comes before $i$ in $\vecp$.
As in the previous case, since $i'$ is independent in $\vecc$ when inspected by $\greedy$, it is also independent in $\vecp$ when inspected by $\greedy$.
Thus, $i'$ must have either been accepted by $\greedy$ or the budget constraint is violated.
In the former case, this contradicts that $i' \notin S_{\vecp}$ while in the latter case, this would contradict that $i \in S_{\vecp}$ since $\greedy$ would have terminated before inspecting module $i$.
In either case, the premise of this case is also impossible.

The two cases above prove that $S_{\vecp} \subseteq S_{\vecc}$.
We now prove the second assertion.
Suppose that the second assertion is false.
This means that there is some $i \in S_{\vecp}$ and some $i' \in S_{\vecc} \setminus S_{\vecp}$ such that $\pi_{\vecc}^{-1}(i') < \pi_{\vecc}^{-1}(i)$.
In other words, module $i'$ has better bang-per-buck than module $i$ (according to $\vecc$) but module $i'$ is not selected by $\greedy$ in $\vecp$ while module $i$ is selected by $\greedy$ in $\vecp$.
Following a similar argument in the above case analysis, it must be that module $i'$ is inspected by $\greedy$ before module $i$ when the price vector is $\vecp$ and $i'$ must be independent of the set selected by $\greedy$ thus far.
Thus if $i'$ is not selected then it must be that the budget constraint would have been violated which further implies that $i$ could not have been selected.
This would contradict that $i \in S_{\vecp}$.
This proves the second assertion.

Finally, we prove the last assertion.
If $S_{\vecp} = S_{\vecc}$ there is nothing to prove.
So let $i$ denote the module in the second assertion of the lemma and let $j = \argmax_{j' \in S_{\vecc} \setminus S_{\vecp}} \{\vecv(j') / \vecc(j')\}$; this is the next module in the bang-per-buck order according to $\vecc$.
By following a similar argument in the proof of the first assertion (particularly, for case 1),
since $j \notin \spn(S_{\vecc} \cap \pi_\vecc[\pi_\vecc^{-1}(i)])$, we also have $j \notin \spn(S_{\vecp} \cap \pi_\vecp[\pi_\vecp^{-1}(i)])$.
Thus, the reason that $j$ is not accepted by $\greedy$ in $\vecp$ is because the budget constraint is violated if module $j$ is taken.
Since $\vecp(j) \leq \vecc(j) + \eps \leq \lambda B + \eps$, we have $\vecp(S_{\vecp}) \geq B - \vecp(j) \geq B - \lambda B - \eps$.
Note also that $\frac{\vecv(j')}{\vecp(j')} \geq \frac{\vecv(j)}{\vecc(j) + \eps}$ for all $j' \in S_{\vecp}$ by definition of $\greedy$.
We conclude that
\begin{align*}
    \vecv(S_{\vecp})
    & \geq \frac{\vecv(j)}{\vecc(j) + \eps} \cdot (B - \lambda B - \eps) \\
    & = \frac{\vecv(j)}{\vecc(j)} \cdot \frac{\vecc(j)}{\vecc(j)+\eps} \cdot (B - \lambda B - \eps) \\
    & \geq \frac{\vecv(j)}{\vecc(j)} \cdot \frac{\lambda B}{\lambda B+\eps} \cdot (B - \lambda B - \eps).
\end{align*}
Next, we bound $\vecv(S_{\vecc} \setminus S_{\vecp})$.
Note that the highest bang-per-buck module (according to $\vecc$) in $S_{\vecc} \setminus S_{\vecp}$ is module $j$.
This gives the trivial bound that
\[
    \vecv(S_{\vecc} \setminus S_{\vecp})
    \leq B \cdot \frac{\vecv(j)}{\vecc(j)}
    \leq \frac{B}{B-\lambda B-\eps} \cdot \frac{\lambda B+\eps}{\lambda B} \cdot \vecv(S_{\vecp}).
\]
We conclude that
\[
    \vecv(S_{\vecc})
    = \vecv(S_{\vecp}) + \vecv(S_{\vecc} \setminus S_{\vecp})
    \leq \left( 1 + \frac{B}{B-\lambda B-\eps} \cdot \frac{\lambda B+\eps}{\lambda B} \right) \vecv(S_{\vecp}). \qedhere
\]
\end{proof}

For the remainder of this section, we assume that $\vecv(i) = 1$ for all modules $i$.
\begin{lemma}
Suppose that $\vecv(i) = 1$ for all modules $i$.
Further, suppose that $\vecc(i) / B \leq \lambda$ for all $i$.
Let $S_{\vecc}$ be the set selected by $\greedy$ when the costs are $\vecc$ and let $S_{\OPT} \in \argmax \{ |S| \,:\, \vecc(S) \leq B, S \in \I \}$. Then $|S_{\OPT}| \leq \left( 1 + \frac{\lambda}{1-\lambda} \right) |S_{\vecc}|$.
\end{lemma}
\begin{proof}
If $\greedy$ spends strictly less than $(1-\lambda) B$ then it must be that $|S_{\vecc}| = |S_{\OPT}|$ otherwise there was an independent module that could have been added, as the cost of all modules is strictly less than $\lambda B$.

So now suppose that $\vecc(S_{\vecc}) \geq (1-\lambda) B$.
Suppose that modules are sorted such that $\vecc(1) \leq \ldots \leq \vecc(n)$ and let $k = |S_{\vecc}|$.
We thus have that $k \geq \frac{(1-\lambda)B}{\vecc(k)}$.
Next, an upper bound on $|S_{\OPT}|$ is to take the first $k+1$ modules that form an independent set.
This will violate the budget constraint (otherwise $\greedy$ would have picked it).
Thus, $|S_{\OPT}| - |S_{\vecc}| \leq 1 \leq \frac{\lambda B}{\vecc(k^*)} \leq \frac{\lambda}{1-\lambda} |S_{\vecc}|$.
Rearranging gives the lemma.
\end{proof}

By combining the above two lemmas we have the following.
\begin{theorem}
\label{theorem:unweighted_approx}
Fix an instance $\instance$ where $\vecv(i) = 1$ for all $i$.
Further, suppose that $\vecc(i) / B \leq \lambda$ for all $i$.
Let $\eps > 0$ be sufficiently small, as required by Lemma~\ref{lemma:unweighted_greedy_approx}.
Let $\vecp$ be an $\eps$-equilibrium price, let $S_{\vecp}$ be the set selected by $\greedy$ when the prices are $\vecp$, and let $S_{\OPT} \in \argmax \{ |S| \,:\, \vecc(S) \leq B, S \in \I \}$. Then $|S_{\OPT}| \leq \left(1 + \frac{\lambda}{1-\lambda} \right)\left( 1 + \frac{1}{1-\lambda - \eps/B}\cdot \frac{\lambda + \eps / B}{\lambda} \right) \cdot |S_{\vecp}|$.
\end{theorem}
In particular, for any fixed instance, as $\eps \to 0$, the approximation ratio tends to $\frac{2-\lambda}{(1-\lambda)^2}$.
\section{Price Competition for Weighted Buyer with Matroid Constraints}
\label{app:weighted_matroid}
In this section, we consider the pricing game when the platform uses Algorithm~\ref{alg:greedy} for module selection.
\subsection{Proof of Claim~\ref{claim:weighted_invarient}}\label{proof_of_weighted_invariant}
The invariants are trivial if $k = 1$ since $A^{k-1}  = \emptyset$.
We now assume that the invariant holds at iteration $k-1$ and prove that it remains true at iteration $k$.

For the first invariant, the only modules that increase their price is $T^{k-1} \subseteq \pi^{k-1}([k-1])$.
Since their bang-per-buck is at most that of module $\pi^{k-1}([k])$ this means that all modules in $\pi^{k-1}([k-1])$ still come before module $\pi^{k-1}([k])$.
So $\pi^{k-1}([k-1]) = \pi^k([k-1])$.
The fact that $\pi^{k-1}(k) = \pi^{k}(k)$ now easily follows because ties are broken according to $\pi^0$.

The second invariant is because prices of modules in $\pi^k([k-1])$ are only raised to meet the bang-per-buck of module in $\pi^k(k)$.

The third invariant is trivial if Line~\ref{line:budget_check} is false.
Otherwise, it is true because (i) $\vecp^{k-1}(A^{k-1})\leq B$ by the fact that the invariant holds for $k-1$.

At iteration $k$, if $\rank(\pi[k]) = \rank(\pi[k-1]) + 1$ then $\pi(k) \notin \spn(\pi[k-1])$. By inductive hypothesis, since $A^{k-1}$ is a full rank set in the matroid restricted to $\pi[k-1]$, module $\pi(k) \notin \spn(A^{k-1})$. Therefore if condition in Line~\ref{line:if_weighted} is satisfied and algorithm adds module $\pi(k)$ into the set $A^k$. This implies $\rank(A^k) = \rank(\pi[k])$. 
We now consider the case when $\rank(\pi[k]) = \rank(\pi[k-1])$. This implies that  $\pi(k) \in \spn(\pi[k-1])$  By inductive hypothesis, since $A^{k-1}$ is a full rank set in the matroid restricted to $\pi[k-1]$, module $\pi(k) \in \spn(A^{k-1})$. Therefore, $\pi(k)$ forms a unique circuit $C^k$ with $A^{k-1} \cup \pi[k]$. Since, $A^k \gets (A^{k-1} \cup \pi(k)) \setminus \{i\}$ removes a module from a circuit $C^k$, $\rank(A^k) = \rank(A^{k-1}) = \rank(\pi[k])$.  Hence, the forth invariant holds. 

To prove the fifth invariant, due to inductive hypothesis, we have that the module $i$ forms a circuit $C$ with the set of module $A^{k-1}$ where $i$ has the lowest value among the modules in $C$.
At round $k$, if $C\subseteq A^k$ then the invariant holds trivially. Otherwise, it must be the case that some module $j\in C$ was replaced by module $\pi^k(k)$.
This implies that the module $\pi^k(k)$ forms a circuit $C^k$ with $A^{k-1}$ that contains module $j$ and $j$ is the lowest valued module on circuit $C^k$.
We note that $i\notin C^k$.
By circuit axiom of matroids \cite[Lemma~1.1.3]{oxley2006matroid}, we have a circuit $C'\subseteq (C \cup C^k)\setminus j$, therefore, $\rank((C \cup C^k)\setminus j) \leq  |(C \cup C^k)\setminus j| -1$.
Since, fourth invariant implies that $A^k$ is a full rank set and $(C\cup C^k)\setminus \{i,j\} \subseteq A^k$, we have $\rank((C\cup C^k)\setminus \{i,j\}) = |(C\cup C^k)\setminus \{i,j\} |$.
We next claim that $i \in \spn((C\cup C^k)\setminus \{i, j\} )$.
If not then we have, $$\rank((C \cup C^k)\setminus j) = \rank((C \cup C^k)\setminus \{i, j\}) + 1 = |(C \cup C^k)\setminus \{i, j\}| + 1 = |(C \cup C^k)\setminus j|. $$ This contradicts the circuit axiom and therefore, $i \in \spn(C^k \cup C \setminus \{i,j\})$.
We note that $C^k \cup C \setminus \{i,j\} \subseteq A^k$ and all the modules in $C^k \cup C \setminus \{i,j\}$ has value higher than the value of module $i$. This concludes the proof. 

Next, we characterize the selected module at price $\bar {\vec p}$ which is an output of Algorithm~\ref{alg:equilibrium_dynamics_weighted}.
\begin{claim}
    \label{claim:matroid_greedy_weighted_output}
    $S_{\bar{\vecp}} = A^{k^*}$. In addition, $\SE$ is a maximum weight independent set in matroid $\mathcal I$ restricted to $\pi^0(1), \dots , \pi^0(k^*)$. 
\end{claim}
\begin{proof}
    Claim~\ref{claim:weighted_invarient} shows that $S_{\bar{\vecp}} \cap \pi^{k^*}([k^*-1]) = A^{k^*-1} $. The fifth invariant of Claim~\ref{claim:weighted_invarient} implies that for any $k<k^*$, the set of modules $A^k$ forms the maximum weight independent set in the matroid $\mathcal I$ restricted on modules $\pi^0[k]$ \cite[Theorem~1.8.5]{oxley2006matroid}. Now consider the greedy algorithm (Algorithm~\ref{alg:greedy}): due to invariant 2 from Claim~\ref{claim:weighted_invarient}, we observe that it iterates over $\pi^0(1), \dots, \pi^0(k^*-1)$ and selects maximum weighted independent set in the matroid $\mathcal I$ restricted on $\pi[k^*-1]$ if it does not exhaust the budget which is precisely the set $A^{k^* -1 }$. 
    
    We need to show that for any set of selected modules $S^{k'}$ at $k'$-th iteration of the of the  greedy algorithm, $\vecp(S_{k'})< B$. Consider any module $i\in S_{k'}$, we observe that the module $i$ will eventually be swapped by some module $i' \in A^{k}$. Since module $i'$ replaces the module $i$, we have $\vec v(i') > \vec v(i)$. In addition, due to bang-per-buck ordering, we have $\frac{\vec v(i)}{\vec {\bar p}(i)} \geq  \frac{\vec v(i')}{\vec {\bar p} (i')}$. This implies that $\vec {\bar p(i')} \geq \frac{\vec v(i')}{\vec v(i)} \cdot \vec {\bar p}(i) > \vec {\bar p}(i)$.  Due to our assumption $k<k^*$, we have $\vec {\bar p})(S_{k'}) \leq \vec {\bar p(A^k)}\leq B$.

    Then the greedy algorithm visits module $\pi^0(k^*)$. By definition of $k^*$, either $k^*= n$ and the budget constraint is never violated in Line~\ref{line:budget_check_weighted} or budget check condition (Line~\ref{line:budget_check_weighted}) is not satisfied after updating prices at round $k^*$. In the first case, we can observe that $A^n$ is the maximum weight independent set in $\mathcal I$ and $\vec {\bar p}^n(A^n)\leq B$. Therefore, greedy algorithm selects modules $A^n$. 
    In the second case, $A^{k^*} = A^{k^*-1}$ which is maximum weight independent set in matroid $\mathcal I $ restricted on $\pi[k^*-1]$ with $\vec {\bar p}^{k^*}(A^{k^* }) = \vec {\bar p}^{k^*}(A^{k^* -1 })\leq B$. In addition, we must have either $\pi^0(k^*) \notin \spn(A^{k^*-1})$ or $\pi^0(k^*)$ forms circuit $C^*$ with $A^{k^*-1}$ such that $i \neq \pi^0(k^*)$ is not the minimum weight element on $C^k$. In both cases, $\vec {\bar p}^k(A^*\cup \pi^0(k^*)) > B$ or $\vec {\bar p}^k(A^*\cup \pi^0(k^*) \setminus i) > B$ respectively. Therefore, greedy algorithm does not select $\pi^0(k^*)$ and terminates by outputting $A^{k^*-1} = A^{k^*}$. That concludes the proof.    
\end{proof}

\subsection{Proof of Theorem~\ref{thm:weighted_matroid_eq_exists}}\label{proof_weighted_matroid_exists}
 \begin{proof}
     Let $\barp$ be the price computed by Algorithm~\ref{alg:equilibrium_dynamics_weighted} $\optbpb$ be the bang-per-buck of modules selected modules $\SE: A^{k^*}$ by greedy algorithm (Algorithm~\ref{alg:greedy}). Since, Algorithm~\ref{alg:equilibrium_dynamics_weighted} selected maximum weighted set in matroid $\I$ restricted to $\pi^0(1), \dots , \pi^0(k^*-1)$, we claim that no module $\pi^0(k')\notin \SE$ with $k'< k^*$ can deviate their price and get selected. For the sake of contradiction, let module $\pi^0(k')$ updates its price to $p'$ and gets selected by  greedy algorithm. We note that $\pi^0(k')$ forms a circuit $C$ with $\SE$ where $i$ has the lowest value. This implies that the greedy algorithm must have ran out of budget when it reaches to the module  with lowest bang-per-buck (in this case, module $i\in \SE\setminus \{i\}$ with lowest $\frac{\vec v(i)}{\vec c(i)}$) say module $i\in C$. Since $\frac{\vec v(\pi^0(k'))}{p'}\geq  \frac{\vec v(i)}{\barp(i)}$ and $\vec v(i) > \vec v(\pi^0(k'))$. This implies that $\barp(i) > p'$. Since $\barp(\SE)\leq B$ (invariant 3 in Claim~\ref{claim:weighted_invarient}), we conclude that such deviation for module $\pi^0(k')$ can not exist. 
     
     On the other hand, for all modules $i\notin \SE$ with $\frac{\vec v(i)}{\barp(i)}< \optbpb$, we have $\barp(i) = \vec c(i)$ due to Claim~\ref{claim:weighted_invarient}. Therefore, such modules can not deviate and get selected. Finally, to conclude the proof, we need to show that no module in $\SE$ can increase their price and get selected. We first consider the case where $\barp(\SE) = B$. In this case, if any module $i\in \SE$ increases its price then it will not be selected by greedy algorithm. On the other hand, if $\barp(\SE)< B$ then $\barp(\opt(\SE \cup k^*))>B$. Then if any module in $i\in \SE$ increases its price then greedy algorithm selects $\pi^0(k^*)$ and exhausts its budget when it reaches to module $i$. Therefore, no module in $\SE$ can increase its price and improve its utility.  Therefore, we conclude that $\barp$ is an equilibrium price. 
 \end{proof}

\subsection{Proof of Lemma~\ref{lemma:unselected_lt_k_bid_cost}}\label{proof_of_lemma_poa_high_bpb}
In this section, we assume that $\eps > 0$ is sufficiently small, i.e.~$\eps < \min_i \vecc(i)$.
We first start with a couple easy claims.

\begin{claim}
\label{claim:matroid_identical_bang_per_buck}
Fix $\eps < \min_i \vecc(i)$.
Let $\vecp$ be an $\eps$-equilibrium.
For all $i, j \in S_{\vecp}$, we have $\frac{\vecv(i)}{\vecp(i)+\eps} \leq \frac{\vecv(j)}{\vecp(j)}$.
\end{claim}
\begin{proof}
    Suppose that there exists $i, j \in S_{\vecp}$ such that $\frac{\vecv(i)}{\vecp(i)+\eps} > \frac{\vecv(j)}{\vecp(j)}$.
    In this case, suppose that module $i$ increases its price to $\vecp(i) + \eps$.
    As the prices have only increased with $i$'s deviation, $\greedy$ inspects a subset (not necessarily strict) of the modules it inspected when the prices are $\vecp$.
    Since $\eps < \vecc(j)$, $j \in S_{\vecp}$, and module $i$ still comes before module $j$ in the bang-per-buck ordering after deviating, it means that module $i$ is budget-feasible.
    In particular, module $i$ is still inspected and accepted by $\greedy$.
    This contradicts that $\vecp$ is an $\eps$-equilibrium.
\end{proof}

\begin{claim}
\label{claim:eq_price_budget_or_bpb_tight}
Let $\vecp$ be an $\eps$-equilibrium and suppose $\greedy$ selects $k^*$ modules.
Then at least one of the following is true. Either $\vecp(S_\vecp) \geq B - \eps$ or for all $i \in S_{\vecp}$, we have $\frac{\vecv(i)}{\vecp(i)+\eps} \leq \frac{\vecv(k^*+1)}{\vecc(k^*+1)}$.
\end{claim}
\begin{proof}
    If neither conclusion is true then a selected module could raise its price by $\eps$ to be in front of $k^*+1$ but within the budget constraint.
    So at least one of the conclusions is true.
\end{proof}

\begin{proof}[Proof of Lemma~\ref{lemma:unselected_lt_k_bid_cost}]
    Let $i \notin [k^*] \setminus S_{\vecp}$ such that $\vecp(i) > \vecc(i) + \eps$.
    We set $\bar{\vecp}(i) = \vecc(i) + \eps$ and $\bar{\vecp}(j) = \vecp(j)$ for $j \neq i$.
    
    We now check that $\greedy$ selects the same set of modules in $\vecp$ and $\bar{\vecp}$.
    If $i$ is never selected in $\bar{\vecp}$ then it must never be selected in $\vecp$ (since its price is higher). Thus, for all other modules, $\greedy$ behaves identically in $\vecp$ and $\bar{\vecp}$.
    If $i$ is selected then it must be rejected later.
    At this point, Claim~\ref{claim:matroid_first_k_opt} shows that $\greedy$ under both $\vecp$ and $\bar{\vecp}$ must be identical.
    So thereafter, $\greedy$ is identical under both $\vecp$ and $\bar{\vecp}$.

    We need to check that $\bar{\vecp}$ is an $\eps$-equilibrium price for all $j \neq i$.
    Suppose $j$ is not selected but it can reduce its price to a point where it is budget-feasible, considered by $\greedy$, and its price exceeds its cost by $\eps$.
    Let $j' \in S_{\vecp}$ be the module that would remove $j$ if $j$'s price is $\vecc(j) + \eps$ in the price vector $\vecp$.
    If $j$ chooses a deviation to come after $j'$ then $j$ cannot be selected since $j'$ is chosen.
    Suppose that $j$ chooses a deviation to come before $j'$.
    Let $A_{j'}$ be the modules coming before $j'$, including $j'$ in $\bar{\vecp}$.
    Since the total spend by $\greedy$ is monotone increasing in each iteration and since $S_{\bar{\vecp}} \cap A_{j'}$ is budget-feasible, it must be that $\greedy$ reaches $j'$.
    Thus $j$ would be removed in this case as well.
    
    If $j \leq k^*$ and is selected then it cannot increase its price to gain more than $\eps$ utility since the solution in $\vecp$ and $\bar{\vecp}$ are identical.
    In other words, a profitable deviation in $\bar{\vecp}$ would also be a profitable deviation in $\vecp$.
\end{proof}

\subsection{Proof of Lemma~\ref{lemma:unselected_gt_k_to_cost}}\label{proof_of_key_lemma_poa}
\begin{proof}
This proof is a case analysis.
\paragraph{Case 1: $\frac{\vecv(q)}{\vecc(q) + \eps} \leq \frac{\vecv(k^*+1)}{\vecp(k^*+1)}$.}
Since $\vecp$ is an $\eps$-equilibrium, this means that module $q$ is not accepted at price $\vecc(q) + \eps$.
Let $\bar{\vecp}$ be the new price vector with $q$'s deviation to $\vecc(q) + \eps$.
Any module $q' \neq q$ coming before $k^*+1$ is unaffected and any module that comes after $k^* + 1$ cannot deviate.
To see this, if it could have deviated in $\bar{\vecp}$ then it could also have deviated in $\vecp$ since the only difference between $\bar{\vecp}$ and $\vecp$ is that in $\bar{\vecp}$, module $q$ may come in front of module $q'$.

\paragraph{Case 2: $\min_{i \in S_{\vecp}} \frac{\vecv(i)}{\vecp(i) + \eps} \geq \frac{\vecv(q)}{\vecc(q) + \eps} > \frac{\vecv(k^*+1)}{\vecp(k^*+1)}$.}
In this case, we have $\vecp(S_{\vecp}) \geq B - \eps$ by Claim~\ref{claim:eq_price_budget_or_bpb_tight}.
We claim that changing the price of module $q$ to $\vecc(q)+\eps$ gives an $\eps$-equilibrium.
No rejected module $i < k^*$ can deviate since they are already within $\eps$ of their cost.
No accepted module can increase its price by $\eps$ since it would no longer be selected as the budget constraint would be violated.
Finally, no rejected module $i > k^*$ can deviate since it would have to come in front of $k^*$ and if it was selected then that would have been a profitable deviation in $\vecp$ as well.

The last case is the most technically challenging.
\paragraph{Case 3: $\frac{\vecv(q)}{\vecc(q) + \eps} > \min_{i \in S_{\vecp}} \frac{\vecv(i)}{\vecp(i) + \eps}$.}
First, let us set $\bar{\vecp}(q) = \vecc(q) + \eps$. We consider a few subcases.
\paragraph{Case 3a: When running $\greedy$ with $q$'s deviation, $q$ could have been added but the budget was exceeded when inspecting $q$.}
In this case, we first define a temporary price $\vecp'$ by setting $\vecp'(i) = \max\left\{ \vecc(i), \frac{\vecc(q) + \eps(1-\delta)}{\vecv(q)} \cdot \vecv(i) \right\}$ for all previously accepted but now rejected modules $i$, i.e.~the new bang-per-buck of $i$ is equal to $q$ if $i$ can accept that price.
Here, $\delta \in (0, 1)$ will be a small parameter that we choose below.
Note that these modules come before $q$ in the bang-per-buck ordering.
Note that $q$ is still not accepted since now there are more modules in front of $q$.
However, it could be that some previously rejected modules are now accepted since some modules that were previously accepted cannot lower their price set by $q$.
For these modules, we raise their price so that their bang-per-buck is also $\frac{\vecv(q)}{\vecc(q) + \eps(1-\delta)}$ or their current price, whichever is greater.
We let $\bar{\vecp}$ be the price given by the three transformations discussed above (and set $\bar{\vecp}(i) = \vecp(i)$ for any module $i$ that was not discussed).

We claim that the maximum value independent set of all modules coming before $q$ is budget-feasible with the price vector $\bar{\vecp}$.
To see this, first note that the maximum value independent set has decreased in value from $\vecp$ to $\bar{\vecp}$.
Next, observe that any module that was accepted by $\greedy$ in both $\vecp$ and $\bar{\vecp}$ has its bang-per-buck either stay the same or it has increased.
On the other hand, any module that was accepted by $\greedy$ in $\vecp$ but not in $\bar{\vecp}$ may have been replaced by a module with better bang-per-buck.
We conclude that the maximum value independent set of all modules coming before $q$ must be budget-feasible with the price vector $\bar{\vecp}$.

Finally, we show that $\bar{\vecp}$ is indeed an $\eps$-equilibrium.
We require the following claim.
\begin{claim}
    \label{claim:Suq_contains_q}
    Let $S_q$ be the set selected by $\greedy$ with price vector $\bar{\vecp}$ just before $\greedy$ inspects $q$.
    If $\delta > 0$ is sufficiently small then the maximum value independent set in $S_q \cup \{q\}$ contains $q$.
\end{claim}
\begin{proof}
Let $A_q$ be the set of modules that $\greedy$ selects with only $q$'s deviation.
The premise of the case means that $A_q \cup \{q\}$ is independent but $\vecp(A_q) + \vecc(q) + \eps > B$.
Now suppose that the maximum value independent set in $S_q \cup \{q\}$ does not contain $q$.
We will show that $S_q$ must not be budget-feasible which would be a contradiction.
To do this, we will show that for sufficiently small $\delta > 0$, we have $\bar{\vecp}(S_q) - (\vecp(A_q) + \vecc(q) + \eps) > -\eta$ where $\eta = B - (\vecp(A_q) + \vecc(q) + \eps)$.
Indeed, we have
\begin{align*}
\bar{\vecp}(S_q) - (\vecp(A_q) + \vecc(q) + \eps)
& = \bar{\vecp}(S_q \cap A_q) + \bar{\vecp}(S_q \setminus A_q) - (\vecp(A_q) + \vecc(q) + \eps) \\
& \geq \vecp(S_q \cap A_q) + \bar{\vecp}(S_q \setminus A_q) - (\vecp(A_q) + \vecc(q) + \eps) \\
& = \bar{\vecp}(S_q \setminus A_q) - \vecp(A_q \setminus S_q) - \vecc(q) - \eps \\
& \geq \frac{\vecc(q) + \eps(1-\delta)}{\vecv(q)} \cdot \vecv(S_q \setminus A_q) - \frac{\vecc(q) + \eps}{\vecv(q)}{\vecv(q)} \cdot (\vecv(A_q \setminus S_q) + \vecv(q)) \\
& =  \frac{\vecc(q) + \eps}{\vecv(q)} \cdot \left( \vecv(S_q \setminus A_q) - \vecv(A_q \setminus S_q) - \vecv(q) \right) - \frac{\eps \delta \vecv(S_q \setminus A_q)}{\vecv(q)} \\
& = \frac{\vecc(q) + \eps}{\vecv(q)} \cdot \left( \vecv(S_q) - \vecv(A_q) - \vecv(q) \right) - \frac{\eps \delta \vecv(S_q \setminus A_q)}{\vecv(q)}.
\end{align*}
The first inequality is because modules in $S_q \cap A_q$ may have had their price increased from $\vecp$ to $\bar{\vecp}$.
The second inequality is because (i) any module in $S_q \setminus A_q$ is because the module had to have its price reduced after $q$'s deviation and (ii) all modules $i$ in $A_q$ (and thus in $A_q \setminus S_q$) have bang-per-buck at least $\frac{\vecv(q)}{\vecv(q) + \eps}$.
The last equality is obtained by adding and subtracting $\vecv(S_q \cap A_q)$.
Note that $\vecv(S_q) \geq \vecv(A_q) + \vecv(q)$ because $S_q$ is the maximum value independent set of all modules before and including $q$ in $\bar{\vecp}$ while $A_q \cup \{q\}$ is the maximum value independent set of all modules before and including $q$ in $\vecp$ with only $q$'s deviation (which is contained in the set of modules before and including $q$ in $\bar{\vecp}$).
We thus have
\[
\bar{\vecp}(S_q) - (\vecp(A_q) + \vecc(q) + \eps)
\geq -\frac{\eps \delta \vecv(S_q \setminus A_q)}{\vecv(q)}.
\]
A trivial upper bound on $\vecv(S_q \setminus A_q)$ is $\vecv([n])$ so taking $\delta < \frac{\eta \vecv(q)}{\eps \vecv([n])}$ is sufficient for the claim.
\end{proof}
Consider any module $i$ that is accepted.
If module $i$ deviates by more than $\eps$ then Claim~\ref{claim:Suq_contains_q} shows that either $q$ is accepted which would cause $i$ to be rejected or $\greedy$ terminates at $q$ which also means $i$ is rejected.

Now consider any module $i$ is not selected under $\bar{\vecp}$.
If $i < k^*$ was rejected in $\vecp$ and still rejected in $\bar{\vecp}$ then $i$ is bidding within $\eps$ of its cost.
If $i > k^*$ was rejected in $\vecp$ then any deviation in $\bar{\vecp}$ results in more modules in front of $i$ than in $\vecp$.
So $i$ also cannot deviate.

Finally, note that some modules $i < k^*$ which were rejected had their price increased above and then later rejected in $\bar{\vecp}$.
For these modules, we use Lemma~\ref{lemma:unselected_lt_k_bid_cost} to bring their bid back down to their cost.

\paragraph{Case 3b: When running $\greedy$ with $q$'s deviation, module $q$ is budget-feasible when $\greedy$ inspects it but is never accepted, even temporarily.}
This case is straightforward since $\greedy$ behaves identically in before and after $q$'s deviation.

\paragraph{Case 3c: When running $\greedy$ with $q$'s deviation, module $q$ is budget-feasible when $\greedy$ inspects it and is accepted but later rejected.}
In this case, we claim only $q$'s deviation is an $\eps$-equilibrium price.
First, we check that $S_{\vecp}$ remains budget-feasible.
Let $A$ be the set accepted by $\greedy$ up to but before $k^*$.
In this case, $\vecv(A) \leq \vecv(S_{\vecp})$ because, by Claim~\ref{claim:matroid_first_k_opt}, $A$ (resp.~$S_{\vecp}$) is the maximum value independent set in all the modules before $k^*$ excluding (resp.~including) $k^*$.
Moreover, we know that the price paid by $\greedy$ is always monotone increasing and $S_{\vecp}$ is budget-feasible.

Finally, we check that this is an $\eps$-equilibrium.
No accepted module can deviate because either the budget constraint is tight or its deviation would cause it to come after a module which violates the budget constraint.
Any rejected module $i < k^*$ is already bidding its cost so it cannot deviate.
Finally, let $i > k^*$ be rejected and suppose it deviates to its cost plus $\eps$.
We claim it must still be rejected.
If $\greedy$ never accepts $i$ then we are done.
Otherwise, suppose $\greedy$ does accept $i$ in $\bar{\vecp}$.
Then is must be in $\spn(S_{\vecp})$.
It also must be the case that in $\vecp$, $\greedy$ accepted $i$ but later rejected it.
Let $j$ be the last module inspected by $\greedy$ in the circuit formed by $S_{\vecp} \cup \{i\}$.
Then $j$ must be budget-feasible because $S_{\vecp}$ is and the set selected by $\greedy$ until $j$ is inspected is at most $\vecv(S_{\vecp})$.
So taking $j$ will reject $i$.
\end{proof}

\section{Proof of Theorem~\ref{thm:weighted_matroid_approx}}\label{proof_of_quality_matroid}
\begin{lemma}
\label{lemma:vecp_most_of_budget}
Let $\vecp$ be an $\eps$-equilibrium and let $k^*$ be the last module selected by $\greedy$.
Suppose that $\vecp(i) \leq \vecc(i) + \eps$ for all $i \notin S_{\vecp}$ and $i \leq k^* + 1$.
Then $\vecp(S_{\vecp}) \geq (1 - \lambda) B - \eps$ where $\lambda > \max_i \frac{\vecc(i)}{B}$.
\end{lemma}
\begin{proof}
    We assume that the maximum value set in $[k^* + 1]$ contains $k^* + 1$.
    If not and $\vecp(S_{\vecp}) < B - \eps$ then $\vecp$ is not an $\eps$-equilibrium since accepted modules can still raise their price past the bang-per-buck of $k^* + 1$.

    Let $p$ be the price of the module that $k^* + 1$ kicks out if it is added (with $p = 0$ if it does not kick out any module).
    Suppose that $\vecp(k^*+1) > \lambda \cdot B$.
    Then we claim that $\vecp(S_{\vecp}) \geq (1-\lambda)B - \eps$.
    If not, then let us take $\vecp'(k^* + 1) = B - \vecp(S_{\vecp}) + p$.
    Note that $B - \vecp(S_{\vecp}) + p \geq \lambda B + p + \eps > \vecc(k^*+1) + \eps$.
    So this means that $k^*+1$ has a profitable deviation contradicting that $\vecp(S_{\vecp}) < (1-\lambda) \cdot B - \eps$.
    
    On the other hand, if $\vecp(k^*+1) \leq \lambda \cdot B$ and $k^*+1$ was not taken then it must be that $\vecp(S_{\vecp}) \geq (1 - \lambda) B$.
\end{proof}

\begin{lemma}
Let $\vecp$ be an $\eps$-equilibrium and $\lambda = \max_i \vecc(i) / B$ and suppose that $\lambda < 1/3$.
Then
\[
    \vecv(S_{\vecc}) \leq \left( 1 + \frac{2\lambda}{1-3\lambda} + \frac{\eps / \lambda B}{1-\lambda - \eps / B} \right) \vecv(S_{\vecp}).
\]
\end{lemma}
\begin{proof}
We assume that the modules are sorted in bang-per-buck order according to $\vecp$, i.e.~$\frac{\vecv(1)}{\vecp(1)} \geq \ldots \geq \frac{\vecv(n)}{\vecp(n)}$.

By Lemma~\ref{lemma:unselected_lt_k_bid_cost} and Lemma~\ref{lemma:unselected_gt_k_to_cost}, we can assume that if $k^*$ is the last module selected by $\greedy$ then for all un-selected modules $i \neq k^* + 1$, we have $\vecp(i) \leq \vecc(i) + \eps$.

First observe that
\begin{equation}
    \label{eqn:Sp_lb}
    \vecv(S_{\vecp}) \geq \left[(1 - \lambda) B - \eps\right] \frac{\vecv(k^*)}{\vecp(k^*)},
\end{equation}
where we used Lemma~\ref{lemma:vecp_most_of_budget}.

Next we bound $\vecv(S_{\vecc} \setminus [k^*+1])$.
For any $i \geq k^* + 1$, we have that
\[
    \frac{\vecv(i)}{\vecc(i) + \eps} \leq \frac{\vecv(k^*)}{\vecp(k^*)}.
\]
Let $\delta = \frac{\eps}{\max_i c_i} = \frac{\eps}{\lambda B}$ so that $\eps \leq \delta \vecc(i)$ for all $i$.
We thus have $\frac{\vecv(i)}{\vecc(i) + \eps} \geq \frac{\vecv(i)}{(1+\delta) \vecc(i)}$ for all $i$ which implies that for any $i \geq k^* + 1$, we have
\[
    \frac{\vecv(i)}{\vecc(i)} \leq 
    (1+\delta) \cdot \frac{\vecv(k^*)}{\vecp(k^*)}.
\]
Thus,
\begin{equation}
    \label{eqn:Sc_ub}
    \vecv(S_{\vecc} \setminus [k^* + 1])
    \leq (1+\delta) \frac{\vecv(k^*)}{\vecp(k^*)} \cdot B.
\end{equation}

We now consider two cases.
\paragraph{Case 1: The maximum value independent set in $[k^*+1]$ does not contain $k^* + 1$.}
In this case, we have 
\[
    \vecv(S_{\vecp}) = \vecv(S_{\vecp} \cap [k^*+1])
    \geq \vecv(S_{\vecc} \cap [k^*+1])
\]
since Claim~\ref{claim:matroid_first_k_opt} shows that $\greedy$ on $\vecp$ finds the maximum value independent set in $[k^*]$ which, in this case, is also the maximum value independent set in $[k^* + 1]$.
We thus have
\begin{align*}
    \vecv(S_{\vecc}) - \vecv(S_{\vecp})
    & = \vecv(S_{\vecc} \cap [k^* + 1]) - \vecv(S_{\vecp}) + \vecv(S_{\vecc} \setminus [k^* + 1]) \\
    & \leq (1 + \delta) \cdot \frac{\vecv(k^*)}{\vecp(k^*)} \cdot B \\
    & \leq \frac{1+\delta}{1-\lambda - \eps / B} \vecv(S_{\vecp}),
\end{align*}
using Eq.~\eqref{eqn:Sp_lb} for the last inequality.
Rearranging, we have
\[
    \vecv(S_{\vecc}) \leq \left( 1 + \frac{1+\delta}{1-\lambda - \eps / B} \right) \cdot \vecv(S_{\vecp}).
\]
\paragraph{Case 2a: The maximum value independent set in $[k^*+1]$ does contain $k^* + 1$ and $\lambda < 1/5$.}
First, let $A = \left\{ i \in S_{\vecp} \,:\, \frac{\vecv(i)}{\vecp(i)} \leq \frac{\vecv(S_{\vecp})}{(1-3\lambda) B} \right\}$.
We claim that $\vecp(A) \geq \lambda B$.
Suppose, for the sake of contradiction, that $\vecp(A) < \lambda B$.
Then $\vecp(S_{\vecp} \setminus A) > (1-\lambda) B - 2\lambda B = (1-3\lambda)B$.
We would thus have
\[
    \vecv(S_{\vecp})
    \geq \vecv(S_{\vecp} \setminus A)
    > \frac{\vecv(S_{\vecp})}{(1-3\lambda)B} (1-3\lambda) B = \vecv(S_{\vecp}).
\]
This is absurd, so we conclude that $\vecp(A) \geq \lambda B$.

We now claim that $\vecv(k^*+1) \leq \frac{2\lambda}{1-3\lambda} \vecv(S_{\vecp})$.
Again, suppose for the sake of contradiction, that $\vecv(k^*+1) > \frac{2\lambda}{1-3\lambda} \vecv(S_{\vecp})$.
Then if it deviated to $2 \lambda B \geq 2 \vecc(k^*+1) \geq \vecc(k^*+1) + \eps$,
its bang-per-buck would be
\[
    \frac{\vecv(k^*+1)}{\vecc(k^*+1)}
    > \frac{\vecv(S_{\vecp})}{(1-3\lambda)B},
\]
which would put it in front of $A$ and there would be sufficient budget when $\greedy$ visits $k^*+1$.
Since $k^*+1$ is in the maximum value independent set in $[k^*+1]$, it must therefore be selected by $\greedy$.
This would contradict that $\vecp$ was an $\eps$-equilibrium.
We conclude that $\vecv(k^*+1) \leq \frac{2\lambda}{1-3\lambda} \vecv(S_{\vecp})$.

Now continuing as in case 1, we have
\begin{align}
    \vecv(S_{\vecc}) - \vecv(S_{\vecp})
    & = \vecv(S_{\vecc} \cap [k^* + 1]) - \vecv(S_{\vecp}) + \vecv(k^*+1) + \vecv(S_{\vecc} \setminus [k^* + 1]) \\
    & \leq \vecv(k^*+1) + (1 + \delta) \cdot \frac{\vecv(k^*)}{\vecp(k^*)} \cdot B \label{eq:recall_analysis}\\
    & \leq \left( \frac{2\lambda}{1-3\lambda} + \frac{1+\delta}{1-\lambda - \eps/B}\right) \vecv(S_{\vecp}). \qedhere
\end{align}
\end{proof}

\begin{lemma}\label{lem:greedy_approx}
Let $S_{\OPT} \in \argmax \{ S \,:\, \vecc(S) \leq B, S \in \I \}$.
Let $\lambda = \max_i \vecc(i) / B$.
Then $\vecv(S_{\OPT}) \leq \left( 1 + \frac{\lambda}{1-\lambda} \right) \cdot \vecv(S_{\vecc})$.
\end{lemma}
\begin{proof}
    Sort the modules such that $\frac{\vecv(1)}{\vecc(1)} \geq \ldots \geq \frac{\vecv(n)}{\vecc(n)}$.
    Let $k^*$ be the last module selected by $\greedy$.
    If has the same value as $S_{\OPT}$ then we are done.
    Otherwise, a standard fact is that taking the first $k^*+1$ modules that form an independent set is an upper bound on $\vecv(S_{\OPT})$.
    We have that $\vecv(S_{\vecc}) \geq \frac{\vecv(k^*)}{\vecc(k^*)} \cdot (1-\lambda) \cdot B$ and
    $\vecv(S_{\OPT}) - \vecv(S_{\vecc}) \leq \vecv(k^* + 1) \leq \frac{\vecv(k^*)}{\vecc(k^*)} \cdot \lambda B \leq \frac{\lambda}{1-\lambda} \vecv(S_{\vecc})$.
    Rearranging gives the lemma.
\end{proof}

Combining the results above, we have prove the main theorem.

\section{Missing Proofs of Section~\ref{sec:convergence_dynamics_of_learning}}
In this section, we present the missing proofs from Section~\ref{sec:convergence_dynamics_of_learning}.
In the price competition analysis, we described an algorithm to compute $\eps$-equilibrium prices for the modules for the price competition game instance $\instance$ where $\mathcal I$ is a matroid. However, the modules can not implement an $\eps$-equilibrium price without full information about the game instance $\instance$. In this section, we show that a simple price dynamics where each module implements a multiplicative weight type learning algorithm (Definition~\ref{def:multiplicative_weight_learning}) converges to an $\eps$-equilibrium.  

More formally, we consider a repeated setting of price competition game instance $\instance$ where $\mathcal I$ is a matroid and $B=1$. At each time step $t$, each module $m_i$ specifies a price $\vec p^t(i)$ from a discretized set $\mathcal B = \{\delta, 2\delta, \ldots, 1\}$. 
The platform then chooses a subset of the modules according to the greedy algorithm (Algorithm~\ref{alg:greedy}).
Importantly, the module owners are only aware of their own cost and feedback provided by the platform after each round, which includes whether they are accepted or rejected at a given price and what is the maximum price at which the module could be selected at the current round.
In particular, they do not know the buyer's valuation of their module or the existence of other modules in the market, other than what they can infer from the feedback described above. 

\paragraph{Assumptions on Discretization}Throughout the section, we assume that $\vec c(i) > \delta^{1/3}$. We note that we can assume this without loss of generality by letting the discretization be more finer. In addition, we assume that $1 \leq \vec v(i)\leq n^2$ and $\delta < \frac{1}{n^3}$. Finally, to avoid tie-breaking at the equilibrium price, when $\barp(\SE)< B$ and $\SE'$ being the optimal set of modules in matroid $\mathcal I $ restricted to $\pi[k^*]$, we have $\barp(\SE' )> B + 2\cdot \sqrt \delta$. Here, $k^*$ is the last iteration of Algorithm~\ref{alg:equilibrium_dynamics_weighted}.

\subsection{Proof of Lemma~\ref{lem:eqm_price_dominates}}
We observe that for any module $j\in H$, $\frac{\vec v(j)}{\vec c(j)} \leq  \frac{\vec v(i)}{\barp(i)}$ for $i\in \SE$. Therefore, at any price vector $\vec p_{-i}$, $\frac{\vec v(j)}{\vec p(j)} \leq  \frac{\vec v(i)}{\barp(i)}$.

Let $\pi$ be the bang-per-buck ordering over the set of modules at price $\vec p$ and $\bar S$ be the set of modules selected by the greedy algorithm at price $\vec p$. At price $\vec p$, we observe that when the greedy algorithm (Algorithm~\ref{alg:greedy}) reaches module $i$, it adds module $i$ in the working solution $G$. Since $i$ gets selected at price $\vec p$, we have that no module with a lower bang-per-buck than the module can swap module $i$ from the working solution $G$. We let $\pi(\bar i) \in H$ be the module with the highest bang-per-buck such that module $i$ does not belong to the maximum weight independent set in matroid $\mathcal I$ restricted on $\pi(1),\dots, \pi(\bar i)$. We note the greedy algorithm does not reach the module $\pi(\bar i)$ otherwise module $i$ could have been swapped out from the set $\bar S$.

Let $\pi'$ be the bang-per-buck ordering over the modules at price $\vec p'$ and $S_{\vec p'}$ be the set of selected modules at price $\vec p'$. Next, we claim that at price $\vec p'$, greedy algorithm (Algorithm~\ref{alg:greedy}) reaches the module $i$ and adds it into the working solution. Let $L'$ be the set of modules with bang-per-buck $\leq \frac{\vec v(i)}{\barp(i)}$ at price $\vec p'$. Since $H\cap L' = \emptyset$, we have that $L'\subseteq L$. Since, $i\in S_\vec p$, if the greedy algorithm (Algorithm~\ref{alg:greedy}) reaches module $i$ then it adds module $i$ in the working solution.

Let $\pi(\bar i) = \pi'(i')$. We note that since bang-per-buck of module $\frac{\vec v(\pi(\bar i))}{\vec p(\pi(\bar i))} < \frac{\vec v(i)}{\barp(i)}$, we have $\{\pi(1),\dots , \pi(\bar i)\} = \{\pi(1),\dots, \pi(i')\}$. In addition, the price of all modules that come after module $i$ on the bang-per-buck ordering $\pi'$ is unchanged from the price vector $\vec p$, we conclude that the greedy algorithm can not reach module $\pi(\bar i)$ on the ordering $\pi'$ which concludes the proof.

\subsection{Proof of Lemma~\ref{lem:module_with_worst_bpb_rejected}}
We let $\pi$ be the bang-per-buck over the modules at price $\vec p$. Let $\pi(\ell):= i^*$. We observe that $ \SE \subseteq  \pi[\ell -1]$. Let $S^*[\ell]$ be the maximum weight independent set in matroid restricted to $\pi[\ell -1]$ chosen by greedy algorithm (Algorithm~\ref{alg:greedy}). We claim that $\vec p(S^*[\ell]) >B$. 

First We note that $S^*[\ell] \cap L = \emptyset$ as $\SE \subseteq S^*[\ell]$. In addition, for any module $i\in S^*[\ell] \setminus \SE'$ must satisfy that $i\in H$. We note that $i\in \SE \setminus S^*[\ell]$ forms circuit with $S^*[\ell]$ and can be exchanged by module $j\in \SE' \setminus S^*[\ell]$ such that $j\in H$ and $\vec v(j)>\vec v(i)$. Note that for $j\in H$, $\frac{\vec v(j)}{\vec p(j)}< \frac{\vec v(i)}{\barp(i)}$ and since $\vec v(i)< \vec v(j)$, we have $\barp(i)<\vec p(j)$. This implies that $\vec p(S^*[\ell])> \vec p(\SE')> B$. Therefore, $i^*$ can not be selected by the greedy algorithm (Algorithm~\ref{alg:greedy}). This further implies that $\vec p(i^* \cup S_\vec p)>B$.

\subsection{Proof of Lemma~\ref{lem:stability_of_eqm_price}}
    We consider module $\pi^0(k^*)$ where $k^*$ denotes the last iteration of Algorithm~\ref{alg:equilibrium_dynamics_weighted}. Now, we consider module $i\in \SE$ that submits price $\geq \barp(i)(1 + \sqrt \delta)$. We first show that $\frac{\vec v(i)}{\barp(i) (1+ \sqrt \delta)} / \frac{\vec v(j)}{\barp(j) +10\delta} < 1$. This follows because $\vec c(i) > \sqrt \delta$ and $\frac{\vec v(i)}{\barp(i)} = \frac{\vec v(j)}{\barp(j)} = \optbpb$. Now,  if $\barp(\SE) = B$ then since rest of the modules comes before module $i$ in the bang-per-buck ordering, it can no-longer be affordable as $\vecp(\SE)> B - 10\delta \cdot n + \sqrt \delta >B + \sqrt \delta -\delta^{2/3}>B$ assuming $\delta < \frac{1}{n^3}$. On the other hand, if $\vec p(\SE)< B$ then let $\SE'$ be the optimal set of modules in matroid restricted to $\pi^0[k^*]$. We have that $\barp(\SE') > B$. Again, since $\frac{\vec v(i)}{\barp(i)} = \optbpb$ for modules $i\in \SE'$. Therefore, via a similar argument module $i\in \SE'$ with $\barp(i) + \delta$ can not be affordable by a greedy algorithm.
    
    Now, we consider that module $i\in \SE'\setminus \{\pi^0(k^*)\}$ sets price $\geq \barp(i) - \sqrt \delta$. In this case, module $i$ will be selected as it comes before rest of the modules in bang-per-buck ordering and since $\barp(\SE')> B + \sqrt \delta$, total price of $\vec p(\SE')> \barp(\SE') > B -10\cdot \delta \cdot n - \sqrt \delta > B - \sqrt \delta -10n\delta $ by assumption on discretization. This implies that the greedy algorithm only visits elements in $\pi^0[k^*]$ before it runs out of budget and module $i$ gets selected. Now, due to Lemma~\ref{lem:eqm_price_dominates}, we have that module $i$ also gets selected at price $\barp(i)$. 
    
\subsection*{Proof of Theorem~\ref{thm:matroid_convergence}}
We then leverage Lemma~\ref{lem:eqm_price_dominates} to define an event stating that after a sufficiently large number of iterations, all modules in $\SE$ start setting their price higher than their equilibrium prices.  More formally, we let $T_0$ be sufficiently large the round (determined later) and define $\mathcal E^0 = \{ \forall s\geq T_0:\vec p^t(i) > \Delta_i,  \forall i \in \SE  \}$.

We first observe that given any price vector $\vec p$ and conditioned on the event $\mathcal E^0$, the module with the worst bang-per-buck in the set $\SE'$ defined as $\SE' :=\SE \cup \{\pi^0(k^*)\}$ where $k^*$ is the last iteration of Algorithm~\ref{alg:equilibrium_dynamics_weighted}, will not be selected due to Lemma~\ref{lem:module_with_worst_bpb_rejected}.
Our overall proof approach is to show that under conditioned on the event $\mathcal E^0$, the modules in $\SE'$ literately stop posting the prices that lead to smaller bang-per-buck for the platform since it will not be selected at such a higher price. In order to demonstrate that, we define an order over the possible prices w.r.t. their bang-per-buck values for the buyer. Next, we define an order over the prices with respect to their bang-per-buck value. We iteratively define 
$$(  b_{(k)},   i_{(k)}): = \argmin_{\{(b,i)\in \bar{ \mathcal B}  \setminus \{(   b_{(1)},   i_{(1)}),\dots ,(   b_{(k-1)},   i_{(k-1)}) \} \}} \left \{\frac{\vec v(i)}{b} \right \}.$$
Here, $(   b_{(1)},   i_{(1)}): = \argmin_{(b,i) \in \bar{\mathcal B}} \{\frac{\vec v(i)}{b} \}$ and $\bar{\mathcal B}:=\bigcup_{i\in \SE'} \{(\Delta_i + \delta, i) , \dots , (1,i) \}$.  First, we prove the following claim that shows that the module $i_{(1)}$ will never price $b_{(1)}$ with high probability after $\poly(1/\delta)$ many rounds as it will never be selected.

Next, we extend this argument and iteratively show that the modules will stop posting prices larger than their respective $\Delta_i$s. We let $T_1 < T_2 <\dots <T_K$ and $T^*_1 < T^*_2 <\dots <T^*_K$ such that $T_k \leq T_{k}^*$ for some $K < \frac{n}{\delta}$ and the events $$\mathcal E^p:= \{ \forall s\geq T_{p} \text{ and } j\in M: b_{i_{(p)}} < b_{(p)}\},$$  satisfying the following condition: 
\begin{enumerate}
	\item We now inductively define $T_k$ as follows: suppose we are given $T_1,\dots , T_{k-1}$; $T^*_0,\dots , T^*_{k-1}$ and conditioned on the events $\mathcal E^1,\dots , \mathcal E^{k-1}$, we let $T_k^* > T^*_{k-1}$ be the smallest round (if exists) such that, 
	\begin{equation}
		\sum_{s\leq T^*_k} u_{i_{(k)},s}(\Delta_{i_{(k)}},\vec p^s_{-i_{(k)}}) - u_{i_{(k)},s}(b,\vec p^s_{-i_{(k)}}) \geq \delta^6 \cdot T_k^*
	\end{equation}
	\item  We let $T_k\leq T_k^*$ be the smallest stage such that there exists some $b_{i_{(k)}}' \notin\{b_{(1)},\dots .b_{(k-1)}\}$ for all $s \in (T_k,T_k^*]$, such that, 
	\begin{equation}
		\sum_{t\leq s} u_{i_{(k)},t} (b'_{i_{(k)},t},\vec p^s_{-i_{(k)}}) - u_{i_{(k)},t}(b_{i_{(k)}}, \vec p^s_{-i_{(k)}}) \geq \frac{C}{\gamma_s \cdot \delta}.
	\end{equation}
\end{enumerate}
Above,  $\mathcal E^p$ captures the event where the module $i_{(p)}$ sets their price with higher bang-per-buck than $p$-th lowest bang-per-buck in the set of bids. In addition, we can observe that by the definition of $T_p$,  we have that the module $i_{(p)}$ begins setting the price lower than $b_{i_{(p)}}$ with high probability due to existence of some lower bid $b$ which has significantly higher cumulative utility. Therefore, if we show an existence of bounded $T_i $s then we can essentially show that any module $i\in \SE'$ starts setting price lower than their corresponding $\Delta_i + 2\cdot \delta$ after some finite many rounds while event $\mathcal E^0$ ensures that module $i\in \SE'$ is bidding higher than $\Delta_i$. We first make the following observation: 
\begin{observation}\label{obs:dropped_price_with_hp_weighted}
	For any $p<K$ and $s\geq T_p$, we have, $$\Pr\left [\text{Price $b_{(p-1)}$ for module $i_{(p)}$ is not selected at round }s \mid \bigcap_{p'\leq p} \mathcal E^{p'} \right] = 1.$$
\end{observation}
\begin{proof}
	The proof of the observation follows from the fact that the price $b_{(p-1)}$ has the worst bang-per-buck at stage $s$ once conditioned on the event $\bigcap_{p'\leq p} \mathcal E^{p'}$. Therefore, Lemma~\ref{lem:module_with_worst_bpb_rejected} implies that the greedy algorithm does not select price $b_{(p-1)}$. 
\end{proof}
Above, we note that if Condition~\ref{condition1_weighted} satisfies for for some $k$ then for $T_k = T_k^*$, Condition~\ref{condition2_weighted} trivially satisfies. We next prove an upper bound on $T_i$ which shows an existence of desired $T_i$s which is one of the most crucial steps in proving Theorem~\ref{thm:matroid_convergence}. We recall our distorted payment rule: for initial rounds  $t\leq T_0$, each selected module $m_i\in M$ gets payment of $\vec p^t(i) + \delta^2\cdot \vec p^t(i)$  and the rest of the modules get payment of $\delta^2 \cdot \vec p^t(i)$ and later after round $t > T_0$, all selected module $m_i\in M$ gets payment of $\vec p^t(i) + \frac{\delta^4}{ \vec p^t(i)}$  and the rest of the modules get payment of $\frac{\delta^4}{ \vec p^t(i)}$.  

\begin{lemma}\label{lem:bounding_iterations}
	For any $k \geq 1$, conditioning on the event $\bigcap_{i=1}^{k} \mathcal E^i$, with probability $1$, we have,
	\begin{equation*}
		T_k \leq \frac{1}{\delta^6} \cdot \left( \sum_{i=0}^{k-1}\sqrt {T_i} + (1+\delta^2)\cdot T_0\right).
	\end{equation*}
\end{lemma}
\begin{proof}
	In order to prove the claim, we prove the following via induction on $k$: for any $k\geq 0$ and conditioned on the event $\bigcap_{p\leq k} \mathcal E^p$, with probability $1$, we have,
	\begin{equation*}
		\sum_{s\leq T_{k}} u_{i_{(k)},s} (b_{(k)}, \vec p^s_{-i_{(k)}} ) \leq \frac {(1+ 2\cdot \delta)} {b_{i_{(k)}}\delta^2} \cdot \sum_{i=1}^{(k-1)}  \sqrt {T_i}  +  (T_{k} - T_{k-1})\cdot \frac{\delta^4}{b_{i_{(k)}}} + T_0 \textbf{ and }T_k \leq \frac{1}{\delta^4} \cdot \sum_{i=0}^{k-1}\sqrt {T_i} + T_0.
	\end{equation*}
	The base case of induction for $k=0$ follows since $ \sum_{s\leq T_{1}} u_{i_{(1)},s} (b_{(1)}, \vec p^s_{-i_{(1)}} ) \leq T_0$ and $T_0 \leq T_0$. Suppose the claim holds for $k-1$. We first consider the following case:
	\paragraph{Case-1 ($i_{(k)} = i_{(k-1)}$) :} In this case, we can observe that, 
	\begin{align*}
		\sum_{s\leq T_{k}} u_{i_{(k)},s} (b_{(k)}, \vec p^s_{-i_{(k)}} ) &\leq \frac{C}{\gamma_{T_{k-1}} \cdot \delta} + \sum_{s\leq T_{k-1}} u_{i,s} (b_{(k-1)}, \vec p^s_{-i_{(k-1)}})\\
		&= \frac{\sqrt{T_{k-1}}}{ \delta} +  (T_{k} - T_{k-1})\cdot \frac{\delta^4}{b_{i_{(k)}}} + \sum_{s\leq T_{k-1}} u_{i,s} (b_{(k-1)}, \vec p^s_{-i_{(k-1)}})\\
		&\leq \frac{\sqrt{T_{k-1}}}{ \delta} +  (T_{k} - T_{k-1})\cdot \frac{\delta^4}{b_{i_{(k)}}} +\frac {(1+ 2\cdot \delta)} {b_{i_{(k)}}\delta^2} \cdot \sum_{i=1}^{(k-2)} \sqrt {T_i} + (T_{k-1} - T_{k})\cdot \frac{\delta^4}{b_{i_{(k)}}} + T_0 \\
		&\leq  \frac{\sqrt{T_{k-1}}}{ \delta} +  \frac{\sqrt{T_{k-1}}}{\delta^2 \cdot b_{i_{(k)}}} +  (T_{k} - T_{k-1})\cdot \frac{\delta^4}{b_{i_{(k)}}} \\
		& + \frac {(1+ 2\cdot \delta)} {b_{i_{(k)}}\delta^2}\cdot \sum_{i=1}^{(k-2)} \sqrt {T_i} + (T_{k-1} - T_{k})\cdot \frac{\delta^4}{b_{i_{(k)}}} + T_0 \\
		&\leq \frac {(1+ 2\cdot \delta)} {b_{i_{(k)}}\delta^2} \cdot \sum_{i=1}^{(k-1)}  \sqrt {T_i}  +  (T_{k} - T_{k-1})\cdot \frac{\delta^4}{b_{i_{(k)}}} + T_0.
	\end{align*}
	Above, the first inequality holds by the definition of $T_{k-1}$ and the fact that $i_{(k)} = i_{(k-1)}$, the equality holds because $u_{i_{(k-1)},s}(b_{(k)},\vec b_{-i_{(k-1)}}) = \frac{\delta^4}{b_{i_{(k)}}}$ for any $s\geq T_{k-1}$ due to Observation~\ref{obs:dropped_price_with_hp_weighted} once conditioned on the event $\bigcap_{i=1}^{k-1} \mathcal E^i$. The second and third inequality holds due to inductive hypothesis since $T_{k-1} - T_{k-2}\leq \frac{1}{\delta^4} \cdot \sqrt{T_{k-1}}$.
	
	\paragraph{Case-2 ($i_{(k)} \neq i_{(k-1)}$):} In this case, we let $\bar k$ be the largest index such that $i_{\bar k} = i_{(k)}$. If such $\bar k$ does not exist then we let $\bar k = 1$. We observe that either $b_{(\bar k)} = b_{(k)} +\delta$ or $\bar k = 1$ since all possible prices are discretized within the additive factor of $\delta$ and $b_{(1)} =1$. 
	We condition on the event $\bigcap_{i=1}^{k-1} \mathcal E^i$.  We define for any $k'\leq k$,  $\tilde b_{i_{(k')}} = b_{i_{(k')}} - \delta $.
	For any $k'\leq k$, by definition of $T_{k'}$, we have,
	\begin{align}\label{eq:iterative_expansion_eq_1}
		\sum_{t < T_{\bar k}} u_{i_{(\bar k)},t} (\tilde b_{i_{(\bar k)},t},\vec p^s_{-i_{(k-1)}}) &+ u_{i_{(\bar k)},T_{\bar k}} (\tilde b_{i_{(\bar k)},T_{\bar k}},\vec p^s_{-i_{(k-1)}}) \leq  \frac{C}{\gamma_{T_{\bar k}} \cdot \delta} + \sum_{t < T_{\bar k}} u_{i_{(\bar k)},t}(b_{i_{(\bar k)}}, \vec p^s_{-i_{(\bar k)}}) + 1 + \frac{\delta^4}{b_{i_{(\bar k)}}}  \notag \\
		%&= \frac{\sqrt{T_{\bar k}}}{\delta} + \sum_{t\leq T_{\bar k-1}} u_{i_{(\bar k)},t}(b_{i_{(\bar k)}}, \vec p^s_{-i_{(\bar k)}}) + (T_{\bar k} - T_{\bar k-1})\cdot \frac{\delta^4}{b_{i_{(\bar k)}}}  + 1 \\
		&\leq \frac{2\sqrt{T_{\bar k}}}{\delta} + \sum_{t\leq T_{\bar k}} u_{i_{(\bar k)},t}(b_{i_{(\bar k)}}, \vec p^s_{-i_{(\bar k)}})  \notag\\
		&\leq \frac{2\sqrt{T_{\bar k}}}{\delta} + \frac {(1+ 2\cdot \delta)} {b_{i_{(k)}}\delta^2} \cdot \sum_{i=1}^{(\bar k-1)}  \sqrt {T_i} + (T_{\bar k} - T_{\bar k-1})\cdot \frac{\delta^4}{b_{i_{(\bar k)}}} + T_0\notag\\
		&\leq \frac {(1+ 2\cdot \delta)} {b_{i_{(k)}}\delta^2} \cdot \sum_{i=1}^{\bar k}  \sqrt {T_i}  + T_0
	\end{align}
	Above, the first inequality holds because $T_{\bar k}$ be the smallest stage for which Condition~\ref{condition2_weighted} holds and the utility of module $i_{(\bar k)}$ at round $T_{\bar k}$ can be at most $1+ \frac{\delta^4}{b_{i_{(k)}}}$. The third inequality holds due to the inductive hypothesis. The second and last inequality holds because $T_{\bar k}>1$ and $\delta << \frac 1 2$. 
	
	 Next, we claim that under event $\bigcap_{p=1}^k \mathcal E^p$, from rounds $T_{\bar k}$ to $T_{k}$, price $\tilde b_{i_{(\bar k)}} =b_{i_{(k)}}$ can not be the winning price for module $i_{(k)}$. Now, notice that under event $\bigcap_{p=1}^k \mathcal E^p$, after round $s>T_{(\bar k)}$, module $i_{(\bar k-1)}$ does not get selected and for bang-per-buck ordering $\pi_s$ at round $s$ with $\pi^s[\ell] = i_{(\bar k-1)}$, due to Lemma~\ref{lem:module_with_worst_bpb_rejected}, we have $\vec p^s(S^*[\ell])>B$ where $S^*[\ell]$ be the maximum weight independent set in matroid $\mathcal I$ restricted to $\pi^s[\ell]$. Since, $b_{(k)}> \Delta_i + \sqrt \delta$, we have that if module $\bar k$ gets selected at price $\tilde b_{i_{(\bar k)}}$ then module $\vec p^s(S^*[\ell])\leq B$ which is a contradiction. Therefore, we have,
	 \begin{align}\label{eq:eq2}
	     \sum_{T_{\bar k}<t < T_{k}} u_{i_{(\bar k)},t} (\tilde b_{i_{(\bar k)},t},\vec p^s_{-i_{(k-1)}}) &\leq (T_{k}  - T_{\bar k})\cdot \frac{\delta^4}{b_{(k)}} \leq(T_{k} - T_{k-1}) \cdot \frac{\delta^4}{b_{(k)}}  + \frac {(1+ 2\cdot \delta)} {b_{i_{(k)}}\delta^2}\cdot  \sum_{i=\bar k+1}^{ k-1} \sqrt {T_{i}} 
	 \end{align}
	 Above, the second inequality holds due to the inductive hypothesis. Combining the Inequalities~\ref{eq:iterative_expansion_eq_1}, \ref{eq:eq2}, we obtain,
	 \begin{align*}
	     	\sum_{s\leq T_{k}} u_{i_{(k)},s} (b_{(k)}, \vec p^s_{-i_{(k)}} )&=  \sum_{T_{\bar k}<t < T_{k}} u_{i_{(\bar k)},t} (\tilde b_{i_{(\bar k)},t},\vec p^s_{-i_{(k-1)}}) +  \sum_{t < T_{\bar k}} u_{i_{(\bar k)},t} (\tilde b_{i_{(\bar k)},t},\vec p^s_{-i_{(k-1)}})\\
	     	&\leq \frac {(1+ 2\cdot \delta)} {b_{i_{(k)}}\delta^2} \cdot \sum_{i=1}^{(k-1)}  \sqrt {T_i}  +  (T_{k} - T_{k-1})\cdot \frac{\delta^4}{b_{i_{(k)}}} + T_0
	 \end{align*}
	 We then observe that module $i_{(k)}$ can be selected at price $b_{i_{(k)}}$ in less than $$\frac{1}{b_{(k)}} \cdot \left( \frac {(1+ 2\cdot \delta)} {b_{i_{(k)}}\delta^2} \cdot \sum_{i=1}^{(k-1)}  \sqrt {T_i}  + T_0\right), $$ many rounds.  Due to Lemma~\ref{lem:eqm_price_dominates}, we have that whenever the module gets selected at price $b_{(k)}$, it also gets selected at price $\Delta_i$. Therefore, for any $T_{k}^*\geq \frac{1}{\delta^6} \cdot \left( \sum_{i=0}^{k-1}\sqrt {T_i} + (1+\delta^2)\cdot T_0\right)$, we can bound the utility difference of $b_{(k)}$ and $\Delta$ as:
	 \begin{align*}
     &\sum_{s\leq T^*_k}( u_{i,s}(\Delta_i,\vec p^s_{-i}) - u_{i,s}(b_{i_{(k)}},\vec p^s_{-i}) ) - \delta^6\cdot T_{k^*}\\
     &\geq \left(\frac{(b_{(k)} - \Delta_i)\cdot \delta^4}{b_{(k)} \cdot  \Delta_i} -\delta^6 \right)\cdot T_k^* - \frac{(b_{(k)} - \Delta_i)}{b_{(k)} } \cdot \left( \frac {(1+ 2\cdot \delta)} {b_{i_{(k)}}\delta^2} \cdot \sum_{i=1}^{(k-1)}  \sqrt {T_i}  + T_0\right)\\
     &\geq   \left(\frac{(b_{(k)} - \Delta_i)\cdot \delta^4}{b_{(k)} \cdot  \Delta_i} -\delta^6 \right) \cdot \left( \frac{1}{\delta^6} \cdot \left( \sum_{i=0}^{k-1}\sqrt {T_i} + (1+\delta^2)\cdot T_0\right) \right) - \frac{(b_{(k)} - \Delta_i)}{b_{(k)} } \cdot \left( \frac {(1+ 2\cdot \delta)} {b_{i_{(k)}}\delta^2} \cdot \sum_{i=1}^{(k-1)}  \sqrt {T_i}  + T_0\right)\\
     &\geq \frac{(b_{(k)} - \Delta_i)}{b_{(k)} \cdot \Delta_i} \cdot\left(  \left( \frac{1}{\Delta_i \delta^2} - \frac {(1+ 2\cdot \delta)} {b_{i_{(k)}}\delta^2} - \frac{b_{(k)} \cdot \Delta_i}{(b_{(k)} - \Delta_i)} \right) \cdot \sum_{i=1}^{(k-1)}  \sqrt {T_i}  + \left( \frac{1}{\delta^2} - (1+ \delta^2)  \right) \cdot T_0 \right)\geq 0.\\
   \end{align*}
Above, the second inequality holds because of the assumption on $T_k^*$, the third inequality holds because the factor in multiplication with $T_0$ is positive and the final inequality holds because $\Delta_i > \vec c(i) >\sqrt \delta $ and $b_{(k)} > \Delta_i + 10 \cdot  \delta $. Since, $T_k \leq T_{k}^*$, this concludes the proof. 
\end{proof}
The above lemma immediately implies the following bound on $T_k$. 
\begin{lemma}
Conditioned on the event $\bigcap_{p=0}^k \mathcal E^p$, we have $T_k \leq \frac{k^2\cdot T_0}{\delta^{13}}$ with probability one.  
\end{lemma}
\begin{proof}
The lemma follows via induction on $k\geq 2$. We get $T_k \leq \frac{k\cdot (k-1)}{\delta^{2.5}} + \frac{(1+\delta^2)\cdot T_0}{\delta^6} \leq \frac{k^2\cdot T_0}{\delta^{13}}$.
\end{proof}
To complete the proof of Theorem~\ref{thm:matroid_convergence}, we need to show that $\Pr\left [\bigcap_{i=0}^K \mathcal E^i \right] \geq 1 - O \left( \poly\left ( \frac{1}{\delta} \right) \cdot \exp\left ( - \frac{1}{\delta} \right) \right). $ Since $K \leq \frac{n}{\delta}$, we have $T_K \leq  \poly \left(n, \poly \left( \frac 1  \delta \right) \right)$ which will complete the proof that all modules will be setting prices $\geq \Delta_i - \delta $ and $\leq \delta_i + 10\cdot \Delta_i$ with high probability after $\poly(1/\delta , n)$ rounds.

\begin{lemma}\label{lem:key_bounding_probability}
	For small enough $\delta>0$, we have $\Pr\left [\bigcap_{i=1}^K \mathcal E^i \mid \mathcal E^0\right] \geq 1 - O \left( \poly\left ( \frac{1}{\delta} \right) \cdot \exp\left ( - \frac{1}{\delta} \right) \right).$
\end{lemma}
\begin{proof}
	We show via induction that for any $t<K$, we have $$\Pr \left [ \mathcal E^k \mid \bigcap_{i=0}^{k-1} \mathcal E^i \right] \geq 1 - \frac{(k^2 +1)\cdot T_1}{\delta^6}\cdot \exp \left( -\frac{1}{\delta} \right).$$ 
	For simplicity of notations, we let $ \bar{ \mathcal E}^t : = \bigcap_{i=0}^{t-1} \mathcal E^i$. We fix $T_k^* = \frac{k^2\cdot T_0}{\delta^{13}}$.
	\begin{align*}
	%\label{eq:boudning_prob}
		\Pr\left[ \forall s\geq T_k: \bigwedge_{\ell = 1}^k   \vec p^s(i_{(\ell)}) < b_{i_{(\ell)}} \mid \bar{ \mathcal E}^k  \right] &=   \Pr\left[ \forall s\geq T_k:   \vec p^s(i_{(\ell)}) <b_{i_{(k)}} \mid \bar{ \mathcal E}^k  \right] \notag \\
		&= 1 - \Pr \left[\exists s\geq T_k:\vec p^s(i_{(\ell)}) \geq b_{i_{(1)}} \mid \bar{ \mathcal E}^k  \right] \notag \\
		&\geq 1 - \sum_{T_k\leq s\leq T_k^*}\Pr \left[\vec p^s(i_{(\ell)}) \geq b_{i_{(1)}} \mid \bar{ \mathcal E}^k  \right]\\
		&+ \sum_{s\geq  T_k^*}\Pr \left[\vec p^s(i_{(\ell)}) \geq b_{i_{(1)}} \mid \bar{ \mathcal E}^k   \right] \notag\\
		&\geq  1 - T_K \cdot \exp \left(- \frac{1}{\delta}\right) - \sum_{s\geq T_k^*} \exp \left( - \gamma_s \cdot \delta^6 \cdot s \right) \notag\\
		& =  1 - \frac{(K^2 +1)\cdot T_0}{\delta^{13}} \cdot \exp \left(- \frac{1}{\delta}\right) - \sum_{s\geq T_k^*} \exp \left( - C\cdot \delta^6 \cdot \sqrt s \right)\\
		&\geq 1 - \frac{2(k^2 +1)\cdot T_0}{\delta^{13}} \cdot \exp \left( -\frac{1}{\sqrt \delta} \right).
	\end{align*} 
	Above, the first equality holds due to conditioned on the event $\bar{ \mathcal E}^k $, second inequality holds due to union bound, the third inequality holds due to the inequality $\Pr[\vec p^s(i_{(\ell)}) < b_{(k)} \mid\bar{ \mathcal E}^k ] \geq  1 - \exp \left( - \frac 1 \delta \right)$ for all $s\geq T_k$ and $T_k^* \leq  \frac{(k^2 +1)\cdot T_1}{\delta^6}$. This completes the proof of the claim. 
\end{proof}
Finally, we need to show that the probability of event $\mathcal E^0$ is exponentially high.

\subsubsection*{Lower bounding $\Pr[\mathcal E^0]$}
To this end, we consider the following event $$\mathcal F^0:= \left\{\forall i\in \SE : \text{$\Delta_i$  gets selected for $\leq \frac{T_0}{\delta^{16}}$ from first $\frac{T_0}{\delta^{20}}\cdot (1+\delta^{18}\cdot (1-5\delta))$ rounds} \right \}.$$ 
We let $\frac{T_0}{\delta^{\ell_i}}$ be the number of times module $i$ gets selected at price $\Delta_i$. We let $T'_i= \frac{T_0\cdot \Delta_i}{\delta^{\ell_i +4}}\cdot (1+\delta^{\ell_i +2}\cdot (1-5\delta))$. We note that $\ell_i$ is a random variable.  

Note that once we conditioned on $\mathcal F^0$, we set $T_i = T_i^* = \frac{T_0\cdot \Delta_i}{\delta^{\ell_i +4}}\cdot (1+\delta^{\ell_i +2}\cdot (1-5\delta))$ for all $i\in [K]$. We observe that for any module $i\in \SE$ and price $p> \Delta_i + 10 \cdot \delta$ and $\delta_i < 1 - 10\cdot \delta$, with probability $1$, we have, 
	 \begin{align*}
	 %\label{eq:upper_bound_utility_large_bid}
     &\sum_{s\leq T^*_k}( u_{i,s}(\Delta_i,\vec p^s_{-i}) - u_{i,s}(p,\vec p^s_{-i}) )  \geq \frac{(p- \Delta_i)\cdot \delta^4}{p \cdot  \Delta_i}\cdot (T_k^* -T_0)   - (p - \Delta_i) \cdot \left( \frac{T_0}{\delta^{\ell_i}}\right) - (p - \Delta_i)\cdot  \delta^2 \cdot T_0\notag\\
     &\geq \frac{(p- \Delta_i)\cdot \delta^4}{p \cdot  \Delta_i}\cdot \left(\frac{T_0\cdot \Delta_i}{\delta^{\ell_i +4}}\cdot (1+\delta^{\ell_i+2}\cdot (1-5\delta)) -T_0 \right)  - (p - \Delta_i) \cdot \left( \frac{T_0}{\delta^{\ell_i}}\right) - (p - \Delta_i)\cdot  \delta^2 \cdot T_0\notag \\
     &\geq  \left ( \delta^2 \left (\frac{1- 5\delta }{p\cdot \Delta_i} \right)  - 1 \right)\cdot T_0\geq \delta^3 \cdot T_0
   \end{align*}
   Above, the first inequality follows because $\Delta_i$ gets selected at most $\leq \frac{T_0}{\delta^{\ell_i}}$ times and whenever the price by module $i$ being $\Delta_i$ is not selected then price  $p$ can not be selected as well (Lemma~\ref{lem:eqm_price_dominates}). The second inequality follows by the definition of $T_k^*$. The last inequality follows because $p> \Delta_i + 10 \cdot \delta$ and $\Delta_i < 1 - 10\cdot \delta$.
   
   In addition, for any bid $b< \Delta_i$, again due to Lemma~\ref{lem:eqm_price_dominates}, whenever the price by module $i$ being $\Delta_i$ is not selected then price  $p$ can not be selected as well. Therefore, we have,
	 \begin{align}
     \sum_{s\leq T^*_k}( u_{i,s}(\Delta_i,\vec p^s_{-i}) - u_{i,s}(p,\vec p^s_{-i}) ) &\geq (\Delta_i - p)\cdot  \delta^2 \cdot T_0 + \frac{(\Delta_i - p)\cdot T_0}{\delta^{\ell_i}} -\frac{( \Delta_i - p)\cdot \delta^4}{p \cdot  \Delta_i}\cdot (T_k^* -T_0) \notag \\
     &\geq (\Delta_i - p)\cdot  \delta^2 \cdot T_0 + \frac{(\Delta_i - p)\cdot T_0}{\delta^{\ell_i}} -\frac{( \Delta_i - p)\cdot \delta^4}{p \cdot  \Delta_i}\cdot \frac{T_0\cdot \Delta_i}{\delta^{\ell_i +4}} \notag  \\
     & = \frac{(p - \Delta_i)}{ \Delta_i}\cdot \left( \delta^8\cdot (1 - 5\delta ) \right) \cdot T_0 \geq \delta^9 \cdot T_0.\label{eq:upper_bound_utility_small_bid}
   \end{align}
   
   This leads to the following lemma.
   
   \begin{lemma}
   For $T_0\geq \frac{n^2}{\delta^{22}}$ and conditioned on the event $\mathcal F^0$, for all the rounds from $T': = \frac{T_0}{\delta^6}(1 + \delta^4 (1-5\delta))$ to  $\frac{T_0}{\delta^6}(1 + \delta^4 (1-5\delta)) + T_0$, all modules in $i\in \SE'$ sets prices between $[\Delta_i - 10\cdot \delta , \Delta_i + 10 \cdot \delta]$ with probability $1 - \poly(n,1/\delta) \cdot \exp(- \frac 1 {\sqrt \delta})$.  
   \end{lemma}
   \begin{proof}
       For any $s\geq T'$, we let $\mathcal H^s(i)$ be the event where module $i$ submits price from the set $[\barp(i) , \barp(i) + \sqrt \delta]$. Due to Lemma~\ref{lem:stability_of_eqm_price},for any $s>T'$, conditioning on event $\bigcap_{i}\bigcup_{T' \leq t\leq s} \mathcal H^t(i)$, we have for any $i\in \SE'$, $\barp(i)$ is accepted at any round $T' \leq t\leq s$ and $\barp(i) \sqrt \delta$ is not accepted at round $t$. This further implies that $p\notin [\Delta_i , \barp(i)+ \sqrt{\delta}]$, we have 
       \begin{equation*}
           \sum_{s\leq t}( u_{i,s}(\barp(i),\vec p^s_{-i}) - u_{i,s}(p,\vec p^s_{-i}) ) \geq (t - T')\cdot \delta + T_0\cdot \delta^9
       \end{equation*}
       This further implies that conditioned on the event $\bigcap_{i}\bigcup_{T' \leq t\leq s} \mathcal H^t(i)$ and property of multiplicative weight update algorithm, we have that
       \begin{align*}
           \Pr\left [\bigcap_i \mathcal H^t(i)\mid \bigcap_i\bigcap_{T'<s<t} \mathcal H^s(i)\right ]&\geq 1 - 2\cdot \frac{n}{\delta}\cdot \exp \left( - \frac{C\cdot (\delta^9 \cdot T_0 + \delta (t- T')}{\sqrt t} \right)\\
           & \geq  1 - 2\cdot \frac{n}{\delta}\cdot \exp \left( - \frac{C\cdot (\frac{n^2}{\delta^{13}} + \delta (t- T')}{\sqrt {t- T' + \frac{n^2}{\delta^{22}}}} \right)\\
           &\geq 1 - 2\cdot \frac{n}{\delta}\cdot \exp \left( - C\cdot \left (\frac{n^2}{\delta} + \delta \sqrt{(t- T')}\right) \right)
       \end{align*}
       Next, via conditional expectations, and assuming $\delta$ small enough, we conclude that,
       \begin{equation*}
          \Pr\left [\bigcap_i\bigcap_{s>T'} \mathcal H^s(i)\right ] \geq 1 - \poly(1/\delta , n)\cdot \exp(- 1/\sqrt \delta). 
       \end{equation*}
   \end{proof}
   On the other hand, when $\mathcal F$, does not hold, then all modules win a sufficiently large number of times that leads to for any $i\in \SE'$ and $p<\bar p - 10\cdot \delta$ we have,
   \begin{equation*}
       \sum_{s\leq T^*_K}( u_{i,s}(\Delta_i,\vec p^s_{-i}) - u_{i,s}(p,\vec p^s_{-i}) ) \geq \delta^9 \cdot T_0.
   \end{equation*}
   Since, $\mathcal E^0 \subseteq \bigcap_i\bigcap_{s>T'} \mathcal H^s(i)$, we conclude that $\Pr\left [\mathcal E^0 \right]>1 - \poly(1/\delta , n) \cdot \exp \left( - \frac{1}{\sqrt \delta} \right)$.

\bibliographystyle{plainnat}
\bibliography{ref}

\end{document}